\documentclass[a4paper,11pt,onecolumn,accepted=2026-03-27]{quantumarticle}
\pdfoutput=1
\usepackage[utf8]{inputenc}
\usepackage[english]{babel}
\usepackage[T1]{fontenc}
\usepackage{amsmath,amssymb,amsthm,cmll}
\usepackage{hyperref}
\usepackage[numbers,sort&compress]{natbib}

\usepackage{mathrsfs,mathtools}

\usepackage[all]{xy}

\def\Indep{\biguplus}
\def\indep{\uplus}
\def\cover{\ll}
\def\aff{\operatorname{Aff}}
\def\lin{\operatorname{Lin}}
\def\Span{\operatorname{Span}}
\def\Ce{\mathcal C}
\def\Te{\mathcal T}
\def\Fe{\mathcal F}
\def\Pe{\mathcal P}
\def \Tr{\mathrm{Tr}\,}
\def\Se {\mathcal S}
\def\permut{\mathscr{S}}
\def\<{\langle\,}
\def\>{\,\rangle}
\def\vtl{\vartriangleleft}
\def \Af{\mathrm{Af}}
\def \FV{\mathrm{FinVect}}

\newtheorem{theorem}{Theorem}[section]
\newtheorem{lemma}[theorem]{Lemma}
\newtheorem{coro}[theorem]{Corollary}
\newtheorem{prop}[theorem]{Proposition}

\theoremstyle{remark}
\newtheorem{remark}{Remark}[section]
\newtheorem{exm}{Example}[section]

\begin{document}

\title{On the structure of  higher order quantum maps}

\author{Anna Jenčová}
\affiliation{Mathematical Institute, Slovak Academy of Sciences, Štefánikova 49, 814 73
Bratislava, Slovakia}
\orcid{ }
\email{jenca@mat.savba.sk}
\maketitle

\begin{abstract}
  We study higher order quantum maps in the context of a *-autonomous category of affine
subspaces. We show that types of higher order maps can be identified with certain Boolean
functions that we call type functions. By an extension of this identification, the algebraic structure of
Boolean functions is inherited by some sets of quantum objects including higher order maps. 
Using the M\"obius transform, we assign to each
type function a poset whose elements are labelled by subsets of indices of the involved 
spaces. We then show that the type function corresponds to a comb type if and only if the
poset is a chain. We also devise a procedure for decomposition of the poset to a set of
basic chains from which the type function is constructed by taking maxima and minima of
concatenations of the basic chains in different orders. On the level of higher order maps,
 maxima and minima correspond to affine mixtures and intersections, respectively.

  \end{abstract}

\section{Introduction}

Higher order quantum maps encompass all quantum objects and admissible transformations
between them. Starting from the set of quantum states, the whole hierarchy of maps can be built,
including at each step all transformations between objects built at the previous steps. 
The basic notion here is that of a quantum channel, defined as a completely positive
trace preserving map. Quantum channels represent  physical transformations between quantum
state spaces, but all the higher order quantum maps can be described as
quantum channels with a special structure, see
\cite[Prop.~2]{apadula2024nosignalling}.

Transformations between quantum channels are given by  quantum
supermaps, consisting of a pair of channels connected by an ancilla
\cite{chiribella2008transforming}. We may then consider transformations from superchannels
to channels and so on, building a hierarchy of quantum
combs \cite{chiribella2008quantum,chiribella2009theoretical}, also called quantum
games
\cite{gutoski2007toward},  which are given by quantum circuits with some free slots.  
The free input and output spaces of such circuits have a fixed ordering, and therefore have a
definite causal structure. Causally ordered transformations include quantum testers
\cite{chiribella2008memory} or PPOVMs \cite{ziman2008process}, describing tests and
measurements on quantum channels and combs, and other types of maps, e.g.  no-signalling channels.

Besides the causally ordered maps, the hierarchy of higher order maps  contains
transformations that have {indefinite causal
structure}, 
\cite{oreshkov2012quantum,chiribella2013quantum, perinotti2017causal}. Such maps include, for example,  transformations between
quantum combs or transformations 
 of no-signalling channels, which can be obtained  as affine mixtures of quantum combs
 having different causal orders \cite{oreshkov2012quantum,perinotti2017causal}. An example  is a quantum switch
\cite{chiribella2013quantum}, which transforms a pair of channels $\Phi$ and $\Psi$ into a coherent
combination of their compositions $\Phi\circ\Psi$ and $\Psi\circ \Phi$. 

In recent years, several theoretical frameworks for description of higher order maps were
proposed, see also the recent review \cite{taranto2025higher}. In \cite{perinotti2017causal,bisio2019theoretical}, a formalism of {types} is introduced,
built from elementary types (corresponding to state spaces) and transformations. A
categorical approach is taken in
\cite{kissinger2019acategorical,simmons2022higher,simmons2024acomplete}, where a *-autonomous
category $\mathrm{Caus}[\mathcal C]$ is constructed over a compact closed symmetric monoidal category satisfying
certain properties. The work \cite{hoffreumon2026projective} builds on these results 
using the description of higher order objects by superoperator projections, which extends a similar
description of {process matrices} in \cite{araujo2015witnessing}. In a similar approach, the projections
corresponding to transformations between given quantum objects were characterized in
\cite{milz2024characterising}.

 The formalism of types provides a very general framework for
description of higher order maps. The inductive construction does not  refer to any 
specific properties of quantum systems and can therefore be transferred  to  more general 
contexts.  The connection of types to *-autonomous categories enables us to recognize
equivalence of certain types and to apply  the underlying algebraic (Boolean algebras,
\cite{hoffreumon2026projective}) and logical (BV-logics, \cite{simmons2022higher})
structures to study causal relations inherent in higher order  maps of a given type. 
The combinatorial description of types proved useful in the study of signalling relations
and the restrictions they pose on possible compositions of the maps,
\cite{apadula2024nosignalling}. On the other hand,
the axiomatic approach of type theory in itself seems  insufficient for the study of
 constructions of higher order maps (such as QC-CC or QC-QC maps \cite{wechs2021quantum},
 routed quantum circuits \cite{vanrietvelde2025consistent}) or 
 their possible physical implementations.

The aim of the present work is to gain more understanding about the structure of types of
higher order quantum maps, in an approach that combines the works
\cite{kissinger2019acategorical,simmons2022higher},
\cite{bisio2019theoretical,apadula2024nosignalling},
\cite{hoffreumon2026projective} and \cite{milz2024characterising}. 
A crucial tool in these works is the representation of the transformations  via the Choi
isomorphism, which allows to identify transformations of a given type  with elements of 
the intersection of a certain subspace in the space of multipartite hermitian operators
with the cone of positive operators and a hyperplane determined by a
positive multiple of the identity. 
Since the cone of positive
  operators remains the same whichever type of higher order maps we are dealing with and
  the hyperplanes differ only by a positive multiple, it follows that types of 
  higher order maps can be represented in terms of subspaces of multipartite hermitian
  operators, or equivalently, by projections onto them.  This is the idea behind the
  projective characterization of process matrices in \cite{araujo2015witnessing} and
  higher order quantum maps in \cite{hoffreumon2026projective, milz2024characterising}.

In the first part  (Section \ref{sec:affine}), we combine the above idea with the categorical
approach of \cite{kissinger2019acategorical,simmons2022higher}. We investigate the
category $\Af$, whose objects are finite dimensional vector spaces with a distinguished
affine subspace. We prove that with the structures inherited from the
 compact closed structure of $\FV$, $\Af$ becomes a *-autonomous category. The internal
 hom  $X\multimap Y=(X\otimes Y^*)^*$ in $\Af$ directly corresponds to the set of all linear
 transformations mapping elements of  the affine
subspace related to $X$ into the affine subspace related to $Y$. 
Restricting to spaces of
hermitian matrices, we show that the intersection with the positive cone gives exactly the
Choi matrices of completely positive maps preserving given affine subspaces, in
particular, the higher order quantum maps can be characterized by analysing the internal 
homs in $\Af$. Using the special structure of hermitian matrices, we can also obtain a 
projective characterization as in  \cite{milz2024characterising}. 

We define first order objects in
$\Af$ as those where the distinguished affine subspace is a hyperplane. The higher order
objects are then defined by applying the internal homs, or equivalently, tensor products
and duality, on a finite set $\{X_1,\dots,X_n\}$ of first order objects. Along the lines of
\cite{bisio2019theoretical}, we then find an independent decomposition of the tensor
product $V=V_{X_1}\otimes \dots\otimes V_{X_n}$ of the underlying vector spaces into subspaces  labelled by binary strings. We use this
decomposition to define an injective map from a certain subalgebra  $\Fe_n$ of Boolean
functions into objects of $\Af$. This map depends on the choice of the
set  $\{X_1,\dots,X_n\}$. Its range consists of objects that have the tensor product
$V$ as the underlying subspace and inherits the boolean lattice
structure of $\Fe_n$, with $\wedge$ and $\vee$ given by
corresponding operations on the affine subspaces. 
In categorical terms, $\wedge$ and $\vee$ are given by pullbacks and pushouts,
respectively, moreover, the above map transforms tensor products of functions into tensor
products of objects. Any higher order
object built over $n$ first order objects can be found in the range of such a map, the
corresponding function $f\in \Fe_n$ being uniquely determined by its type. These elements
in $\Fe_n$ are called the {type functions}.
Let us note that in the case of quantum objects, the
 structure inherited from $\Fe_n$ corresponds to the algebraic structure of the
 superoperator projections studied in \cite{hoffreumon2026projective}. 
 
The second part of the paper (Section \ref{sec:type}) is devoted to investigation  of the structure of type
functions.  We first notice that there are special type functions $p_S$, $S\subseteq [n]$, whose conjugates are mapped to objects
representing channels. We then use the M\"obius transform to express any function in
$\Fe_n$ as a combination of these functions. This leads to definition of a subposet
$\Pe_f$ in $2^n$ of subsets $S\subseteq [n]$ such that $p_S$ has a nonzero coefficient. We show
that if $f$ is a type function, then $\Pe_f$ is a graded poset with even rank and $f$ is
fully determined by $\Pe_f$. It is also shown that $f$ is a type function corresponding to
combs if and only if  $\Pe_f$ is a chain (of even length), and the subsets in the chain determine the input and
output spaces of the combs. We then introduce a {causal product} on $\Fe_n$, which in
case of chain types  corresponds to connections of the chains. This operation can be
seen as a version of the 'prec' connector of \cite{hoffreumon2026projective} and the
sequential product $<$  of \cite{simmons2022higher, simmons2024acomplete}.  The  structure theorem
of type functions then proves that any type function is obtained from
several basic chain types by taking maxima and minima (or minima and maxima) over their
concatenations  in different orders. This can be seen as a version 
of the normal form of \cite{hoffreumon2026projective}. 

In Section \ref{sec:pf0} we then show that instead of the poset $\Pe_f$, we may use a
smaller poset $\Pe_f^0$, whose vertices are labelled by subsets of $[n]$. We
show that any type function is fully determined  by $\Pe_f^0$, by presenting a procedure
for  decomposing the poset into a set of chains, from which the type function is
constructed via the structure theorem.

 Let us stress that the introduced framework is only suitable for the study
of features of maps that can be expressed by affine constraints. There are many important
properties that are not affine, such as causal nonseparability, LOSR (local operations and
shared randomness) or LOSE (local operations and shared entanglement) bipartite channels,
etc. For these one should include other restrictions e.g. on the positive cone, but such
considerations are outside the scope of this work.

Let us also remark that to obtain our results for higher order quantum maps, we could have
restricted our investigation to the *-autonomous subcategory of quantum objects (see Section
\ref{sec:affine_q}), which would make our constructions somewhat easier. Specifically, in
the general case there is no natural  identification of a linear space $V$ with its dual, and we also have to
specify some fixed element in the hyperplane for all involved first order objects. It is not yet clear
whether such a more general approach can be really useful, on the other hand,
the complications
involved are not that difficult to tackle. The more general category of affine subspaces was used with
the idea that similar considerations might be applicable for investigation of higher order
maps in the framework of general probabilistic theories, see e.g.
\cite{plavala2023general} for an overview of GPTs  and \cite{bavaresco2024indefinite}
for some suggestions for higher order constructions in this framework.

\section{Affine subspaces and higher order maps}\label{sec:affine}

\subsection{The category $\FV$} \label{sec:fv}

The category $\FV$ is a basic  example of a compact closed category, whose structure
underlies all the results in this paper. For a standard reference on category theory and symmetric
monoidal categories, see \cite{maclane2013categories}. For more on compact closed categories, see
\cite{ponto2014traces}. For the use of categories in quantum theory, see
\cite{heunen2019categories}.

Let  $\FV$ be the category of finite dimensional real vector spaces with linear maps. 
We will denote the usual tensor product by $\otimes$, then  $(\FV,\otimes, I=\mathbb R)$
is a symmetric monoidal category, with the associators, unitors and symmetries given by
the obvious isomorphisms 
\begin{align*}
\alpha_{U,V,W}&:(U\otimes V)\otimes W\simeq U\otimes (V\otimes W), \\
\lambda_V&: I\otimes
V\simeq
V, \qquad \rho_V: V\otimes I\simeq V,\\
\sigma_{U,V}&: U\otimes V\simeq V\otimes U.
\end{align*}
For any permutation $\pi\in \permut_n$, we will denote the isomorphism $V_1\otimes
\dots\otimes V_n\to V_{\pi^{-1}(1)}\otimes \dots\otimes V_{\pi^{-1}(n)}$ again by $\pi$.

Let  $(-)^*: V\mapsto V^*$ be the usual vector space dual, with duality denoted by
$\<\cdot,\cdot\>: V^*\times V\to \mathbb R$. We will use the canonical identification
$V^{**}=V$ and $(V_1\otimes V_2)^*=V_1^*\otimes V_2^*$. With this duality, $\FV$ is
compact closed. This means that for each object $V$, there are morphisms $\eta_V: I\to V^*\otimes
V$ (the 'cup') and $\epsilon_V: V\otimes V^*\to I$ (the 'cap') such that the following
'snake'
identities hold:
\begin{equation}\label{eq:snake}
(\epsilon_V\otimes V)\circ (V\otimes \eta_V)=V,\qquad (V^*\otimes \epsilon_V)\circ
(\eta_V\otimes V^*)=V^*,
\end{equation}
here we denote the identity map on the object $V$ by $V$. We will see in
Example \ref{exm:quantum} below that in the case related to quantum theory, that is when
$V$ is the space of complex hermitian
matrices, the two morphisms are given by (a multiple of) the maximally entangled state.
In fact, as
it is now well established, the identities in \eqref{eq:snake} are exactly what lies
behind fundamental protocols such as quantum teleportation
\cite{heunen2019categories}. These identities also enable us to identify linear maps with
their Choi matrices, which is an important tool in the higher order quantum theory.

Let us identify the 'cup' and 'cap'  morphisms for an object $V$ in $\FV$.
First, $\eta_V$ is a linear map $\mathbb R\to V^*\otimes V$, which can be
identified with the element $\eta_V(1)\in V^*\otimes V$ and   $\epsilon_V\in (V\otimes
V^*)^*=V^*\otimes V$ is again an element of the same space.  Choose a basis
$\{e_i\}$ of $V$, let $\{e_i^*\}$ be the dual basis of $V^*$, that is,
$\<e_i^*,e_j\>=\delta_{i,j}$. Let us define
\[
\eta_V(1)=\epsilon_V:=\sum_i e_i^*\otimes e_i.
\]
It is easy to see that this definition does not depend on the choice of the basis, indeed,
$\epsilon_V$ is the linear functional on $V\otimes V^*$ defined by
\[
\<\epsilon_V, x\otimes x^*\>=\<x^*,x\>,\qquad x\in V, \ x^*\in V^*.
\]
It is also easily checked that the snake identities \eqref{eq:snake} hold.

For two objects $V$ and $W$ in $\FV$, we will denote the set of all morphisms (i.e. linear
maps) $V\to
W$ by $\FV(V,W)$. An object $V\multimap W$ is called an internal hom if for every object
$U$ there exists an
isomorphism (called currying) 
\[
\FV(U\otimes V,W)\simeq \FV(U,V\multimap W)
\]
natural in $U$, $V$ and $W$. In a compact closed category an internal hom always exists
and can be given as $V\multimap W=V^*\otimes W$.

The set  $\FV(V,W)$  itself has a structure of a real vector
space and it is well known that we have the identification 
$\FV(V,W)\simeq V^*\otimes W=V\multimap W$. This can be given as follows: for each $f\in \FV(V,W)$, we have 
$C_f:=(V^*\otimes f)(\epsilon_V)=\sum_i e_i^*\otimes f(e_i)\in V^*\otimes W$. Conversely,
since $\{e_i^*\}$ is a basis of $V^*$, 
any element $w\in V^*\otimes W$ can be uniquely written as $w=\sum_i e_i^*\otimes w_i$ for
$w_i\in W$, and since $\{e_i\}$ is a basis of $V$, the assignment $f(e_i):=w_i$ determines a
unique map $f:V\to W$. The relations between $f\in \FV(V,W)$ and $C_f\in V^*\otimes W$ can
be also written as
\[
\<C_f,x\otimes y^*\>=\<\epsilon_V,x\otimes f^*(y^*)\>=\<f^*(y^*),x\>=\<y^*,f(x)\>,\qquad x\in
V,\ y^*\in W^*,
\]
here $f^*:W^*\to V^*$ is the adjoint of $f$.

The following two examples  are most important for us.

\begin{exm}\label{exm:classical} Let $V=\mathbb R^N$. In this case, we fix the canonical basis $\{|i\>,\
i=1,\dots,N\}$. We will identify $(\mathbb R^N)^*=\mathbb R^N$, with duality
$\<x,y\>=\sum_i x_iy_i$, in particular, we identify $I=I^*$. We then have
$\epsilon_V=\sum_i |i\>\otimes |i\>$ and if $f:\mathbb R^N\to \mathbb R^M$ is given by the
matrix $A$ in the two canonical bases, then  $C_f=\sum_i |i\>\otimes A|i\>$ is the
vectorization of $A$.

\end{exm}

\begin{exm}\label{exm:quantum} Let $V=M_n^h$ be the space of $n\times n$ complex hermitian matrices. We again
identify $(M_n^h)^*=M_n^h$, with duality $\<A,B\>=\Tr AB$. Let us choose the basis in $M_n^h$, given as
\[
\left\{|j\>\<k|+|k\>\<j|,\ j\le k,\ i\biggl(|j\>\<k|-|k\>\<j|\biggr),\ j<k\right\}.
\]
Then one can check that
\[
\left\{\frac12\biggl(|j\>\<k|+|k\>\<j|\biggl),\ j\le k,\
\frac{i}{2}\biggl(|k\>\<j|-|j\>\<k|\biggr),\ j<k\right\}
\]
is the dual basis and we have
\[
\epsilon_V=\sum_{j,k} |j\>\<k|\otimes |j\>\<k|,
\]
note that up to normalization, this is a maximally entangled state in
$\mathbb C^n\otimes \mathbb C^n$.
For any $f:M_n^h\to M_m^h$, 
\[
C_f=\sum_{j,k} |j\>\<k|\otimes f(|j\>\<k|)
\]
is the Choi matrix of $f$.

\end{exm}

 In the rest of this work, we will use the following notations. 
For subspaces in a vector space, we will use $\vee$ and $\wedge$ for the usual
lattice operations on subspaces, that is, for $V_1,V_2\subseteq V$, we put $V_1\wedge
V_2=V_1\cap V_2$ and $V_1\vee V_2=\mathrm{Span}(V_1\cup V_2)$. The direct sum will be
denoted by $+$, that is, $V_1+V_2=V_1\vee V_2$ with $V_1\wedge V_2=\{0\}$. For $V_1,\dots,
V_n\subseteq V$, we denote the direct sum by $\sum$, that is, $\sum_i V_i=\bigvee V_i$
with $V_i\wedge V_j=\{0\}$.  The notation
$\oplus$ is reserved for the orthogonal direct sum of subspaces in $M_n^h$ with inner
product $\<\cdot,\cdot\>$.

\subsection{The category $\Af$}

We now introduce the category $\Af$, whose objects  are of the form $X=(V_X,A_X)$, where
$V_X$ is an object in $\FV$  and $A_X\subseteq V_X$ is a proper affine subspace, see Appendix
\ref{sec:app_affine} for definitions and basic properties. Morphisms $X\xrightarrow{f} Y$ in $\Af$ are linear maps $f:V_X\to V_Y$  such that
$f(A_X)\subseteq A_Y$. For two objects $X$ and $Y$ and a  linear map $f:V_X\to V_Y$, we
write $X\xrightarrow{f} Y$   with
the meaning that $f$ is a morphism in $\Af$. In particular, if $V_X=V_Y=V$, then  $X\xrightarrow{id_V} Y$ means
that  $A_X\subseteq A_Y$. The set of all morphisms $X\xrightarrow{f} Y$ in $\Af$ will be
denoted by $\Af(X,Y)$.

For any object $X$, we put
\begin{align*}
L_X:=\lin(A_X), \quad  S_X:=\Span(A_X), \quad D_X=\dim(V_X),\quad d_X=\dim(L_X).
\end{align*}
By \eqref{eq:affine_l} and \eqref{eq:affine_s},  we have
\begin{equation}\label{eq:ALS}
A_X=a+L_X=S_X\cap \{\tilde a\}^*,
\end{equation}
for any choice of elements $a\in A_X$ and $\tilde a\in A^*_X$.

We will next introduce a monoidal structure $\otimes$ as follows. For two objects $X$ and
$Y$, we put  $V_{X\otimes Y}=V_X\otimes V_Y$ and construct the affine subspace
$A_{X\otimes Y}$ as the affine span of 
\[
A_X\otimes A_Y=\{a\otimes b,\ a\in A_X,\ b\in A_Y\}.
\]
Fix any $a_X\in A_X$, $\tilde a_X\in A^*_X$ and $a_Y\in A_Y$, $\tilde
a_Y\in A^*_Y$. 
Since $a_X\otimes a_Y\in A_X\otimes A_Y\subseteq \{\tilde a_X\otimes \tilde a_Y\}^*$, the affine span of
$A_X\otimes A_Y$ is a proper affine subspace and we have by Lemma \ref{lemma:dual}
\begin{equation}\label{eq:aff_tensor}
A_{X\otimes Y}:=\aff(A_X\otimes A_Y)=\{A_X\otimes A_Y\}^{**}.
\end{equation}

\begin{lemma}\label{lemma:tensor_spaces}
For any $a_X\in A_X$, $a_Y\in A_Y$, we  have
\begin{align}
L_{X\otimes Y}&=\lin(A_X\otimes A_Y)=\Span(\{x\otimes y-a_X\otimes a_Y,\ x\in A_X,\ y\in
A_Y\})\label{eq:lxy1}\\
&= (a_X\otimes L_Y)+ (L_X\otimes a_Y)+ (L_X\otimes L_Y)\label{eq:lxy}
\end{align}
(here $+$ denotes the direct sum of subspaces). We also have
\[
S_{X\otimes Y}=S_X\otimes S_Y.
\]
\end{lemma}

\begin{proof} The equality \eqref{eq:lxy1} follows from Lemma \ref{lemma:dual}. For any $x\in A_X$, $y\in A_Y$
 we have
\[
x\otimes y-a_X\otimes a_Y=a_X\otimes (y-a_Y)+(x-a_X)\otimes a_Y+(x-a_X)\otimes (y-a_Y),
\]
so that $L_{X\otimes Y}=\lin(A_X\otimes A_Y)$ is contained in the subspace on the RHS of \eqref{eq:lxy}.
Let $d$ be the dimension of this subspace, then clearly
\[
d_{X\otimes Y}\le d\le d_X+d_Y+d_Xd_Y.
\]
On the other hand, any element of $S_X$ has the form $tx$ for some $t\in \mathbb R$ and
$x\in A_X$, so that it is easily seen that $S_X\otimes S_Y=S_{X\otimes Y}$. 
Hence 
\begin{align*}
d_{X\otimes Y}&=\dim(L_{X\otimes Y})=\dim(S_{X\otimes
Y})-1=\dim(S_X)\dim(S_Y)-1=(d_X+1)(d_Y+1)-1\\
&=d_X+d_Y+d_Xd_Y.
\end{align*}
This completes the proof.

\end{proof}

\begin{lemma}\label{lemma:monoidal} Let $I=(\mathbb R,\{1\})$. Then 
$(\Af,\otimes, I)$ is a symmetric monoidal category.
\end{lemma}

\begin{proof} Note that this structure is inherited from the symmetric monoidal structure
in $\FV$. To show that $\otimes$ is a functor, we have to check that for $X_1\xrightarrow{f} Y_1$ and $X_2\xrightarrow{g} Y_2$ in
$\Af$, we have $X_1\otimes Y_1\xrightarrow{f\otimes g} X_2\otimes Y_2$ which amounts to
showing that 
\[
(f\otimes g)(A_{X_1\otimes Y_1})\subseteq A_{X_2\otimes Y_2}.
\]
Let $x\in A_{X_1}$, $y\in A_{Y_1}$, then $f(x)\otimes g(y)\in A_{X_2}\otimes
A_{Y_2}\subseteq A_{X_2\otimes Y_2}$. Since  $A_{X_1\otimes Y_1}$ is the affine subspace
generated by $A_{X_1}\otimes A_{Y_1}$, the above inclusion follows by linearity of $f\otimes
g$. 

It only remains to prove that the associators, unitors and symmetries from
$\FV$ are morphisms in $\Af$. We will prove this for the associators $\alpha_{X,Y,Z}:V_X\otimes (V_Y\otimes V_Z)\to
(V_X\otimes V_Y)\otimes V_Z$, the other proofs are similar. We need to check that
$\alpha_{X,Y,Z}(A_{X\otimes(Y\otimes Z)})\subseteq A_{(X\otimes Y)\otimes Z}$. It is easily
checked that $A_{X\otimes(Y\otimes Z)}$ is the affine span of elements of the form
$x\otimes (y\otimes z)$, $x\in A_X$, $y\in A_Y$ and $z\in A_Z$, and we have
\[
\alpha_{X,Y,Z}(x\otimes (y\otimes z))=(x\otimes y)\otimes z\in A_{(X\otimes Y)\otimes Z}
\]
for all such elements. The desired inclusion follows by linearity.

\end{proof}

By Corollary \ref{coro:dual},  the dual affine subspace $A^*_X$ is a proper affine subspace in $V_X^*$, so that
$X^*:=(V_X^*,A^*_X)$ is an object in $\Af$. We have $X^{**}=X$ and the corresponding subspaces
are related as
\begin{equation}\label{eq:duality}
L_{X^*}=S_X^\perp,\qquad S_{X^*}=L_X^\perp.
\end{equation}
It is easily seen that for any  $X\xrightarrow{f} Y$, the adjoint map satisfies $f^*(
A_Y^*)\subseteq A^*_X$, so that $Y^*\xrightarrow{f^*} X^*$ and the duality $(-)^*$ is a
full and faithful functor 
$\Af^{op}\to \Af$. As we will see below, $\Af$ with this monoidal structure and
duality is not compact closed. Nevertheless, we next show that it is *-autonomous, which
is a weaker property meaning that $(\Af,\otimes, I)$ is a symmetric monoidal category
and the duality $(-)^*$ is a full and faithful contravariant functor such that the
internal hom is given by $X\multimap Y=(X\otimes Y^*)^*$, see \cite{barr1979star} for details.
This property will be crucial for the structure of higher order objects studied further.

\begin{theorem} $(\Af,\otimes,I)$ is a *-autonomous category, with duality $(-)^*$, such
that $I^*=I$.

\end{theorem}

\begin{proof} By Lemma \ref{lemma:monoidal}, we have that $(\Af,\otimes, I)$ is a symmetric
monoidal category. We have also seen that the duality $(-)^*$ is a full and faithful
contravariant functor. We only need to check the natural isomorphisms 
\[
\Af(X\otimes Y,Z^*)\simeq \Af(X,(Y\otimes Z)^*).
\]
Since $\FV$ is compact, we have the natural isomorphisms
\[
\FV(V_X\otimes V_Y,V^*_Z)\simeq \FV(V_X,V_Y^*\otimes V_Z^*),
\]
determined by the equalities
\[
\<f(x\otimes y),z\>=\<\hat f(x),y\otimes z\>,\qquad x\in V_X,\ y\in V_Y,\ z\in V_Z,
\]
for $f\in \FV(V_X\otimes V_Y,V_Z^*)$ and the corresponding morphism  $\hat f\in \FV(V_Z,V_Y^*\otimes V_Z^*)$. Since
$A_{X\otimes Y}$ is an affine span of $A_X\otimes A_Y$, we see that
$f$ is in $\Af(X\otimes Y, Z^*)$ if and only if $f(x\otimes y)\in A^*_Z$ for all $x\in A_X$, $y\in
A_Y$, that is, 
\[
1=\<f(x\otimes y),z\>=\<\hat f(x),y\otimes z\>\qquad \forall x\in A_X,\
\forall y\in A_Y,\ \forall z\in A_Z.
\]
But this is equivalent to
\[
\hat f(x)\in (A_Y\otimes A_Z)^*=A^*_{Y\otimes Z},\qquad \forall x\in A_X,
\]
which means that $\hat f\in \Af(X, (Y\otimes Z)^*)$.

\end{proof}
A *-autonomous category is compact closed if it satisfies $(X\otimes Y)^*=X^*\otimes
Y^*$. 
In general, $X\parr Y=(X^*\otimes Y^*)^*$ defines a dual symmetric monoidal
structure that is different from $\otimes$. 
We next show that $\Af$ is not compact closed.

\begin{prop}\label{prop:noncompact} For objects in $\Af$, we have $(X\otimes
Y)^*=X^*\otimes Y^*$ exactly in one of the following situations:
\begin{enumerate}
\item[(i)] $X\simeq I$ or $Y\simeq I$,
\item[(ii)] $d_X=d_Y=0$,
\item[(iii)] $d_{X^*}=d_{Y^*}=0$.
\end{enumerate}

\end{prop}

\begin{proof} Since $\FV$ is compact, we have $V_{(X\otimes Y)^*}=(V_X\otimes
V_Y)^*=V_X^*\otimes V_Y^*=V_{X^*\otimes Y^*}$. It is also easily seen by definition that
$A_{X^*}\otimes A_{Y^*}=A^*_X\otimes A^*_Y\subseteq A^*_{X\otimes Y}=A_{(X\otimes Y)^*}$, so that we always have $A_{X^*\otimes
Y^*}\subseteq A_{(X\otimes Y)^*}$.  Hence the equality holds if and
only if $d_{X^*\otimes Y^*}=d_{(X\otimes Y)^*}$. From Lemma
\ref{lemma:tensor_spaces}, we see that
\[
d_{X^*\otimes Y^*}=d_{X^*}+d_{Y^*}+d_{X^*}d_{Y^*}.
\]
On the other hand, we have using \eqref{eq:duality} that $L_{(X\otimes Y)^*}=S_{X\otimes
Y}^\perp=(S_X\otimes S_Y)^\perp$, so that
\[
d_{(X\otimes Y)^*}=D_XD_Y-\dim(S_X)\dim(S_Y)=D_XD_Y-(d_X+1)(d_Y+1).
\]
Taking into account that by \eqref{eq:duality} we have $d_{X^*}=D_X-d_{X}-1$, similarly
for $d_{Y^*}$, we obtain

\[
d_{(X\otimes Y)^*}-d_{X^*\otimes Y^*}=d_Xd_{Y^*}+d_{X^*}d_Y.
\]
This is equal to 0 iff $d_Xd_{Y^*}=d_Yd_{X^*}=0$, which amounts to the conditions in the
lemma.

\end{proof}

In a *-autonomous category, the internal hom can be identified as $X\multimap Y=(X\otimes
Y^*)^*$. The underlying vector space is $V_{X\multimap Y}=(V_X\otimes
V_Y^*)^*=V_X^*\otimes V_Y=V_X\multimap V_Y$
and we have seen in Section \ref{sec:fv} that we may identify this space with
$\FV(V_X,V_Y)$, by $f \leftrightarrow C_f$. This property is extended to $\Af$, in the
following sense.

\begin{prop}\label{prop:ihom_morphisms} For any objects $X,Y$ in $\Af$, the map $f\mapsto C_f$ is a bijection
of $\Af(X,Y)$ onto $A_{X\multimap Y}$. In particular,  $A^*_X$ can be identified with
$\Af(X,I)$.

\end{prop}

\begin{proof} Let $f\in \FV(V_X,V_Y)$. Since by definition $A_{X\multimap Y}=A^*_{X\otimes
Y^*}=(A_X\otimes A_Y^*)^*$, we see that $C_f\in A_{X\multimap Y}$ if and only
if for all $x\in A_X$ and $\tilde y\in A^*_Y$, we have
\[
1=\<C_f, x\otimes \tilde y\>=\<\tilde y,f(x)\>.
\]
This latter statement is clearly equivalent to $f(A_X)\subseteq A_Y$, so that $f\in
\Af(X,Y)$. 

\end{proof}

\subsubsection{Quantum and classical objects}
\label{sec:affine_q}
We now restrict to objects such that the underlying vector spaces are spaces
of hermitian matrices, as in Example
\ref{exm:quantum}. We may also restrict morphisms between such spaces to completely positive maps. We show that this restriction
amounts to taking an intersection of $A_{X\multimap Y}$ with the cone of positive semidefinite
matrices. {This, and subsequent examples,  shows that for characterization of sets  of
higher order quantum maps
 it is enough to
work with the category $\Af$.} 

An object $X$ of $\Af$ will be called {quantum} if $V_X=M_n^h$ for some $n$ and $A_X$ is an
affine subspace such that both $A_X$ and $A^*_X$ contain a positive multiple of the identity matrix
$E_n$\footnote{We use the notation $E_n$, and not $I_n$, to avoid the slight chance that
it might be confused with the monoidal unit.}.
(recall that we identify $(M_n^h)^*=M_n^h$). 

\begin{prop}\label{prop:ihom_quantum} Let $X$, $Y$ be quantum objects in $\Af$. Then 
\begin{enumerate}
\item[(i)] $X^*$ and $X\otimes Y$ are quantum objects as well. 
\item[(ii)] Let $V_X=M_n^h$, $V_Y=M_m^h$. Then for any $f\in \FV(M_n^h,M_m^h)$, we have
$C_f\in A_{X\multimap Y}\cap M_{mn}^+$ if and only if $f$ is completely positive and
\[
f(A_X\cap M_n^+)\subseteq A_Y\cap M_m^+.
\]
\end{enumerate}

\end{prop}

\begin{proof} The statement (i) is easily seen from  $A_X\otimes A_Y
\subseteq  A_{X\otimes Y}$ and $A^*_X\otimes A^*_Y\subseteq A^*_{X\otimes
Y}$. To show (ii), let $C_f\in   A_{X\multimap Y}\cap M_{mn}^+$. By the properties of the Choi
isomorphism $f$ is completely positive and by Proposition \ref{prop:ihom_morphisms},
$f(A_X)\subseteq A_Y$, this proves one implication. For the converse, note that we only
need to prove that under the given assumptions, $f(A_X)\subseteq A_Y$, for which it is enough
to show that $A_X\subseteq \aff(A_X\cap M_n^+)$. To see this, let  $t_XE_n\in  A_X$ for
some $t_X>0$.  Any element in $A_X$ can be written in the form $t_XE_n+v$ for some $v\in L_X$.
Since $t_XE_n\in int(M_n^+)$, there is some $s>0$ such that $a_\pm:=t_XE_n\pm sv\in M_n^+$, and
since $\pm sv\in L_X$, we see that $a_\pm \in A_X\cap M_n^+$. It is now easily checked
that
\[
t_XE_n+v=\frac{1+s}{2s}a_++\frac{s-1}{2s}a_-\in \aff(A_X\cap M_n^+). 
\]

\end{proof}

\begin{remark} Note that the same results can be obtained if we only assume that both $A_X$ and
$A_X^*$ contain an interior point of the positive cone $M_n^+$.

\end{remark}

\begin{remark} The above results show that the quantum objects with the morphisms
restricted by complete positivity form a *-autonomous subcategory in $\Af$. Observe also that this
subcategory is equivalent to the category $\mathrm{Caus}[\bf{CPM}]$, constructed in
\cite{kissinger2019acategorical}, which is referred to as higher-order causal category for
quantum theory. 

\end{remark}

Let $X=(M_n^h,A_X)$ be a quantum object. Since both $A_X$ and $A_X^*$ contain a positive multiple
of the identity $E_n$ and $\<E_n,a\>=\Tr[a]$ for any $a\in M_n^h$, we see that the object
$X$ is uniquely identified by the dimension $n_X=n$,
the  subspace $S_X\in M_n^h$ containing $E_n$ and a positive constant $c_X>0$ such that
\begin{equation}\label{eq:quantum_object}
A_X=\{a\in S_X,\ \Tr[a]=c_X\}.
\end{equation}
We then have
\[
L_X=\{v\in S_X,\ \Tr[v]=0\},\qquad S_X=L_X\oplus \mathbb R\{E_n\},
\]
where $\oplus$ denotes the orthogonal direct sum with respect to the inner product
$\<a,b\>=\Tr[ab]$ in $M_n^h$. 

Since we identify $M_n^h$ and $(M_n^h)^*$ using $\<\cdot,\cdot\>$, we also identify the
annihilator 
$S_X^\perp$ with the orthogonal complement of $S_X$ in $M_n^h$. We obtain
\begin{equation}\label{eq:duality_q}
S_{X^*}=L_{X^*}\oplus \mathbb R\{E_n\}=S_X^\perp\oplus \mathbb R\{E_n\}.
\end{equation}
Let $a\in A_X$ and  $\tilde a\in A_{X^*}$. Then $\tilde a=\tilde v+tE_n$ for $\tilde v\in
S_X^\perp$ and $t\in \mathbb R$ and from
\[
1=\<a,\tilde a\>=t\Tr[a]=tc_X
\]
we get $t=c_X^{-1}$, so that $c_{X^*}=\Tr[\tilde a]=c_X^{-1}n$. It is
also clear from \eqref{eq:aff_tensor} that if $Y$ is another quantum  object, then we have
$c_{X\otimes Y}=c_Xc_Y$. 

 The description of the affine subspace for a quantum object in  \eqref{eq:quantum_object} is, up
to the positivity constraint, basically the same as the quantum object set in \cite[Def.~1]{milz2024characterising}, where the linear subspace $S_X$ is characterized as the
range of an orthogonal projection (with respect to the Hilbert-Schmidt inner product). As
we have seen in Proposition \ref{prop:ihom_morphisms}, the affine subspace $A_{X\multimap
Y}$ is precisely the set of Choi matrices of  maps between affine sets $A_X$ and $A_Y$.
Our next result then corresponds to \cite[Theorem 2]{milz2024characterising}, where this
set is described in terms of the corresponding projections and trace constraints.

\begin{lemma}\label{lemma:internal_hom_q}  Let $X=(M_n^h,A_X)$, $A_X=\{a\in S_X,\
\Tr[a]=c_X\}$, $Y=(M_m^h,A_Y)$, $A_Y=\{a\in S_Y,\
\Tr[a]=c_Y\}$ be  quantum objects. Then $X\multimap Y=(M^h_{mn}, A_{X\multimap Y})$, where
$A_{X\multimap Y}$ is determined by 
\[
S_{X\multimap Y}=  (S_X^\perp\otimes M_m^h)\oplus (S_X\otimes L_Y)\oplus \mathbb
R\{E_{mn}\},\qquad c_{X\multimap Y}=nc_X^{-1}c_Y.
\]
\end{lemma}

\begin{proof} From $X\multimap Y=(X\otimes Y^*)^*$, we have by Lemma
\ref{lemma:tensor_spaces} and \eqref{eq:duality}, together with the above considerations, that
\[
S_{X\multimap Y}=L_{X\multimap Y}\oplus \mathbb R\{E_{mn}\},\quad c_{X\multimap
Y}=mn(c_Xc_{Y^*})^{-1}=(c_Xc_Y^{-1}m)^{-1}mn=nc_X^{-1}c_Y
\]
and
\begin{align*}
L_{X\multimap Y} =
(S_X\otimes S_{Y^*})^\perp =(S_X^\perp\otimes M_m^h) \vee (M_n^h\otimes
L_Y)=(S_X^\perp\otimes M_m^h)\oplus (S_X\otimes L_Y).
\end{align*}

\end{proof}

\begin{exm}[States, channels and combs] \label{exm:quantum_maps} A basic example of a quantum object
is $X=(M_n^h,A_X)$, determined by $S_X=V_X=M_n^h$ and
$c_X=1$. Let us denote this object by $\Se_n$. The set $A_{\Se_n}\cap M_n^+$ is the set of
$n$ by $n$ density matrices.  By Proposition
\ref{prop:ihom_morphisms}, $\Ce_{m,n}:=\Se_m\multimap \Se_n$ is a quantum object as well, with
underlying vector space $M_{mn}^h$, and
Proposition \ref{prop:ihom_quantum} shows that  $A_{\Ce_{m,n}}\cap M_{mn}^+$ is the set
of Choi matrices of completely positive maps $f:M_m\to M_n$ mapping states to states, 
i. e. {quantum channels}. 

Note that the dual object
$\Ce^*_{m,n}=\Se_m\otimes \Se_n^*$ represents the set of Choi matrices of replacement
channels $M_n\to M_m$, that is, channels that map any state in $M_n$ to a fixed state
in $M_m$. We also have $\Se^*_n=\Se_n\multimap I=\Ce_{n,1}$ and $\Se_n=I\multimap \Se_n=\Ce_{1,n}$. 
We can proceed inductively as follows. Put 
\[
\Ce^2_{m,n,k,l}=\Ce_{m,n}\multimap \Ce_{k,l}=(\Se_m\multimap \Se_n)\multimap (\Se_k\multimap
\Se_l),
\]
then
$A_{\Ce^2_{m,n,k,l}}\cap M_{mnkl}^+$
is the set of Choi matrices of completely positive maps $M_{mn}\to M_{kl}$, mapping Choi matrices of channels
to Choi matrices of channels, or {quantum 2-combs},
\cite{chiribella2008transforming,chiribella2009theoretical}. The 1-combs
coincide with quantum channels. The {quantum
$(N+1)$-combs} for any $N$ is  the set of Choi matrices of
completely positive maps mapping $N$-combs to 1-combs. The corresponding objects in
$\Af$ are then given as
$\Ce^N_{n_1,\dots,n_{2N}}=\Ce^{N-1}_{n_2,\dots,n_{2N-2}}\multimap \Ce_{n_1,n_{2N}}$. By
Proposition 
\ref{prop:ihom_quantum}, these are all quantum objects. It was proved in
\cite{chiribella2009theoretical} that
the completely positive maps corresponding to $N$-combs have the form 

\begin{center}
\includegraphics[scale=0.8]{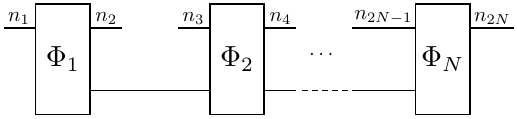}
\end{center}
for some channels $\Phi_1,\dots, \Phi_N$.
\end{exm}

\begin{exm}[No-signalling channels and process matrices]\label{exm:ns_pm}
No-signalling channels describe the situation when no signalling is allowed between two
parties $A$ and $B$. As it was shown in \cite{chiribella2013quantum, gutoski2009properties}, this
corresponds to bipartite channels that can be obtained by affine combinations of channels
acting locally on $A$ and $B$. In our setting, the  bipartite channels are represented by
the object $\Ce_{n_An_B,m_Am_B}$ and the no-signalling channels by  $\Ce_{n_A,m_A}\otimes
\Ce_{n_B,m_B}$, see also \cite{kissinger2019acategorical}.

Process matrices were introduced in \cite{oreshkov2012quantum} as
positive linear functionals that map tensor products of pairs (or
more generally tuples) of Choi matrices of channels to 1:
\begin{center}
\includegraphics[scale=0.8]{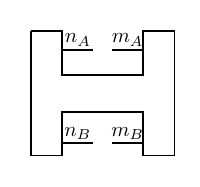}
\end{center}

\noindent
By linearity, this property extends to all
no-signalling channels. In our setting, the process matrices correspond to $(\Ce_{n_A,m_A}\otimes
\Ce_{n_B,m_B})^*$.

\end{exm}

We can define classical objects in $\Af$ in a similar way, replacing $M_n^h$ by $\mathbb
R^N$ and the positive cone by $\mathbb R_+^N$, and we require that  both
$A_X$ and $A^*_X$ contains a positive multiple of the unit vector $e_N=(1,\dots,1)\in
\mathbb R^N$. A statement similar to Proposition
\ref{prop:ihom_quantum} holds in this case,
with complete positivity replaced by positivity. Again, it can be seen that the
*-autonomous subcategory of classical objects and positive morphisms is equivalent to the
category $\mathrm{Caus}[\bf{Mat}(\mathbb R_+)]$ of \cite{kissinger2019acategorical}.

We can similarly treat
classical-to-quantum and quantum-to-classical maps as morphisms between these types of
objects, satisfying appropriate positivity assumptions.

\begin{exm}[Partially classical maps]\label{exm:quantum_classical}  We may similarly define the basic
classical object as 
\[
\Pe_k:=(\mathbb R^k, \{x\in \mathbb R^k,\ \sum_i x_i=1\}).
\]
In this case, $\mathcal A_{\Pe_k}\cap \mathbb R^k_+$ is the probability simplex. We then
obtain further  useful objects by combining $\Pe_k$ with  quantum objects. For example, it can be
easily seen that $\Se_n\multimap \Pe_k$ intersected with the cone $M_n^+\otimes \mathbb R^n_+$
corresponds to $k$-outcome {measurements}. Similarly, we obtain $k$-outcome {quantum instruments}
 from $\Se_m\multimap (\Se_n\otimes \Pe_k)$, {quantum multimeters} from $(\Se_m\otimes
\Pe_k)\multimap \Pe_l$, {quantum testers} from  $\Ce^N_{n_1,\dots,n_{2N}}\multimap
\Pe_k$, etc. In general, maps of the form $X^*\multimap \Pe_k$ can be used to define
probabilistic objects related to $X$.  

\end{exm}

\subsection{First order and higher order objects}

We say that an object $X$ in $\Af$ is {first order} if $d_X=D_X-1$, equivalently, $S_X=V_X$.
Another equivalent condition is $d_{X^*}=0$, which means that $A_X$ is determined by a
single element $\tilde a_X\in V_X^*$ as 
\[
A_X=\{\tilde a_X\}^*,\qquad A^*_X=\{\tilde a_X\}.
\]
Note that first order objects, resp. their duals, are exactly those satisfying
condition (iii), resp. condition (ii), in Proposition \ref{prop:noncompact}, in
particular, $(X\otimes Y)^*=X^*\otimes Y^*$ for first order objects $X$ and $Y$.

{Higher order objects} in $\Af$ are objects  obtained from a finite set $\{X_1,\dots,X_n\}$ of first order objects by
taking tensor products and duals. The above is indeed a set, so that all the objects are
different (though they may be isomorphic) and the ordering is not essential. We will also
assume that the monoidal  unit $I$ is not contained in this set. By definition of
$X\multimap Y$,
and since we may identify $X\multimap I$ with $X^*$, we see that higher order objects are
also generated by applying the internal hom inductively on $\{X_1,\dots, X_n\}$ if we allow $X_i=I$ for some  
$i$. It follows that the objects introduced  in Examples \ref{exm:quantum_maps},
\ref{exm:ns_pm} and
\ref{exm:quantum_classical}
are indeed higher order objects in $\Af$  according to the above definition.

Of course, any first order
object is also higher order with $n=1$. Note that we cannot say that a higher order object
generated from $\{X_1,\dots, X_n\}$ is automatically ''of order $n$'', as the following lemma shows. 

\begin{lemma}\label{lemma:1ordertensor} Let $X$, $Y$ be first order, then $X\otimes Y$ is
first order as well.

\end{lemma}

\begin{proof} We have $S_{X\otimes Y}=S_X\otimes S_Y=V_X\otimes V_Y=V_{X\otimes Y}$.

\end{proof}

As we have seen, higher order objects are obtained by applying the internal hom
iteratively. The next lemma shows some properties of such iterations.  

\begin{lemma}\label{lemma:combs} Let $X,Y,Z$ be any objects in $\Af$ and let
$F_0,F_1,F_2$ be first order objects. Then we have
\begin{enumerate}
\item[(i)] $Z\multimap (X\multimap Y)\simeq (Z\otimes X)\multimap Y\simeq X\multimap
(Z\multimap Y)$.
\item[(ii)] $(Z\multimap F_1)\multimap (F_0\multimap F_2)\simeq ((F_0\multimap Z)\multimap
F_1)\multimap F_2$.
\item[(iii)]  If $F=(V_F, \{\tilde a_F\}^*)$ is first order, then $Z\multimap F$ is determined as
\[
A_{Z\multimap F}=\{w\in V_{Z}^*\otimes V_F, (id\otimes \tilde a_F)(w)\in
A_Z^*\}.
\]

\end{enumerate}

\end{lemma}

\begin{proof}
 The statement (i) is a well known property (currying) of  closed
symmetric monoidal categories and (iii) is easily seen from the
definition and properties of $\multimap$.
The proof of (ii) is given in Appendix \ref{app:lemma}. 

\end{proof}

Note also that since we identify $X^{**}=X$ for any object $X$, the isomorphisms in (i)
above are given by the symmetries in $\FV$, that is, by permutations of the components in
the tensor products of the underlying vector spaces. To save some parentheses, we also
assume that the internal homs associate to the right, so we write
$X\multimap Y\multimap Z$ instead of $X\multimap (Y\multimap Z)$. 

\begin{exm}[Channels and Combs] \label{exm:chans_combs} Let  $X$ and $Y$ be first order objects in $\Af$. As we
have seen, $C_1(X,Y):= X\multimap Y$ is then a higher order object, called a {channel} or 
1-comb (We slightly  abuse
the terminology here).  We will inductively construct higher order objects in $\Af$,
similarly as in Example \ref{exm:quantum_maps}. An
{$N$-comb} over first order objects $X_1,\dots,X_{2N}$ is an object
\begin{align*}
C_N(X_1,\dots,X_{2N})&:=C_{N-1}(X_2,\dots,X_{2N-1})\multimap (X_1\multimap X_{2N})\\
&\simeq X_{1}\multimap C_{N-1}(X_2,\dots,X_{2N-1})\multimap X_{2N}\\
&\simeq X_1\multimap(X_2\multimap\dots \multimap (X_{N}\multimap
X_{N+1})\multimap \dots \multimap X_{2N-1})  \multimap X_{2N}\\
&\simeq (\dots (X_1\multimap X_2)\multimap X_3)\multimap\dots )\multimap
X_{2N-1})\multimap X_{2N}
\end{align*}
where the isomorphisms follow by repeated use of the statements (i) and (ii) of Lemma
\ref{lemma:combs}. Note that the first two isomorphism are obtained using only the
part  (i) of the Lemma, so that they hold even if the involved object are not first order, but 
we need to assume first order objects for the isomorphism in the last line. The subspace $A_{C_N}$ for an
$N$-comb $C_N$ can be found inductively, using Lemma \ref{lemma:combs} (iii).
If $X_1,\dots, X_{2N}$ are quantum objects, then $\tilde a_{X_i}$ is always a multiple of
the identity, so we obtain the characterization of quantum combs  in
\cite{chiribella2009theoretical}.

\end{exm}

\section{Combinatorial description of higher order objects}

Types of higher order quantum maps were introduced in 
\cite{perinotti2017causal,bisio2019theoretical},  where a combinatorial
description of the corresponding subspaces is given. More precisely, it is shown that the
subspace for the corresponding type of higher order objects can be combined from
some independent system of  linear subspaces
labelled by binary strings.

In this section we show that we can have a similar description of higher order objects in
$\Af$, though the construction is slightly more complicated,  see also Remark
\ref{rem:innerp} below. We will use boolean
functions to characterize the  subsets of binary strings corresponding to the type, which
will turn out useful for further description of types.   We will need the definitions, notations and results  given in Appendix
 \ref{sec:app_affine} for affine subspaces, 
Appendix \ref{sec:boolean} for boolean functions and the notations given at the end of
Section \ref{sec:fv} for operations on subspaces of a vector space.  The notation $\boxplus$ will always mean the  disjoint union of subsets  of natural numbers and the operadic
composition of permutations, as in Appendix \ref{app:basic}.

For a first order object $X=(V_X, \{\tilde a_X\}^*)$, let us pick an element $a_X\in
A_X$. We have a direct sum decomposition
\begin{equation}\label{eq:direct}
V_X=L_{X,0}+ L_{X,1},
\end{equation}
where $L_{X,0}:= \mathbb R\{a_X\}$, $L_{X,1}:=\{\tilde a_X\}^\perp=L_X$.
We define the {conjugate object}  as $\tilde X=(V_X^*,\{a_X\}^*)$. Note that we always
have $\tilde a_X\in A_{\tilde X}$ and with the choice $a_{\tilde X}=\tilde a_X$, we obtain 
$\tilde{\tilde X}=X$ and 
\begin{equation}\label{eq:complement}
L_{\tilde X,u}=L_{X,1-u}^\perp,\qquad u\in \{0,1\}.
\end{equation}
These  definitions depend on the choice of $a_X$, but we will assume below that this
choice is fixed and that we choose $a_{\tilde X}=\tilde a_X$. Since we will always work
with a finite set of objects at a time, this will not create any problems. 

 The situation is simplified in the case of quantum or classical objects,
where we have an identification of $V_X$ with its dual $V_X^*$ by a standard inner product, and  $\tilde a_X$
is a (positive) multiple of the identity by definition. The hyperplane $\{\tilde a_X\}^*$
contains a (necessarily unique) positive multiple of identity as
well, which we then choose to be the distinguished element $a_X$. Any first  order quantum object $X$ 
can be also determined by the dimension $n_X$ and a positive constant $c_X$
 such that  $X=(M_{n_X}^h,A_X=\{a,\ \Tr[a]=c_X\})$. In this case, we have $\tilde
 a_X=c_X^{-1}E_{n_X}$, $a_X\in n_{X}^{-1}c_XE_{n_X}$ and the conjugate object is the
 first order quantum object determined by $n_{\tilde X}=n_X$ and $c_{\tilde X}=n_Xc_X^{-1}$. 
This gives us the quite
familiar decomposition of any hermitian operator as a multiple of the identity and 
a traceless part:
\begin{equation}\label{eq:quantum_conj}
L_{X,0}=L_{\tilde X,0}= \mathbb R\{E_{n_X}\},\qquad  L_{X,1}=L_{\tilde X,1}=
\{E_n\}^\perp=\{a\in M_{n_X}^h,\ \Tr[a]=0\}
 \end{equation}
Similarly for classical objects.

\begin{remark}\label{rem:innerp} It is also possible to use an identification 
$V_X\simeq V_X^*$ by fixing some inner product  in $V_X$. The element $\tilde
a_X$ then becomes identified with an element of $V_X$, and $A_X=\{\tilde a_X\}^*$ 
contains a (unique) multiple of $\tilde a_X$. Then both $V_X$ and $V_X^*$ are decomposed
into $\mathrm{span}(\tilde a_X)$ and its orthogonal complement. But note that the choice
of the inner product is quite arbitrary, and the constructions below would still depend on
this choice. We could also choose another inner product in $M_n^h$ and obtain a different
decomposition. Nevertheless, the subspaces obtained for higher order objects would still remain the same,
regardless of the chosen approach.

\end{remark}

Let $X_1,\dots,X_n$ be first order objects in $\Af$. Let $a_{X_i}\in A_{X_i}$ be fixed and
let $\tilde X_i$ be the conjugate first order objects. Let us denote $V_i=V_{X_i}$ and 
\[
L_{i,u}:= L_{X_i,u},\qquad  \tilde L_{i,u}:= L_{\tilde X_i,u} \qquad u\in \{0,1\},\ i\in [n].
\]
For a string $s\in \{0,1\}^n$, we define
\[
L_s:=L_{1,s_1}\otimes\dots \otimes L_{n,s_n}, \qquad \tilde L_s:=\tilde
L_{1,s_1}\otimes\dots \otimes \tilde L_{n,s_n},
\]
then by \eqref{eq:direct} we have the direct sum decompositions 
\[
V:=V_1\otimes \dots \otimes V_n=\sum_{s\in \{0,1\}^n} L_s,\qquad V^*=V_1^*\otimes
\dots\otimes V_n^*=\sum_{s\in \{0,1\}^n} \tilde L_s.
\]

\begin{lemma}\label{lemma:Lperp}   For any $s\in \{0,1\}^n$, we have 
\[
L_s^\perp=
\sum_{t\in\{0,1\}^n} (1-{\chi}_s(t))\tilde L_t,\qquad \tilde L_s^\perp=
\sum_{t\in\{0,1\}^n} (1-{\chi}_s(t))L_t.
\]
 Here $\chi_s:\{0,1\}^n\to\{0,1\}$ is the characteristic function of $s$, that is,
 $\chi_s(t)=1$ iff $t=s$.  

\end{lemma}

\begin{proof} Using \eqref{eq:complement} and the direct sum decomposition of $V_i^*$, we get
\begin{align*}
\left(L_{1,s_{1}}\otimes \dots\otimes L_{n,s_{n}}\right)^\perp&= \bigvee_j\left(
V_{1}^*\otimes
\dots \otimes V_{j-1}^*\otimes \tilde L_{j,1-s_{j}}\otimes V_{j+1}^*\otimes\dots \otimes
V_{n}^*\right)\\
&= \bigvee_j \left( \sum_{\substack{t\in \{0,1\}^n\\ t_{j}\ne s_{j}}} \tilde
L_{1,t_{1}}\otimes\dots \otimes \tilde
L_{n,t_{n}}\right)= \sum_{\substack{t\in \{0,1\}^n\\ t\ne s}} \left( \tilde L_{1,t_{1}}\otimes\dots \otimes \tilde
L_{n,t_{n}}\right).
\end{align*}
The proof of the other equality is the same.

\end{proof}

For the fixed first order objects  $X_1,\dots,X_n$ and their conjugate objects $\tilde
X_1,\dots,\tilde X_n$, we introduce the following definitions. Put $a:=
a_{X_1}\otimes\dots \otimes  a_{X_n}$, $\tilde a:= \tilde
a_{X_1}\otimes\dots\otimes  \tilde a_{X_n}$.
For  $f\in \Fe_n$ (see Appendix \ref{app:functions}) define 
\begin{equation}\label{eq:SfAf}
S_f=S_f(X_1,\dots,X_n):=\sum_{s\in \{0,1\}^n} f(s)L_s,\qquad
A_f=A_f(X_1,\dots,X_n):=S_f\cap \{\tilde a\}^*.
\end{equation}
It is clear from definition that $A_f$ is an affine subspace. Since
$f(\theta_n)=1$, the space $S_f$ always contains the subspace $L_0=L_{1,0}\otimes\dots\otimes
L_{n,0}=\mathbb R\{a\}$ and it is clear that $L_s\subseteq \{\tilde a\}^\perp$ for any
$s\ne \theta_n$. It follows that $a\in A_f$, so that $A_f\ne \emptyset$, and since $A_f\subseteq
\{\tilde a\}^*$, we see that  $A_f$ is proper and $\tilde
a\in A^*_f$.  It is easy to see that we have
\[
\lin(A_f)=\sum_{s\in\{0,1\}^n\setminus\{\theta_n\}} f(s)L_s,\qquad \Span(A_f)=S_f.
\]
We may now define the objects
\[
X_f=X_f(X_1,\dots, X_n):=(V,A_f(X_1,\dots,X_n))
\]
in $\Af$.

\begin{prop}\label{prop:Xf_const} Let $X_1,\dots,X_n$ be first order objects and $\tilde
X_1,\dots,\tilde X_n$ the conjugate objects. 
The map $\Fe_n\ni f\mapsto X_f(X_1,\dots,X_n)\in \Af$ is injective and we have the
following properties. 
\begin{enumerate}
\item[(i)] For the least and the largest element in $\Fe_n$, 
\[
X_{p_{n}}=\tilde X_1^*\otimes \dots\otimes \tilde X_n^*=(\tilde X_1\otimes \dots\otimes
\tilde X_n)^*,\qquad
X_{1_n}=X_1\otimes\dots\otimes X_n.
\]

\item[(ii)] The dual object satisfies
\[
X_f^*(X_1,\dots,X_n)=X_{f^*}(\tilde X_1,\dots,\tilde X_n).
\]
\item[(iii)] Let $f_1\in \Fe_{n_1}$, $f_2\in \Fe_{n_2}$. For the decomposition
$[n]=[n_1]\boxplus[n_2]$,
\[
X_{f_1\otimes f_2}(X_1,\dots,X_n)=X_{f_1}(X_1,\dots, X_{n_1})\otimes
X_{f_2}(X_{n_1+1},\dots,X_n).
\]
\item[(iv)] For any $\sigma\in \permut_n$ we have an isomorphism
\[
X_{f\circ\sigma}(X_1,\dots,X_n)\xrightarrow{\sigma}
X_f(X_{\sigma^{-1}(1)},\dots,X_{\sigma^{-1}(n)}) .
\]
\end{enumerate}
\end{prop}

\begin{proof} Since  $L_s$, $s\in \{0,1\}$ is an independent decomposition of $V$, the
subspace $S_f$ has a unique decomposition in terms of $L_s$. It follows 
 that the map $f\mapsto A_f$, and hence also $f\mapsto X_f$ is injective. 
We have 
\[
S_{p_n}=L_{\theta_n}=\mathbb R\{a\}\qquad S_{1_n}=\sum_{s\in \{0,1\}^n}L_s=V,
\]
Since $X_1\otimes\dots\otimes X_n=(V,\{\tilde a\}^*)$ and $\tilde X_1\otimes \dots\otimes
\tilde X_n=(V^*,\{a\}^*)$, this proves (i). For (ii), it is enough to prove that 
$A^*_f(X_1,\dots,X_n)=A_{f^*}(\tilde X_1,\dots, \tilde X_n)$. To see this, we compute
using Lemma \ref{lemma:Lperp} and the fact that the subspaces form an independent
decomposition,
\begin{align*}
\Span( A^*_f)&=\lin(A_f)^\perp=\left(\sum_{s\in\{0,1\}^n\setminus\{0\}}
f(s)L_s\right)^\perp=
\bigwedge_{\substack{s\in\{0,1\}^n\\ s\ne 0, f(s)=1}}L_s^\perp\\ &=
\bigwedge_{\substack{s\in\{0,1\}^n\\ s\ne 0,
f(s)=1}}\left(\sum_{t\in\{0,1\}^n}(1-{\chi}_s(t))\tilde L_t\right)\\
&=\sum_{t\in\{0,1\}^n} \left(\bigwedge_{\substack{s\in \{0,1\}^n\\ s\ne 0, f(s)=1}}
(1-{\chi}_s(t))\tilde L_t\right)=\sum_{t\in \{0,1\}^n} f^*(t) \tilde L_t.
\end{align*}
To see the last equality, note that
\[
\bigwedge_{\substack{s\in \{0,1\}^n\\ s\ne 0, f(s)=1}}
(1-{\chi}_s(t))=\begin{dcases} 1 & \text{if } t=\theta_n\\ 1-f(t) & \text{if } t\ne \theta_n
\end{dcases} \ = f^*(t).
\]
The statement (iii) is easily seen from the definitions. To show (iv), compute
\begin{align*}
\sigma^{-1}(S_f(X_{\sigma^{-1}(1)},\dots,X_{\sigma^{-1}(n)}))&=\sigma^{-1}(\sum_{s}
f(s)L_{\sigma^{-1}(1),s_1}\otimes\dots \otimes
L_{\sigma^{-1}(n),s_n})\\
&=\sum_s f(s) L_{1,s_{\sigma(1)}}\otimes\dots\otimes
L_{n,s_{\sigma(n)}}=S_{f\circ\sigma}(X_1,\dots,X_n).
\end{align*}
It follows that
\[
A_f(X_{\sigma^{-1}(1)},\dots, X_{\sigma^{-1}(n)})=S_f(X_{\sigma^{-1}(1)},\dots,
X_{\sigma^{-1}(n)})\cap
\{\sigma(\tilde a)\}^*=
\sigma(A_{f\circ\sigma}(X_1,\dots,X_n)).
\]

\end{proof}

With the decomposition $[n]=[n_1]\boxplus[n_2]$ and $f\in \Fe_{n_1}$,  $g\in \Fe_{n_2}$, we
will use the notations
\begin{equation}\label{eq:notations}
f\to g:=(f\otimes g^*)^*,\qquad f\parr g:=(f^*\otimes g^*)^*.
\end{equation}

\begin{coro}\label{coro:maps} With the above notations, we have 
\[
X_{f\to g}(X_1,\dots,X_{n_1+n_2})=X_{f}(\tilde X_1,\dots,\tilde X_{n_1})\multimap
X_g(X_{n_1+1},\dots, X_{n})
\]
and
\[
X_{f\parr g}=X_f(X_1,\dots, X_{n_1})\parr X_g(X_{n_1+1},\dots, X_n).
\]
\end{coro}

\begin{proof}
Immediate from Proposition \ref{prop:Xf_const}.
\end{proof}

Since  $\{L_s, \ s\in \{0,1\}^n\}$ is an independent decomposition of $V$,   the subspaces $S_f$ form a
distributive sublattice in the lattice of subspaces of $V$ and we clearly have 
$f\le g$ if and only if $S_f\subseteq S_g$, and $S_{f\wedge g}=S_f\cap S_g$, $S_{f\vee
g}=S_f\vee S_g$.  With these operations and complementation given as $S_f\mapsto S_{f^*}$, 
each sequence $X_1,\dots,X_n$ of first order objects defines a representation of the boolean algebra $\Fe_n$ as
a boolean lattice (i. e. a complemented distributive lattice) of subspaces of a vector
space $V=V_{X_1}\otimes \dots\otimes V_{X_n}$, with a one dimensional subspace $\mathbb R
a$ as the bottom and the whole space as the top element.

If all the first order objects are quantum, we have an identification
$V=V^*=M_N^h$ and both $a$ and $\tilde a$ are positive multiples of the identity $E_N$. It
follows that all $X_f(X_1,\dots,X_n)$ are quantum objects as well. Further, by
\eqref{eq:quantum_conj} the subspaces
$S_f(X_1,\dots, X_n)$ remain unchanged if some $X_i$ are replaced with $\tilde X_i$, in
particular, we get  by \eqref{eq:duality_q}
\[
S_{f^*}(X_1,\dots,X_n)=S_{f^*}(\tilde X_1,\dots,\tilde
X_n)=S_{X^*_f(X_1,\dots,X_n)}=S_f(X_1,\dots,X_n)^\perp+\mathbb
R\{E_N\}.
\]
This shows that the complementation in $\Fe_n$ is reflected as the orthocomplementation in
the sublattice of subspaces in $M_N^h$ containing $\mathbb R\{E_N\}$.

Let us also note that in terms of the objects in the category $\Af$, the lattice structure
in $\{X_f(X_1\dots,X_n),\ f\in \Fe_n\}$ is obtained as follows. We have  $f\le g$ if and only if $X_f\xrightarrow{id_V} X_g$
(which means that $A_f\subseteq A_g$). It can be shown that if $k\le f,g\le h$, then the following diagrams are  a pullback resp.  pushout:
\[
\xymatrix{
X_{f\wedge g}\ar[r]^{id_V}
\ar[d]_{id_V} & X_f\ar[d]^{id_V} \\
X_g \ar[r]_{id_V}& X_h
} \qquad 
\xymatrix{
X_{k}\ar[r]^{id_V}
\ar[d]_{id_V} & X_f\ar[d]^{id_V} \\
X_g \ar[r]_{id_V}& X_{f\vee g}
}
\]
This holds in particular for the bottom and top elements $k=p_n$ and $h=1_n$. In the
situation that the order of the first order objects $X_1,\dots,X_n$ is not fixed, we may
replace the identity arrows above by appropriate permutations, as in Proposition
\ref{prop:Xf_const} (iv).

Our main goal in this paragraph  is to show that the higher order objects are precisely those of the form
$Y=X_f(X_1,\dots,X_n)$ for some choice of the
first order objects $X_1,\dots, X_n$ and a function $f$ that belongs to a special subclass
 $\Te_n\subseteq \Fe_n$. The elements of this subclass will be called the {type
 functions},
 or {types}, and are defined as those functions in $\Fe_n$ that can be obtained by taking
 the constant function $1_1$ in each coordinate and then repeatedly applying
 complementation  and tensor
 products of such functions in any order. The set of indices for which the corresponding
 coordinate  was subjected to taking the dual an odd  number of times will be called the
{inputs} (of $f$) and denoted by $I=I_f$, indices  in $O=O_f:=[n]\setminus I_f$ will be
called {outputs}. The reason for this terminology will become clear later. It is easy to observe that if $f\in \Te_n$, then $O_{f^*}=I_f$ and $I_{f^*}=O_f$. Further,
for $f_1\in \Te_{n_1}$, $f_2\in \Te_{n_2}$, we have $O_{f_1\otimes f_2}=O_{f_1}\boxplus
O_{f_2}$ and  $I_{f_1\otimes f_2}=I_{f_1}\boxplus
I_{f_2}$, see \eqref{eq:disu} for the definition of $\boxplus$.

We have the following  description of the sets of type functions.

\begin{prop}\label{prop:type_min} The system  $\{\Te_n\}_{n\in \mathbb N}$ is the smallest
system  such that
\begin{enumerate}
\item $\Te_1=\Fe_1$, $\Te_n\subseteq \Fe_n$ for all $n$,

\item For $[n]=[n_1]\boxplus [n_2]$, $\Te_{n_1}\otimes \Te_{n_2}\subseteq \Te_{n}$,
\item $\Te_n$  is invariant under permutations: if $f\in \Te_n$, then $f\circ \sigma\in
\Te_n$ for any $\sigma\in \permut_n$,
\item $\Te_n$  is invariant under complementation: if $f\in \Te_n$ then $f^*\in \Te_n$.

\end{enumerate}

\end{prop}

\begin{proof} It is clear by construction that any system of subsets $\{\Se_n\}_n$ with
these properties must contain the type functions and that $\{\Te_n\}_n$ itself has these
properties.

\end{proof}

To simplify some of the  statements in the next sections, we
also introduce the set $\Te_0$ of trivial type
functions on the empty string $\varepsilon$, containing a unique element
$1_\varepsilon: \varepsilon\mapsto 1$. It is easy to see that for any $f\in \Te_n$, we
have $1_\varepsilon^*=1_\varepsilon$ and  $1_\varepsilon \otimes f=f$.

\bigskip

Assume that  $Y$ is a higher order object constructed from a set of distinct first
order objects $Y_1,\dots, Y_n$, $Y_i=(V_{Y_i},\{\tilde a_{Y_i}\}^*)$.
Let us fix elements $a_{Y_i}\in A_{Y_i}$ and construct the conjugate objects $\tilde Y_i$. 
By compactness of $\FV$, we may assume (relabeling the objects if necessary) that the vector space of $Y$ has the form
\[
V_Y=V:=V_{1}\otimes \dots\otimes V_{n},
\]
where  $V_i$ is either $V_{Y_i}$ or $V_{Y_i}^*$, according to whether $Y_i$ was subjected
to taking duals an even or odd number of times. Similarly as for the type functions, the indices such that the first
case is true will be called the  outputs and the subset of outputs in $[n]$ will be denoted
by $O$, or $O_Y$, when we need to specify the object. The set $I=I_Y:=[n]\setminus O_Y$ is
the set of  inputs. Note that although we cannot yet exclude that $Y$ was constructed from
$Y_i$ in several different ways, the input and output spaces are always the same, fixed in 
the structure of $V$.

\begin{theorem}\label{thm:boolean} Let $Y$ be a higher order object, constructed from first
order objects $Y_1,\dots,Y_n$. For $i\in [n]$, let 
$X_i=Y_i$ if $i\in O_Y$ and $X_i=\tilde Y_i$ for $i\in I_Y$. 
There is a unique function $f\in \Te_n$, with $O_f=O_Y$,  such that 
\[
Y= X_f=(V, A_f(X_1,\dots,X_n)).
\]
Conversely, let $X_1,\dots, X_n$  be first order objects  and let
$f\in \Te_n$. Then $Y=X_f$ is a higher order object with $O_Y=O_f$, with underlying first
order objects $Y_1,\dots, Y_n$, where $Y_i=X_i$ for $i\in O_f$ and $Y_i=\tilde X_i$ for
$i\in I_f$.  

\end{theorem}

\begin{proof} Since the map $f\mapsto X_f$ is injective, uniqueness is clear.  To show existence of this
function, we will proceed by induction on $n$. For $n=1$, the assertion is easily seen
to be true, since in this case, we we have either $Y=Y_1$ or $Y=Y_1^*$. In the first case,
$O=\{1\}$, $I=\emptyset$, 
$X_1=Y_1$ and 
\[
S_Y=V_Y=V_1=L_{1,0}+ L_{1,1}=1(0)L_{1,0}+ 1(1)L_{1,1}=S_1(X_1),
\]
so in this case $f\in \Te_1$ is the constant 1. If $Y=Y_1^*$, we have $O=\emptyset$,
$I=\{1\}$, $X_1=\tilde
Y_1$, and then
\[
S_Y=\mathbb R\{\tilde a_{Y_1}\}=L_{1,0}=p_1(0)L_{1,0}+ p_1(1)L_{1,1}=S_{p_1}(X_1),
\]
so that $f=1^*=p_{1}\in \Te_1$. It is clear that $O_f=O_Y$ in both cases. 

Assume now that the assertion is true for
all $m<n$. By construction, $Y$ is either the tensor
product $Y=Z_1\otimes Z_2$, with
$Z_1$ constructed from $Y_{1},\dots, Y_{m}$ and $Z_2$ from $Y_{{m+1}},\dots,
Y_{n}$,
 or $Y$ is the dual of such a product. Let us assume the first case. It is clear that
 $O_{Z_1}\boxplus O_{Z_2}=O_Y$, and similarly for $I$, so that the corresponding objects
 $X_1,\dots, X_m$ and $X_{m+1},\dots,X_n$  remain the same. By the induction 
assumption, there are functions $f_1\in \Te_m$ and $f_2\in \Te_{n-m}$ such that
$O_{f_1}=O_{Z_1}$, $O_{f_2}=O_{Z_2}$ and,  by Proposition \ref{prop:Xf_const}(iii), 
\[
Y=Z_1\otimes Z_2=X_{f_1}(X_1,\dots,X_m)\otimes X_{f_2}(X_{m+1},\dots,X_n)=X_{f_1\otimes
f_2}(X_1,\dots,X_n).
\]
This implies the assertion, with $f=f_1\otimes f_2\in \Te_n$ and $O_f=O_{f_1}\boxplus
O_{f_2}=O_Y$. 
To finish the proof, it is now enough to observe that if the assertion holds for $Y$ then
it also  holds for $Y^*$. So assume that $Y=X_f(X_1,\dots,X_n)$ for some $f\in \Te_n$, 
then by Proposition \ref{prop:Xf_const}(ii), $Y^*=X_f^*= X_{f^*}(\tilde X_1,\dots,\tilde X_n)$. 
By the construction of conjugate objects, we have  $\tilde X_i=\tilde{\tilde Y}_i=Y_i$
if $i\in I_{Y}$ and $\tilde X_i=\tilde {Y}_i$ if $i\in O_Y$. Since by definition and the
assumption,
$O_{Y^*}=I_Y=I_f=O_{f^*}$, this proves the statement.

The converse is proved by a similar induction argument, using Proposition
\ref{prop:Xf_const}.

\end{proof}

Let us stress that in general, the objects $X_f$ depend on the choice of the elements
$a_{X_i}$. From the above proof, it is clear that  the  construction in Theorem
\ref{thm:boolean} does not depend on the choice of the elements $a_{Y_i}\in A_{Y_i}$.
Furthermore, if  all the first order objects are quantum, 
we have $S_Y=S_f(Y_1,\dots,Y_n)$,
since the space $S_f$ is unchanged if some of the objects are replaced by conjugates. 
If $Y_i$ are determined by the dimension $m_i$ and positive constants $c_i$, then 
$\tilde a_i=c_i^{-1}E_{m_i}$ and $a_i=n_i^{-1}c_iE_{m_i}$. It follows that $Y=(M_N^h,A_Y)$ is a quantum
object determined by the subspace $S_f(Y_1,\dots,Y_n)$ and $c=\frac{c_O}{c_I}m_I$, with 
$N=\Pi_i m_i$, $m_I=\Pi_{i\in I}m_i$, $c_O=\Pi_{i\in O} c_i$ and $c_I=\Pi_{\i\in I}c_i$.

\section{The type functions}\label{sec:type}

The aim of this section is to gain some understanding into the structure and properties of
the set of type functions. 
We start by an important example.

\begin{exm}\label{exm:type_channels}
Let $T\subseteq [n]$. It is easily seen that the function  $p_T$ (see Example \ref{ex:pS}
in Appendix \ref{sec:boolean}) is a type function, since we have
\[
p_T(s)=\Pi_{j\in T}(1-s_j)=\Pi_{j\in T} 1^*(s_j)=(\otimes_{j\in T}1^*)(s).
\]
By definition, $T$ is the set of inputs for $p_T$. Let $Y_1,\dots, Y_n$ be 
 first order objects, $Y_i=(V_i, \{\tilde a_{Y_i}\}^*)$ and let $X_1,\dots, X_n$ be first
 order objects as in Theorem \ref{thm:boolean}, so that $X_i=\tilde Y_i$ if $i\in T$ and
 $X_i=Y_i$ otherwise. Let $k=|T|$ and let $\sigma\in \permut_n$ be such that
$\sigma^{-1}(T)=[k]$. Then 
$p_T\circ \sigma=p_{k}\otimes 1_{n-k}$. By Proposition \ref{prop:Xf_const}, it follows that
we have the  isomorphism  
\[
X_{p_T}(X_1,\dots,X_n)\overset{\sigma}{\simeq}X_{p_{k}\otimes
1_{n-k}}(X_{\sigma^{-1}(1)},\dots, X_{\sigma^{-1}(n)})=Y_T^*\otimes Y_{[n]\setminus T},
\]
here $Y_T:=\otimes_{j\in T} \tilde X_j=\otimes_{j\in T} Y_j$ and $Y_{[n]\setminus T}:=\otimes_{j\in
[n]\setminus T} X_j=\otimes_{j\in [n]\setminus T}Y_j$ are first order object by Lemma
\ref{lemma:1ordertensor}. 
 In particular, $A_{Y_T^*}=\{\tilde a_T\}$, where $\tilde a_T:=\otimes_{i\in T} \tilde a_{Y_i}$. It follows that any element in 
\[
A_f(X_1,\dots,X_n)\simeq  A_{Y_T^*\otimes
Y_{[n]\setminus T}}
\]
is of the form $\tilde a_T\otimes \omega$ for some state $\omega\in A_{Y_{[n]\setminus T}}$.
This corresponds to a channel in
$Y_T\multimap Y_{[n]\setminus T}$ that maps any state in $A_{Y_T}$ to the fixed
state $\omega$. In this way, we can see that the type function 
$p_T$ describes the type of replacement channels with set of input indices  in $T$. 

By duality, the elements of $T$ are output indices for the conjugate function $p_T^*$, so
this time we put $X_j=Y_j$ for $j\in T$ and $X_j=\tilde Y_j$ otherwise.
Using Proposition \ref{prop:Xf_const} (ii) and the result for $p_T$, we obtain the isomorphisms
\[
X_{p_T^*}(X_1,\dots,X_n)=X_{p_T}^*(\tilde X_1,\dots,\tilde X_n)\overset{\sigma}{\simeq}
(Y_T^*\otimes 
Y_{[n]\setminus T})^*\overset{\rho}\simeq  Y_{[n]\setminus T}\multimap Y_T,
\]
where $\rho$ is the transposition in $\permut_2$. It follows
that the type function $p^*_T=1-p_T+p_{n}$ corresponds to the type of  channels with
output indices in $T$.

\end{exm}

\begin{lemma}\label{lemma:fh_setting} Let $f\in\Te_n$ and let $O_f=O$,  $I=I_f$. Then
\[
p_I\le f\le p_O^*.
\]

\end{lemma}

\begin{proof} This is obviously true for $n=1$. Indeed, in this case,
$\Te_1=\Fe_1=\{1_1=p_\emptyset,1_1^*=p_{1}\}$. If $f=1_1$, then $O=[1]$, $I=\emptyset$ and 
\[
p_I=p_{\emptyset}=1_1=p_O^*,
\]
the case  $f=p_1$ is obtained by taking complements. Assume that the assertion holds for
$m<n$. Let $f\in \Te_n$ and assume that  $f=g\otimes h$ for some  $g\in
\Te_m$, $h\in \Te_{n-m}$.  By the assumption,
\[
p_{I_g}\otimes p_{I_h}\le g\otimes h\le p^*_{O_g}\otimes p^*_{O_h}\le (p_{O_g}\otimes
p_{O_h})^*,
\]
the last inequality follows from Lemma \ref{lemma:fproduct}. With the decomposition
$[n]=[m][m+1,n]$, we have   
$O_f=O_g\boxplus O_h$, $I_f=I_g\boxplus I_h$, so that by Lemma \ref{lemma:PSPT}, 
$p_{O_f}=p_{O_g}\otimes p_{O_h}$ and
similarly for $p_{I_f}$. Now notice that for any $f\in \Te_n$ we have either $f\approx g\otimes
h$ or $f\approx (g\otimes h)^*$ (see Appendix \ref{app:functions} for the definition of
$\approx$). Since the inequality is easily seen to be preserved by
permutations, and reversed by duality which also switches the input and output sets, the
assertion is proved.

\end{proof}

Combining this with Proposition \ref{prop:Xf_const}, we get the following result
(cf. \cite[Proposition 2]{apadula2024nosignalling}). 

\begin{coro}\label{coro:setting} Let $Y$ be a higher order objects constructed from
first order objects $Y_1,\dots, Y_n$,  $O_Y=O$,
$I_Y=I$.  Then there are
$\sigma_1,\sigma_2\in \permut_n$ such that we have the  morphisms 
\[
(Y_I^*\otimes Y_O)\xrightarrow{\sigma_1} Y
\xrightarrow{\sigma_2} (Y_I\multimap Y_O).
\]

\end{coro}

We also obtain  a simple  way to identify the output indices  of a type
function.

\begin{prop}\label{prop:fh_outputs} For $f\in \Te_n$, $j\in O_f$ if and only if
$f(e^j)=1$, here $e^j=\delta_{1,j}\dots\delta_{n,j}$.

\end{prop}

\begin{proof} Let $j\in O_f$, then by Lemma \ref{lemma:fh_setting}, $p_{I_f}(e^j)=1\le
f(e^j)$, so that $f(e^j)=1$. Conversely, if $f(e^j)=1$, then by the other inequality in
lemma \ref{lemma:fh_setting}, $p_{O_f}(e^j)=0$, whence $j\in O_f$.

\end{proof}

\begin{exm}\label{exm:T2} The type functions for $n=2$ are given as (writing $\bar u=1-u$
for $u\in \{0,1\}$, and $s=s_1s_2$):
\[
1_2(s)=1,\quad p_{2}(s)=\bar s_1\bar s_2,\quad p_{\{1\}}(s)=\bar s_1, \quad 
p_{\{1\}}^*(s)= 1-\bar s_1+\bar
s_1\bar s_2,
\]
and functions  obtained from these by permutation, which gives 6 different elements.
We have seen in Appendix \ref{app:functions} that  $\Fe_n$ has $2^{2^n-1}$ elements, so that $\Fe_2$ has 8
elements in total. The two remaining functions  are
\[
g(s)=1-\bar s_1-\bar s_2+2\bar s_1\bar s_2,\qquad g^*(s)=\bar s_1+\bar s_2-\bar s_1\bar s_2.
\]
It can be  checked directly from Lemma \ref{lemma:fh_setting} and
Proposition \ref{prop:fh_outputs} that $g$ is not a type function. Indeed, if $g\in \Te_2$, we would have $O_g=\emptyset$, so that
$p_{2}\le g\le p_\emptyset^*=p_{2}$, which is obviously not the case. Clearly, also the
complement $g^*\notin \Te_2$. Notice also that $g^*=p_{\{1\}}\vee p_{\{2\}}$, so that
$\Te_2$ is not a lattice. 

\end{exm}

Since $\Fe_2$ can be identified as a sublattice in $\Fe_n$ for
all $n\ge 2$ as $\Fe_2\ni f\mapsto f\otimes 1_{n-2}\in \Fe_n$, the above example shows
that  $\Te_n$, $n\ge 2$ is a  subposet in the distributive lattice 
$\Fe_n$ but not a sublattice,  so that for $f_1,f_2\in
\Te_n$, none of $f_1\wedge f_2$ or $f_1\vee f_2$ has to be a type function.
Nevertheless, we have by the above results that all type functions with the same output
indices are contained in the interval $p_I\le f\le p_O^*$, which is a distributive
lattice. Elements of such an interval  will be called {subtypes}. It is easily seen
that for $n=2$ all subtypes are type functions, but it is not difficult to find a subtype
for $n=3$ which is not in $\Te_3$  (see also Example \ref{exm:T3sub}
below). The objects corresponding to
subtypes are not necessarily
higher order objects, but are embedded in  $Y_I\multimap Y_O$ and contain the replacement
channels. If $f_1$ and $f_2$ have the same output set, then  $f_1\vee f_2$ and $f_1\wedge
f_2$ are subtypes. {By the remarks below Proposition \ref{prop:Xf_const}, the corresponding objects can be
obtained by pushouts resp. pullbacks of the higher order objects corresponding to $f_1$
and $f_2$.}
 A relevant example of a subtype is shown in Section \ref{sec:examples}
(Example \ref{exm:subtype_afp}). An interesting  non-example is given
below.

\begin{exm}[Bistochastic channels] \label{exm:T2_cont}
As we have seen in Example  \ref{exm:T2}, the function $g=p^*_{\{1\}}\wedge p^*_{\{2\}}$
is not a subtype, but nevertheless it corresponds to an interesting set of channels.
Indeed,  $p_{\{1\}}^*$ describes the type of channels with output space labeled by
1 and input space  by 2, while $p_{\{2\}}^*$ has the input and output space reversed. The function $g$ then corresponds to their
intersection, which is precisely the set of bistochastic (or unital) channels. This is not
a (sub)type since the input and output spaces are not fixed. Nevertheless, using the
structure of $\Fe_n$ and the construction in Theorem \ref{thm:boolean}, one can
characterize higher order maps over bistochastic channels. It was shown in
\cite{chiribella2022quantum} that the
bistochastic channels can be interpreted as bidirectional with respect to the time
direction, and higher order maps over them contain operations with indefinite time
direction, such as the quantum time flip.

\end{exm}

\subsection{The poset $\Pe_f$}


By Theorem \ref{thm:mobius}, any boolean function has a unique expression of the form
\[
f=\sum_{T\subseteq [n]} \hat f_Tp_T,
\]
where $\hat f$ is the M\"obius transform of $f$. Using this, we introduce a poset related
to $f$, which will be useful for description of the structure of $f$. We will need the
definitions and basic results in Appendix \ref{app:poset}.

Let $\mathcal P_f$ be the subposet in the
distributive lattice $2^n$,  of elements such that
$\hat f_T\ne 0$. The main result of this paragraph is that any type function $f\in \Te_n$ is fully determined by $\Pe_f$.  
We first need to show how some of the operations on type functions are reflected on  $\Pe_f$.

\begin{lemma}\label{lemma:Pf} Let $f\in \Te_n$.
\begin{enumerate}

\item[(i)] If $\sigma\in \permut_n$, then $S\mapsto \sigma^{-1}(S)$ is an isomorphism of
$\Pe_{f}$ onto $\Pe_{f\circ\sigma}$.
\item[(ii)] For $g\in \Te_m$ and the decomposition $[n+m]=[n]\boxplus [m]$, we have
$\Pe_{f\otimes g}\simeq \Pe_{f}\times \Pe_g$, with the isomorphism given by $(S,T)\mapsto
S\boxplus T$ and 
\[
(\widehat{f\otimes g})_{(S,T)} =\hat f_S\hat g_T.
\]

\end{enumerate}

\end{lemma}

\begin{proof} The statement is proved using Lemma \ref{lemma:PSPT}.  We have
\[
f\circ \sigma=\sum_{S\subseteq [n]} \hat f_S p_S\circ\sigma=\sum_{S\subseteq [n]} \hat
f_Sp_{\sigma^{-1}(S)}=\sum_{S\subseteq [n]} \hat
f_{\sigma(S)}p_{S}.
\]
The statement (i) now follows by uniqueness of the M\"obius transform. Similarly, for
$s=s^1s^2$,
\begin{align*}
f\otimes g(s)&=f(s^1)g(s^2)=\sum_{S\subseteq [n]}\sum_{T\subseteq [m]} \hat f_S\hat
g_Tp_S(s^1)p_T(s^2)=\sum_{S\subseteq [n]}\sum_{T\subseteq [m]} \hat f_S\hat
g_T(p_S\otimes p_T)(s)\\
&= \sum_{S\subseteq [n]}\sum_{T\subseteq [m]} \hat f_S\hat
g_T(p_{S\boxplus T})(s).
\end{align*}
This proves (ii).

\end{proof}

\begin{theorem}\label{thm:graded} Let $f\in \Te_n$, then $\mathcal P_f$ is a graded poset
with even rank $r(f):=r(\Pe_f)\le n$.  Moreover, we have
\[
f=\sum_{S\in \mathcal P_f}(-1)^{\rho_f(S)}p_S,
\]
where $\rho_f$ is the rank function of $\Pe_f$.

\end{theorem}

\begin{proof} We will proceed by induction on $n$. Assume that  $n=1$. Then  $2^n=\{\emptyset, [1]\}$ and
$\Te_1=\{1_1,p_1\}$. For both type functions, $\mathcal P_f$ is a singleton, which 
is clearly a graded poset, with rank $k=0$ and trivial rank function $\rho_f\equiv 0$.  We have
\[
1_1=p_\emptyset=(-1)^{\rho(\emptyset)}p_\emptyset.
\]
The statement for $f=p_1$ follows by duality.  

Assume next that the statement holds for all $m<n$ and let $f\in \Te_n$. By construction,
it is enough to show that the property is invariant under permutations and complement and
that it holds for any $f$ of the form 
$f=f_1\otimes f_2$ for  type functions $f_1\in \Te_{n_1}$, $f_2\in \Te_{n_2}$.
So assume $f$ has the desired property and let $\sigma\in \permut_n$. 
It is clear by the
isomorphism in Lemma \ref{lemma:Pf} (i) that $\Pe_{f\circ\sigma}$ is a graded poset as
well, with the same even rank as $f$ and rank function
$\rho_{f\circ\sigma}=\rho_f\circ \sigma$. Then
\[
f\circ \sigma=\sum_{S\subseteq [n]} (-1)^{\rho_f(S)}p_S\circ\sigma=\sum_{S\subseteq [n]}
(-1)^{\rho_f\circ \sigma(S)}p_S.
\]
Further, assume that we have
\begin{align}
f^*&=1-f+p_n=(1-\hat f_\emptyset)p_\emptyset -
\sum_{\emptyset, [n]\ne S\subseteq [n]} \hat f_Sp_S+(1-\hat f_{[n]})p_n\notag\\
&=(1-\hat f_\emptyset)1 -\sum_{\substack{S\in \mathcal P_f\\ \emptyset \ne S,
[n]\ne S}}
(-1)^{\rho_f(S)}p_S+(1- \hat f_{[n]})p_n.\label{eq:dual_rank}
\end{align}
If $\emptyset \in \Pe_f$, then $\emptyset$ is the least element of $\Pe_f$, so that 
$\rho_f(\emptyset)=0$ and therefore $\hat f_\emptyset =
(-1)^0=1$. Similarly, if $[n]\in \Pe_f$, then $[n]$ is the largest element in $\Pe_f$,
hence it is the last element in any maximal chain. It follows that $\rho_f([n])=r(f)$ and hence
$\hat f_{[n]}=(-1)^{r(f)}=1$ (since the rank $r(f)$  is even). 
Therefore the equality \eqref{eq:dual_rank} implies that $\mathcal P_{f^*}$ differs from $\mathcal P_f$ only in the bottom  and
top elements:  $\emptyset \in \mathcal P_f$ iff  $\emptyset \notin \mathcal P_{f^*}$
and $[n] \in \mathcal P_f$ iff  $[n] \notin \mathcal P_{f^*}$. It follows that $\mathcal
P_{f^*}$ is graded as well, with rank  equal to $r(f)-2$, $r(f)$ or $r(f)+2$, which in any case
is even. Furthermore,  this also implies 
that for all $S\in \Pe_f$, $S\notin \{\emptyset, [n]\}$, we
have  $\rho_{f^*}(S)=\rho_f(S)\pm 1$, according to whether $\emptyset$ was added or removed. The
statement now follows from \eqref{eq:dual_rank}.

Now assume that $f=f_1\otimes f_2$. By the induction assumption, both $\Pe_{f_1}$ and $\Pe_{f_2}$ are graded posets. By Lemma
\ref{lemma:Pf}, $\Pe_f\simeq \Pe_{f_1}\times \Pe_{f_2}$, so that $\Pe_f$  is a graded poset as well, with 
rank function $\rho_f(S,T)=\rho_{f_1}(S)+\rho_{f_2}(T)$ and rank $r(f)=r(f_1)+r(f_2)$.
 By Lemma \ref{lemma:Pf} (ii), we get 
\[
f=\sum_{S\subseteq [n_1], T\subseteq [n_2]} (\widehat {f_1})_S(\widehat
{f_2})_T p_S\otimes p_T=
\sum_{S\subseteq [n_1], T\subseteq [n_2]}(-1)^{\rho_{f_1}(S)+\rho_{f_2}(T)}p_{S\boxplus T}.
\]
This finishes the proof.

\end{proof}

\begin{remark}\label{remark:n} Notice that we need to assume $n$ to be known. Indeed, for any $m$ and
$f$, $\Pe_f$ and $\Pe_{f\otimes 1_m}$ are the same, but the  two type functions are
different. In particular, the corresponding constructions of higher order
objects are different. See also the examples in Section \ref{sec:examples}.

\end{remark}

In the course of the above proof, we have also shown the following.
\begin{coro}\label{coro:pf_dual}  Let $f\in \Te_n$. Then 
\[
\Pe_{f^*} = \Pe_{f} \triangle
\{\emptyset,[n]\},
\]
here $\triangle$ denotes the symmetric difference of two sets: $A\triangle B=(A\cup B)\setminus
(A\cap B)$.
\end{coro}

\subsection{Labelled Hasse diagrams}
\label{sec:labelled}

We  introduce {labels} for the elements of $\Pe_f$ in the following way. 
 For  $S\in \Pe_f$, put
\[
L_S:=\{i\in [n]\ : \ i\in S,\ \forall S'\subsetneq S, i\notin S'\}.
\]
In other words, $i$ is a label for $S$ if $S$ is a minimal element in the subposet of
elements containing $i$ in  $\Pe_f$. 
We will use the notation $L_{S,f}$ if the function $f$ has to be
specified. It is easily seen that for $S\in \mathrm{Min}(\Pe_f)$, $L_S=S$ and for any
$S\in \Pe_f$,  $S=\cup_{S'\subseteq S} L_{S'}$.
It follows that $f\in \Te_n$ (with known $n$) is fully determined by the order relation on $\Pe_f$ and the
label sets. All the information about $f$ can be therefore obtained from the {labelled
Hasse diagram} of $\Pe_f$. 
Some examples  of type functions and their labelled Hasse diagrams  will be
given in Section \ref{sec:examples}.

Let us denote
\[
I_f^F:=\cap_{S\in \mathrm{Min}(\Pe_f)} L_S, \qquad O_f^F:=[n]\setminus
\cup_{S\in \Pe_f}L_S.
\]
It is easily checked by Proposition \ref{prop:fh_outputs} that any  $i\in O_f^F$ is an
output index, since in this case  we have
$f(e^i)=f(\theta_n)=1$. Such elements will be called
the {free outputs} of $f$. If $f$ has some free outputs, then necessarily $[n]\notin
\Pe_f$. Similarly, any  $j\in I_f^F$ is an input of $f$, since $j$ must be
contained in any $T\in \Pe_f$, so that $p_T(e^j)=0$ for all $T\in \Pe_f$ and consequently
$f(e^j)=0$. Such elements will be called {free
inputs} of $f$. The elements of $I_f^F\cup O_f^F$ will be called {free indices} of
$f$.

It is clear that we have
\begin{equation}\label{eq:nofree}
f\approx p_k\otimes g\otimes 1_l
\end{equation}
where
$k=|I_f^F|$, $l=|O_f^F|$  and $g\in \Te_{n-k-l}$ has no free indices. Note that we may
have $k=0$, in which case $f$ has no free inputs and $p_k=1_\varepsilon$ is trivial,
similarly for $l$ (see the paragraph below Proposition \ref{prop:type_min}). As posets,
$\Pe_f\simeq \Pe_g$,  with labels
\[
L_{S,f}=\begin{dcases} \sigma(L_{S,g}), &\text{if } S\notin \mathrm{Min}(\Pe_f)\\
\sigma(L_{S,g})\cup I^F_f, &\text {otherwise},
\end{dcases}
\]
for some $\sigma\in \permut_n$. Clearly, $n-k-l$ has to be specified for $g$.

The  two distinguished elements $\emptyset$ and $[n]$, if present in $\Pe_f$, can be easily recognized
from its structure as a labelled poset. Indeed, $\emptyset\in \Pe_f$ if and only if
$\Pe_f$ has the  smallest element and it  has an empty label. 
Similarly, $[n]\in \Pe_f$ if and
only if $\Pe_f$ has  the largest element and $\cup_{S\in \Pe_f}L_S=[n]$. The basic operations on
type functions are obtained as follows. 

\begin{coro}\label{coro:Pf} Let $f\in \Te_n$, $g\in \Te_m$. Then
\begin{enumerate}
\item[(i)] For $\sigma\in \permut_n$,  $\Pe_{f\circ\sigma}\simeq \Pe_f$, with the labels
changed as $L_S\mapsto \sigma^{-1}(L_S)$.

\item[(ii)] If $[n]\in \Pe_{f^*}$, then  $L_{[n],f^*}=O_f^F$.  All other elements and labels remain the same. 
\item[(iii)] Assume the decomposition $[n+m]=[n]\boxplus[m]$. Then 
$\Pe_{f\otimes g}\simeq \Pe_f\times \Pe_g$, with label sets
\[
L_{(S,T)}=\begin{dcases} L_S\cup (n+L_T), & \text{if } S\in \mathrm{Min}(\Pe_f),\ T\in
\mathrm{Min}(\Pe_g)\\
L_S, & \text{if } S\notin \mathrm{Min}(\Pe_f),\ T\in
\mathrm{Min}(\Pe_g)\\
n+ L_T, &\text{if }S \in \mathrm{Min}(\Pe_f),\ T\notin
\mathrm{Min}(\Pe_g)\\
\emptyset, & \text{otherwise}.
\end{dcases}
\]

\end{enumerate}

\end{coro}

\begin{proof} We only need to prove the statements on the label sets. This is quite clear
in (i). In
(ii), if $[n]\in \Pe_{f^*}$, then the only new indices  not appearing below $[n]$ can be
the free outputs of $f$. In (iii), assume that $i\in L_{(S,T)}$, then $S\boxplus T$ must be a minimal
element in $\Pe_{f\otimes g}$ containing $i$. Hence, either $i\in S$ or $i\in n+T$. In the
first case, $i\in (S'\boxplus T')\le (S\boxplus T)$ whenever  $i\in S'\le S$ and $T'\le T$, so we must have 
$i\in L_S$ and $T\in \mathrm{Min}(\Pe_g)$. Similarly, for $i\in n+T$, we get $i\in n+L_T$
and $S\in \mathrm{Min}(\Pe_f)$.

\end{proof}

We next show that the input and output sets of $f\in \Te_n$ can be easily recognized from
the labels in $\Pe_f$. 

\begin{prop}\label{prop:pfinput} Let $f\in \Te_n$ and $i\in [n]$. Then
\begin{enumerate}
\item[(i)] All $S\in \Pe_f$ such that $i\in L_S$  have the same rank,  which will be
denoted by $r_f(i)$. If  $i\in O_f^F$,   we put $r_f(i):=r(f)+1$. 
\item [(ii)] $i\in O_f$ if and only if $r_f(i)$ is odd.

\end{enumerate}
\end{prop}

\begin{proof} As before, we proceed by induction on $n$. Both assertions are quite trivial for $n=1$,
so assume the statements hold for $m<n$. It is easily seen that the properties are
invariant under permutations. Assume (i) and (ii) hold for $f\in \Te_n$ and consider
$f^*$. If $i\in L_{[n],f^*}$, then $i$ cannot be contained in the label set of any other
element, so (i) is true. Also, by Corollary \ref{coro:Pf}, $L_{[n],f^*}=O_f^F$, so that
$i$ is an input of $f^*$. Since $[n]$ is the largest element of $\Pe_f$, 
$\rho_f([n])=r(f)$ is even, so that (ii) holds as well. By duality, both statements hold
if $i\in O_{f^*}^F$. In all other cases, $i\in L_{S,f}$ if and only if $i\in L_{S,f^*}$, so
(i) is true for $f^*$. By the
proof of Proposition \ref{thm:graded} we have  
$\rho_{f^*}(S)=\rho_f(S)\pm 1$ for any $S$, depending only on the fact whether $\emptyset
\in \Pe_f$. This implies that (i) and (ii)  are preserved by complementation.

 It is now enough to assume that
$f=g\otimes h$ for some $g\in \Te_m$ and $h\in \Te_{n-m}$. 
Suppose without loss of generality that $i\in [m]$, then $i\in L_{S\boxplus T, f}$ if and
only if $i\in L_{S,g}$ and $T\in \mathrm{Min}(h)$.  Since then $\rho_h(T)=0$, we have
by the induction assumption
\[
\rho_f(S\boxplus T)=\rho_g(S)+\rho_h(T)=\rho_g(S)= r_g(i).
\]
The statement (ii) follows from the fact that $i\in O_f$ if and only if $i\in O_g$.

\end{proof}
\subsubsection{Examples}\label{sec:examples} 

In this paragraph, we will  give some examples of simple type functions, their corresponding Hasse
diagrams and types of higher order maps.

\begin{exm}[Free indices]\label{exm:free} Let  $f\in \Te_n$ and let $I^F$/$O^F$ be the
free inputs/outputs of $f$. Up to a permutation, we may assume that $I^F=[k]$ and
$O^F=[n-l,n]$ and then $f=p_k\otimes g\otimes 1_l$ for some type function $g\in
\Te_{n-k-l}$. It follows that 
\[
S_f=(\otimes_{i=1}^k \tilde a_i)\otimes  S_g\otimes
(\otimes_{i=n-l+1}^n V_l).
\]
We clearly have $O_f=O^F\boxplus O$ and $I_f=I^F\boxplus I$, where $O=O_g$ and
$I=I_g$.

Assume that all the basic first order objects are the usual quantum state spaces
(elementary systems), so that  $Y_i=(M^h_{n_i},
\{\Tr[a]=1\})$ and $\tilde a_i=E_{n_i}$.  Let $K=\Pi_{i=1}^k n_i$, $L=\Pi_{i=n-l+1}^n n_i$. Then we have
$A_f=E_K\otimes A_g\otimes A_{\Se_L}$,  (see Example \ref{exm:quantum_maps} for the
notation). Any channel of type $f$ has the Choi matrix of the form
\[
C=E_K\otimes \sum_j\alpha_j C_j \otimes \rho_j =E_K\otimes D,
\]
where $D=\sum_j\alpha_j C_j \otimes \rho_j$ for some channels $C_j$ of type $g$ and states $\rho_j$ on $M_L^h$, with
$\sum_j\alpha_j=1$, so that $\Tr_O D=E_I\otimes \rho$, where $\Tr[\rho]=1$. By the characterization of quantum combs in
\cite{chiribella2009theoretical}, we see that $C$ has the form 
\begin{center}
\includegraphics[scale=0.8]{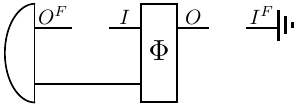}
\end{center}
If $g$ is the type of all channels between elementary systems $I\to O$, then clearly
$\Phi$ is some channel  $I\otimes R\to O$, for some ancillary system $R$. In general, we
may always suppose that $g=(h\to k)=(h\otimes k^*)^*$ for some types $h$ and $k$, and it
can be proved by induction that $\Phi$ is a map of type $h\otimes 1\to k$, where the function 1
represents the ancillary input of $\Phi$. Using the currying equality in Lemma
\ref{lemma:combs} (i), we obtain that $(h\otimes 1\to k)\approx (1\to g)$.

\end{exm}

To describe further examples, we introduce the following notations: for $n\in \mathbb N$,
put
\begin{equation}\label{eq:gammas}
\gamma_n:= \begin{dcases} \sum_{j=0}^n (-1)^jp_{[j]}, & \text{if $n$ is even}\\
\sum_{j=1}^n (-1)^{j-1}p_{[j]} & \text{if $n$ is odd}
\end{dcases}
\end{equation}
(here we put  $[0]=\emptyset$). Note that $\gamma_n\in \Fe_n$ and  $\gamma_n^*=\gamma_{n-1}\otimes 1_1$ for $n\in
\mathbb N$, where we put $\gamma_0=1_\varepsilon$. For each $n$, the poset $\Pe_{\gamma_n}$ is a chain of even
length and we will see in Section  \ref{sec:chains} below that $\gamma_n\in \Te_n$. 

\begin{exm}[$\Te_2$] \label{exm:T2hasse} The labelled Hasse diagrams of elements in $\Te_2$
described in Example \ref{exm:T2} are up
to permutations as
follows:
\begin{center}
\includegraphics[scale=0.8]{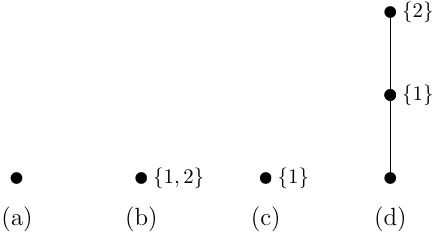}
\end{center}

The diagram (a) corresponds to the function $1=p_\emptyset$, so the diagram
has only one element with no label. This describes  first order objects $Y_1\otimes Y_2$
(states). This type has two free outputs and no inputs.  Diagram (b) corresponds to $p_2=p_{\{1,2\}}$, which is 
$Y_1^*\otimes Y_2^*$ (the unit functional). Here we have two free inputs and no outputs.
The function in (c) is $p_{\{1\}}=p_{1}\otimes 1_1$, which describes replacement
channels $Y_1\to Y_2$. In this case we have one free input and one free output.
The diagram in (d) corresponds to the function $\gamma_2=p_{\{1\}}^*$. As we
have seen in Example \ref{exm:type_channels}, this function is related to channels
$Y_2\multimap Y_1$.
This type has no free indices, note that the input/output indices follow the results of
Proposition \ref{prop:pfinput}.

We now depict  maps of the above types in a  form more familiar in   quantum
theory:
\begin{center}
\begin{minipage}{0.25\textwidth}
\centering
\includegraphics[scale=0.8]{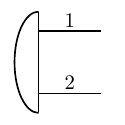}

(a)
\end{minipage}%
\begin{minipage}{0.25\textwidth}
\centering
\includegraphics[scale=0.8]{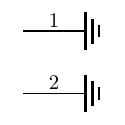}

(b)
\end{minipage}%
\begin{minipage}{0.25\textwidth}
\centering
\includegraphics[scale=0.8]{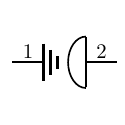}

(c)
\end{minipage}%
\begin{minipage}{0.25\textwidth}
\centering
\includegraphics[scale=0.8]{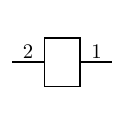}

(d)
\end{minipage}
\end{center}

\end{exm}

\begin{exm}[$\mathcal T_3$] \label{exm:T3}
Note that by the above example, all the posets corresponding to  elements in $\Te_2$ are chains. 
This is also true for $n=3$. Indeed, up to a permutation that does not change the chain structure, 
any $f\in \Te_3$ is either a product of two elements $g\in \Te_2$ and $h\in \Te_1$,
or the dual of such a  product. Since $g$ must be a chain and $|\Pe_h|=1$, their product
must be a chain as well. Taking the dual of a chain only adds/removes the least/largest
elements, so the dual of a chain must be a chain as well. All the diagrams depicted in
(a)-(d) above describe functions  contained in $\Te_3$,  with  an additional free
output 3. The corresponding maps have the form
\begin{center}
\begin{minipage}{0.25\textwidth}
\centering
\includegraphics[scale=0.8]{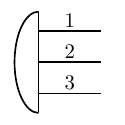}

(a)
\end{minipage}%
\begin{minipage}{0.25\textwidth}
\centering
\includegraphics[scale=0.8]{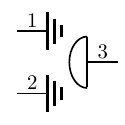}

(b)
\end{minipage}%
\begin{minipage}{0.25\textwidth}
\centering
\includegraphics[scale=0.8]{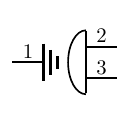}

(c)
\end{minipage}%
\begin{minipage}{0.25\textwidth}
\centering
\includegraphics[scale=0.8]{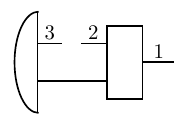}

(d)
\end{minipage}
\end{center}
The only other
elements of $\Te_3$ are (up to permutations):
\begin{center}
\includegraphics[scale=0.8]{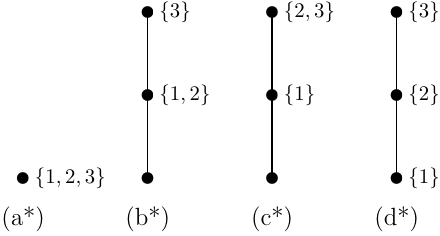}
\end{center}
Notice that in $\Te_3$, the function with diagram in (a*) is the conjugate of  the
function with diagram in  (a), etc. The diagram in (d*) corresponds to $\gamma_3$,
 which is a channel with an additional free input 1:
\begin{center}
\begin{minipage}{0.25\textwidth}
\centering
\includegraphics[scale=0.8]{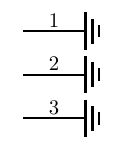}

(a*)
\end{minipage}%
\begin{minipage}{0.25\textwidth}
\centering
\includegraphics[scale=0.8]{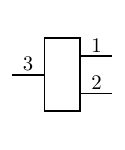}

(b*)
\end{minipage}%
\begin{minipage}{0.25\textwidth}
\centering
\includegraphics[scale=0.8]{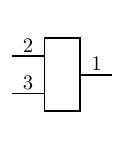}

(c*)
\end{minipage}%
\begin{minipage}{0.25\textwidth}
\centering
\includegraphics[scale=0.8]{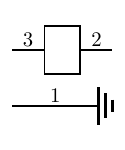}

(d*)
\end{minipage}
\end{center}

\end{exm}

\begin{exm}[$\Te_3$ subtypes]\label{exm:T3sub} We will next find all subtypes for $n=3$, given by the input and output
sets $I=\{1,2\}$ and $O=\{3\}$, all the other nontrivial subtypes can be found by
permutations and duals. The subtypes are given by the condition
\[
p_{\{1,2\}}\le f \le p_{\{3\}}^*.
\]
The function of the left hand side is equal to 1 on all strings  $s$ with $s_1=s_2=0$; the
functions on the right are equal to 0 whenever $s\ne\theta_3$ and $s_3=0$. This leaves us
with only 3 strings where the values of $f$ are not fixed: 101, 011 and 111, and we have
$p_{\{1,2\}}(s)=0$, $p^*_{\{3\}}(s)=1$ for all these strings.
This means that there are 8 subtypes in this case. All types in this interval are
described by the chains in  Examples \ref{exm:T2} and
\ref{exm:T3}.  The two extremal cases correspond to
the diagram (b) in Example \ref{exm:T2}, resp. diagram (c*) in Example \ref{exm:T3} (with
$\{3\}$ in the middle and $\{1,2\}$ on the top). The only remaining types with the given
input/output sets are two chains of the form (d*), with $\{3\}$ in the middle, let us
denote these functions by $f_1$ and $f_2$. Including infima and suprema, we obtain the
first 5 subtypes in  following table (note that $f_1\wedge f_2=p_{\{1,2\}}$):

\begin{center}
\begin{tabular}{|c|c|c|c|c|c|c|c|c|}
\hline
$s$  & $p_{\{1,2\}}$ & $p_{\{3\}}^*$ & $f_1$ & $f_2$ & $f_1\vee f_2$ & $g_1$ & $g_2$ & $g_3$\\
 \hline 
 101 & 0 & 1 &1 &0 &1 &0 &0 &1 \\ 
 011 & 0 & 1 &0 &1 &1 &0 &1 & 0\\ 
 111 & 0 & 1 &0 &0 &0 &1 &1 & 1\\
 \hline
\end{tabular}
\end{center}
Hence there are only four types in this interval, and only one additional subtype that
is obtained by taking joins and meets of chains.  The remaining subtypes $g_1,g_2,g_3$
have no clear interpretation and the corresponding subspaces $S_{g_i}$ depend on
the choice of the element $a_{Y_i}\in A_{Y_i}$, $i=1,2$, e.g.
\[
S_{g_1}=\tilde a_{Y_1}\otimes \tilde a_{Y_2}\otimes
V_{Y_3}+\{a_{Y_1}\}^\perp\otimes \{a_{Y_2}\}^\perp\otimes \{\tilde a_{Y_3}\}^\perp.
\]

\end{exm}

\begin{exm}[$\Te_4$]\label{exm:T4} The only elements in $\Te_4$ such that the posets are not
isomorphic to any of those given above have the diagrams:
\begin{center}
\includegraphics[scale=0.8]{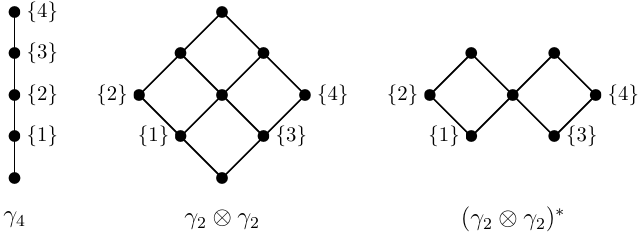}
\end{center}
The corresponding type functions and related higher order objects are: 
\begin{enumerate}
\item $\gamma_4(s)=1-\bar s_1+\bar s_1\bar s_2-\bar s_1\bar s_2\bar s_3+\bar s_1\bar
s_2\bar s_3\bar s_4$, related to 2-combs (superchannels), see Proposition
\ref{prop:chains_combs} below,
\item $(\gamma_2\otimes\gamma_2)(s)=1-\bar s_1-\bar s_3+\bar s_1\bar s_2 + \bar s_1\bar
s_3 +\bar s_3\bar s_4-\bar s_1\bar s_2\bar s_3-\bar s_1\bar s_3\bar s_4 +\bar s_1\bar
s_2\bar s_3\bar s_4$ related to no-signalling bipartite channels,
\item $(\gamma_2\otimes \gamma_2)^*(s)=\bar s_1+\bar s_3-\bar s_1\bar s_2 - \bar s_1\bar
s_3 -\bar s_3\bar s_4+\bar s_1\bar s_2\bar s_3+\bar s_1\bar s_3\bar s_4$ related to process
matrices (Example \ref{exm:ns_pm}).

\end{enumerate}

\end{exm}

\begin{exm}[Adapters] \label{exm:subtype_afp} An adapter is a  map  that transforms process matrices to process matrices.
Let $g=(\gamma_2\otimes \gamma_2)^*\in \Te_4$ be the type functions of process matrices,
then the type function for adapters is  $a=g\to g=(g\otimes g^*)^*\in \Te_8$. The poset
$\Pe_a$ is already rather big and complicated, we will show the diagram of the corresponding reduced poset
in Example \ref{exm:adapter} (Section \ref{sec:pf0}). In the quantum case, this type was
described in \cite{milz2022resource} as the set of operations in the resource theory of
causal connection. As a possible choice of free operations, the
set of admissible free preserving (AFP) adapters  was suggested. If we denote the set of
inputs of the source/target process matrix by $I$/$I'$, the free preserving adapters are
those of the type $f=p_I\to p_{I'}=(p_I\otimes p_{I'}^*)^*$. The AFP adapters then
correspond to the subtype $a\wedge f$. 

\end{exm}

\subsection{Chains and combs}
\label{sec:chains}
 We have seen that for some type functions the poset $\Pe_f$ is a chain, which is also a
 basic example of a graded poset. A chain in $2^n$ has the form  $\Ce=\{S_1\subsetneq S_2\subsetneq \dots \subsetneq
S_N\}$, $S_i\subseteq [n]$. Note that the length of the chain $\Ce$ is $N-1$.
It is clear that  $\Ce$ is graded with rank $N-1$
and rank function $\rho(S_i)=i-1$. 

\begin{prop}\label{prop:chains} For a chain   $\Ce=\{S_1\subsetneq S_2\subsetneq \dots \subsetneq
S_N\}$, the function  
\[
f=f_\Ce:=\sum_{i=1}^N (-1)^{i-1} p_{S_i}
\]
is a type function if and only if $N$ odd. In this case, we say that $f$ is a chain type.

\end{prop}

\begin{proof}
By Proposition \ref{thm:graded}, if $f\in \Te_n$, then the rank of $f$ must be even, so
that $N$ must be odd. 
We will show that the converse is true. We proceed by induction on $N$. For $N=1$, we have
$f=p_{S_1}\in \Te_n$. Assume that the statement holds for all odd numbers $M<N$ and let
$\Ce$ be a chain as above. It is easily checked that   
\[
f\approx  p_{n_1}\otimes g\otimes 1_{n-n_N},
\]
where $n_i:=|S_i|$ and  $g\in
\Fe_{n_N-n_1}$ is the function for a chain $\Ce'$  in $2^{n_N-n_1}$ of the form $\Ce':=\{\emptyset\subsetneq S'_2\subsetneq \dots
\subsetneq [n_N-n_1]\}$. Since $f$ is a type function if $g$ is, this shows that we may assume that 
the chain $\Ce$ contains $\emptyset$ and $[n]$.  But then 
\[
f=1+\sum_{j=2}^{N-1} (-1)^{j-1}p_{S_j}+ p_{n} 
\]
and
\[
f^*=1-f+p_{n}=\sum_{j=1}^{N-2} (-1)^{j-1}p_{S_{j+1}},
\]
By the induction assumption $f^*\in \Te_n$, hence also $f=f^{**}\in
\Te_n$.
\end{proof}

Let $f\in \Te_n$ be a chain type and let $\Pe_f=\{S_1\subsetneq \dots \subsetneq
S_N\}$ be the corresponding chain. There is a decomposition of $[n]$ given as
\begin{equation}\label{eq:chain_decomp}
T_0:=S_1,\quad T_j:=S_{j+1}\setminus S_{j},\ j=1,\dots,N-1, \quad T_{N}:=[n]\setminus S_N.
\end{equation}
It is clear that the label sets are given as $L_{S_j}=T_{j-1}$, $j=1,\dots, N$ and
it can be easily seen from Proposition \ref{prop:pfinput} that 
\begin{equation}\label{eq:chain_io}
I_f=\bigcup_{j=0}^{(N-1)/2}{T_{2j}}, \qquad O_f=\bigcup_{j=0}^{(N-1)/2}{T_{2j+1}}\cup
O_f^F,\qquad I_f^F=T_0,\qquad O_f^F=T_{N}
\end{equation}
(note that $N$ must be odd).  As we have seen, $f\approx p_{n_1}\otimes g\otimes 1_{n-n_N}$
and $g$ is a chain type with no free indices. By Proposition \ref{prop:Xf_const}, we have 
for any collection of first order objects
\[
X_f(X_1,\dots,X_n)\overset{\sigma}{\simeq} \tilde X^*_{I_f^F}\otimes
X_g(X_{\sigma^{-1}(1)},\dots,X_{\sigma^{-1}(n_N-n_1)})\otimes X_{O_f^F},
\]
for some $\sigma\in \permut_n$. We will show below that chain types correspond to an
important kind of higher order objects.

\begin{prop}\label{prop:chains_combs}  Let $f\in \Te_n$ be a chain type with
$\Pe_f=\{\emptyset=S_1\subsetneq\dots \subsetneq S_N=[n]\}$, with
label sets  $T_i=L_{S_{i+1}}$,  $i=1,\dots,n$.  Let $Y=X_f(X_1,\dots,X_n)$ for some first order objects $X_1,\dots, X_n$. 
Then for $N\ge 3$, $Y$ is an $(N-1)/2$-comb. 
More precisely, let $Y_1,\dots, Y_n$ be such that $Y_i=X_i$ for $i\in O_f$ and $Y_i=\tilde X_i$ for $i\in I_f$.
Then
\[
Y{\overset{\sigma}{\simeq}}
Y_{T_{N-1}}\multimap (Y_{T_{N-2}}\multimap\dots \multimap (Y_{T_{\frac{N+1}2}}\multimap
Y_{T_{\frac{N-1}2}})\multimap\dots\multimap Y_{T_2})\multimap Y_{T_1} 
 \]
 where we put $Y_T=\otimes_{j\in T} Y_j$.

\end{prop}

\begin{proof} Let $Y_1,\dots,Y_n$ be as assumed, then  by \eqref{eq:chain_io}, 
\[
Y_{T_i}=\begin{dcases} \otimes_{j\in T_i}X_j,  & \text{ if } i\text{ is odd},\\
\otimes_{j\in T_i}\tilde X_i,& \text{ if }i\text{ is even}.
\end{dcases}
\]
We will proceed by induction on $N$.
Let $N=3$, then $f=1-p_{S_2}+p_{n}$, and we see from Example \ref{exm:type_channels}
that  $Y\overset{\sigma}{\simeq} Y_{T_2}\multimap  Y_{T_1}$. 
Assume the assertion is true for  $N-2$.  As in the proof of Proposition \ref{prop:chains}, we see that 
\[
f^*=\sum_{i=1}^{N-2}(-1)^{i-1}p_{S_{i+1}}\approx p_{n_2}\otimes g\otimes 1_{n-n_{N-1}}
\]
where $g\in \Te_{n_{N-1}-n_2}$ is the chain type for  a chain $\{\emptyset\subsetneq
\sigma(S_3\setminus S_2)\subsetneq \dots\subsetneq \sigma(S_{N-1}\setminus
S_{2})=[n_{N-1}-n_2]\}$, for some $\sigma\in \permut_n$ such that $\sigma(S_2)=[n_2]$
and $\sigma(T_{N-1})=\sigma([n]\setminus S_{N-1})=[n-n_{N-1}]$. 
 By Proposition \ref{prop:Xf_const}, we see that 
\[
X_f(X_1,\dots,X_n)=X_{f^*}^*(\tilde X_1,\dots, \tilde X_n)\overset{\sigma}\simeq (Y_{T_{N-1}}\otimes
\tilde X_g\otimes Y^*_{T_1})^*= Y_{T_{N-1}}\multimap \tilde
X_g\multimap Y_{T_1}
\]
where $\tilde X_g=X_g(\tilde X_{\sigma^{-1}(1)},\dots, \tilde X_{\sigma^{-1}(n_{N-1})})$. 
Since $g$ satisfies the induction assumption, and $\tilde {\tilde
X}_i=X_i$, we obtain 
\[
\tilde X_g\overset{\sigma'}{\simeq}
Y_{T_{N-2}}\multimap \dots\multimap (Y_{T_{\frac{N+1}2}}\multimap
Y_{T_{\frac{N-1}2}})\multimap \dots \multimap Y_{T_2},
\]
for some permutation $\sigma'$. Since $Y_{T_j}$ are first order objects, this proves the
result (see Example \ref{exm:chans_combs}).

\end{proof}

\medskip

 Using the above result and Example \ref{exm:free}, we find the quantum comb
corresponding to a general chain of odd length $N$ with possibly nontrivial free inputs
($T_0=S_1\ne \emptyset$) and free outputs ($T_N=[n]\setminus S_N\ne\emptyset$). The chain
and the comb are depicted below. Note that the causal ordering of the spaces in the comb goes down the chain, so
it is  opposite to the order of the chain.

\begin{center}
\begin{minipage}[c]{0.3\textwidth}
\centering
\includegraphics[scale=0.8]{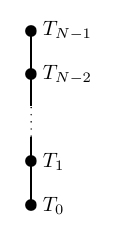}
\end{minipage}
\begin{minipage}[c]{0.5\textwidth}
\centering
\includegraphics[scale=0.8]{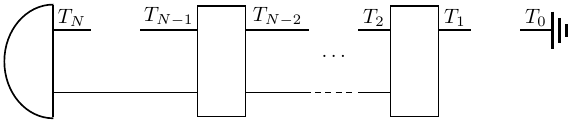}
\end{minipage}
\end{center}

\medskip

\subsection{Connecting chains: the causal product}

It is easy to see that two chains can be appended to create a single chain using the
ordinal sum, and any chain
of more than one elements can be decomposed as an ordinal sum of  chains 
(see Appendix \ref{app:poset} for the definitions). Such operations are trickier for
chain types, since the chains have to be of even length.

To motivate our next definition, let us consider the chain type $\gamma_6$ of length 6. It can
be easily seen that it can be composed of two chain types of length 2 and 4 as follows
\begin{equation}\label{eq:causal_motivate}
\gamma_6(s)=\gamma_2(s_1s_2)-\bar s_1\bar s_2+ \bar s_1\bar
s_2\gamma_4(s_3s_4s_5s_6)=((\gamma_2-p_2)\otimes
1_2+p_2\otimes \gamma_4)(s^1s^2),
\end{equation}
where $s^1=s_1s_2$ and $s^2=s_3\dots s_6$.
As we have noticed above, the causal ordering of the spaces in the corresponding comb goes
down the chain, so that the part of the chain corresponding to $\gamma_2$ cannot influence
the outputs in the part described by $\gamma_4$. Hence we can say that $\gamma_2$ is in the causal
future of $\gamma_4$. We will next extend the expression in
\eqref{eq:causal_motivate}  for
any pair of functions in $\Fe_n$.

For a fixed decomposition $[n]=[n_1]\boxplus[n_2]$ and functions
$f_1\in \Fe_{n_1}$, $f_2\in \Fe_{n_2}$, we define their 
causal product as 
\[
f_1\vartriangleleft f_2:=(f_1-p_{n_1})\otimes 1_{n_2}+p_{n_1}\otimes f_2.
\]
For  $s^1\in \{0,1\}^{n_1}$ and $s^2\in \{0,1\}^{n_2}$, this function acts as
\begin{equation}\label{eq:causal_product}
(f_1\vtl f_2)(s^1s^2)= f_1(s^1)-p_{n_1}(s^1)+p_{n_1}(s^1)f_2(s^2)=\begin{dcases} f_1(s^1), &
\text{ if } s^1\ne \theta_{n_1},\\
   f_2(s^2), & \text{ if } s^1=\theta_{n_1}.
   \end{dcases}
\end{equation}

\begin{remark}  As in the case of chains, we may interpret $f_1\vtl f_2$ as '$f_1$ is in
the causal future of $f_2$'. For first order objects $X_1,\dots,X_n$, the corresponding subspace has the form
\[
S_{f_1\vtl f_2}(X_1,\dots,X_n)= L_{f_1}(X_1,\dots,x_{n_1})\otimes V_{n_2}+a_{n_1}\otimes
S_{f_2}(X_{n_1+1},\dots,X_n), 
\]
with $V_{n_2}=\otimes_{i=1}^{n_2} V_{X_{n_1+i}}$ and $a_{n_1}=\otimes_{j=1}^{n_1}
a_{X_j}$. Notice that in the case of quantum objects, this corresponds to the one-way
signalling composition $\succ$ defined in \cite{hoffreumon2026projective} for the
corresponding orthogonal projections, i.e. denoting by $P_f$ the projection onto $S_f$, we
have $P_{f_1\vtl f_2}=P_{f_1}\succ P_{f_2}$. 

\end{remark}

The following properties are immediate from \eqref{eq:causal_product}. See also
\cite{hoffreumon2026projective} for the corresponding properties of $\succ$.  

\begin{lemma}\label{lemma:causal_product}
Let $f_1,g_1\in \Fe_{n_1}$, $f_2,g_2\in \Fe_{n_2}$. Then $f_1\vartriangleleft f_2\in \Fe_{n_1+n_2}$ and we
have 
\begin{enumerate}
\item[(i)] $(f_1\vtl f_2)^*=f_1^*\vtl f_2^*$,
\item[(ii)]$(f_1\vee g_1)\vtl (f_2\vee g_2)=(f_1\vtl f_2)\vee ( g_1\vtl g_2)=(f_1\vtl
g_2)\vee (g_1\vtl f_2)$,
\item[(iii)] $(f_1\wedge g_1)\vtl (f_2\wedge g_2)=(f_1\vtl f_2)\wedge ( g_1\vtl g_2)=(f_1\vtl
g_2)\wedge (g_1\vtl f_2)$.
\end{enumerate}
Moreover, for any $f_3\in \Fe_{n_3}$, and for the decomposition $[n]=[n_1]\boxplus
[n_2]\boxplus [n_3]$, we have 
\[
(f_1\vtl f_2)\vtl f_3=f_1\vtl (f_2\vtl f_3).
\]
\end{lemma}

We can also combine $f_1$ and $f_2$ in the opposite order:
\[
f_2\vtl f_1: =1_{n_1}\otimes (f_2-p_{n_2})+f_1\otimes p_{n_2},
\]
so that
\begin{equation}\label{eq:causal_product_op}
(f_2\vtl f_1)(s^1s^2)=f_2(s^2)-p_{n_2}(s^2)+p_{n_2}(s^2)f_1(s^1)=\begin{dcases} f_2(s^2), & \text{ if }
s^2\ne \theta_{n_2},\\
   f_1(s^1), & \text{ if } s^2=\theta_{n_2}.
   \end{dcases}
\end{equation}
Of course, this product has similar properties as listed in the above lemma.
To avoid any confusion, we have to bear in mind the fixed decomposition $[n]=[n_1]\boxplus
[n_2]$ and that for the concatenation $s=s^1s^2$, $f_i$ acts on $s^i$. 

\begin{lemma}\label{lemma:causal_tensor} In the situation as above, we have
\[
f_1\otimes f_2 = (f_1\vtl f_2)\wedge (f_2\vtl f_1),\qquad  f_1\parr f_2 = (f_1\vtl f_2)\vee
(f_2\vtl f_1).
\]
\end{lemma}

\begin{proof} The first equality follows again by straightforward computation from \eqref{eq:causal_product}
and \eqref{eq:causal_product_op}: let
$s^1\in \{0,1\}^{n_1}$, $s^2\in \{0,1\}^{n_2}$ and compute
\begin{align*}
(f_1\vtl f_2)\wedge (f_2\vtl
f_1)(s^1s^2)&=\left(f_1(s^1)+p_{n_1}(s^1)(f_2(s^2)-1)\right)\left(f_2(s^2)+p_{n_2}(s^2)(f_1(s^1)-1)\right)\\
=f_1(s^1)f_2(s^2),
\end{align*}
here the last equality follows from the fact that $f_i(s^i)(1-f_i(s^i))=0$ (since $f_i(s^i)\in
\{0,1\}$) and the fact that $p_{n_1}$ is the least element in $\Fe_{n_1}$, so that
$p_{n_1}(s^1)(f_1(s^1)-1)=p_{n_1}(s^1)-p_{n_1}(s^1)=0$. 

 The second equality in the statement follows easily from the first, using
the definition of $\parr$ in \eqref{eq:notations}, Lemma \ref{lemma:causal_product} (i) and
the de Morgan laws in $\Fe_n$. 
\end{proof}

Using the last part of Lemma \ref{lemma:causal_product},
for a decomposition $[n]=\boxplus_i[n_i]$ and $f_i\in \Fe_{n_i}$,  we may define the function $f_1\vtl\dots \vtl f_k\in
\Fe_{n}$. Note that we have for $s=s^1\dots s^k$, 
\begin{align*}
(f_1\vtl \dots\vtl f_k)(s)&=(f_1-p_{n_1})(s^1)+p_{n_1}(s^1)(f_2-p_{n_2})(s^2)+\dots + 
 p_{n-n_k}(s^1\dots s^{k-1})f_k(s^k)\\
&= \begin{dcases} f_1(s_1) & \text{ if } s^1\ne \theta_{n_1}\\
f_2(s^2) & \text{ if } s^1=\theta_{n_1}, s^2\ne \theta_{n_2}\\
\dots & \\
f_k(s^k) & \text{ if } s^1=\theta_{n_1},\dots,  s^{k-1}=\theta_{n_{k-1}}.
\end{dcases}
\end{align*}
For any permutation $\pi\in \permut_k$, we define  $f_{\pi^{-1}(1)}\vtl \dots \vtl
f_{\pi^{-1}(k)}\in \Fe_n$ in an obvious way.

For the smallest and the largest element in $\Fe_n$, the causal product behaves as
follows.
\begin{lemma}\label{lemma:onechain_causal}
Let  $f\in \Fe_{n_1}$ and let $n_2\in \mathbb N$. Then for the decomposition
$[n_1+n_2]=[n_1]\boxplus [n_2]$,  
\[
f\vtl 1_{n_2}= f\otimes 1_{n_2}\le 1_{n_2}\vtl f\approx  (f^*\to 1_{n_2})
\]
and
\[
p_{n_2} \vtl f= f \otimes p_{n_2}\le f\vtl p_{n_2} \approx (1_{n_2}\to f).
\]
In particular,
\[
1_{n_1}\vtl p_{n_2}\approx 1_{n_2}\to 1_{n_1}=1-p_{[n_1]}+p_{n_1+n_2}
\]
is the chain type for $\{\emptyset\subsetneq [n_1]\subsetneq [n_1+n_2]\}$. Similar
properties hold for the decomposition $[n_1+n_2]=[n_2]\boxplus[n_1]$.
\end{lemma}

\begin{proof}
Immediate from the definition of the causal product,  Lemma
\ref{lemma:causal_product} (i) and Lemma \ref{lemma:causal_tensor}.

\end{proof}

Note that if $g=1_k\vtl f$, the elements in the set $[k]$ are output indices of $g$ that
are in the causal future of $f$, so that we may interpret $1_k\vtl f$ as adding a global
future to the maps represented by $f$. Similarly, $f\vtl p_k$ may be interpreted as adding
a global past. Using  Lemma \ref{lemma:onechain_causal} and currying (Lemma
\ref{lemma:combs} (i)), we
obtain
\begin{align*}
1_{k}\vtl f \vtl p_{l}&\approx 1_{l}\to (1_{k}\vtl f)\approx 1_{l}\to (f^*\to
1_{k})\approx f^*\to (1_{l}\to 1_{k}).
\end{align*}

\begin{exm}\label{exm:pm_past_future} Let $f\in \Te_4$ be the process matrix type,
$f=(\gamma_2\otimes\gamma_2)^*$ (see Example \ref{exm:T4}, part 3.). By the above computation,
we obtain
\[
1_{1}\vtl (\gamma_2\otimes \gamma_2)^* \vtl p_{1}\approx (\gamma_2\otimes\gamma_2)\to
\gamma_2,
\]
which describes process matrices with global past and future. In the diagrams
below, $(\gamma_2\otimes \gamma_2)^*$ acts on $s_1\dots s_4$, $1_1$ acts on $s_5$ and $p_1$ acts
on $s_6$:

\begin{center}
\begin{minipage}[b]{0.45\textwidth}
\centering
\includegraphics[scale=0.8]{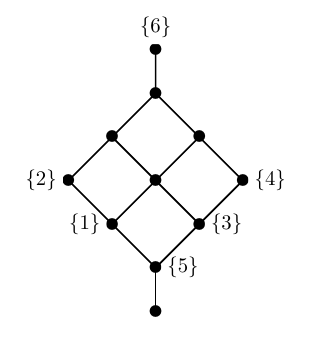}
\vskip -2mm
{\footnotesize $1_1\vtl(\gamma_2\otimes \gamma_2)^*\vtl p_1$}
\end{minipage}%
\begin{minipage}[b]{0.45\textwidth}

\centering
\includegraphics[scale=1]{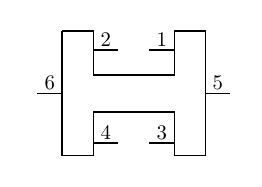}

{\footnotesize Process matrix with global past and future }
\end{minipage}%

\end{center}

\end{exm}

\begin{remark}[Type extensions] \label{rem:extensions} Observe that the causal product 
 $1_1\vtl f$ exactly matches the inductive definition of an  extension of the type $f$ by an elementary 
 type, introduced in \cite[Def. 2]{bisio2019theoretical}. Indeed, the elementary types of
 \cite{bisio2019theoretical} correspond to the (quantum) first order objects, which are 
 obtained from the type function $1_1$.  If $f=1_1$, then $1_1\vtl f=1_1\otimes
 1_1=1_{2}$ is the tensor product of two elementary types, 
and if $f\approx g\to h$ for some type functions $g$
and $h$, then we have 
\[
1_1\vtl f\approx 1_1\vtl(g\to h)\approx (g\otimes h^*)\to 1_1\approx g\to (h^*\to 1_1)\approx
g\to (1_1\vtl h),
\]
where we used Lemma \ref{lemma:onechain_causal} and currying. Similarly, $p_1\vtl
p_1=p_{2}$ and 
\[
f\vtl p_1\approx (g\to h)\vtl p_1\approx 1_1\to (g\to h)\approx g\to (1_1\to h)\approx
g\to (h\vtl p_1),
\]
so $f\vtl p_1$ can be understood as a dual extension. We can now prove that the
higher order maps  are complete, in the following sense: For any $f\in \Te_n$ and $g\in
\Te_k$, any collections  of first order objects $X_0,X_1,\dots,X_n,X_{n+1}$ and $Y_1,\dots, Y_k$, and any morphism
$X_f(X_1,\dots,X_n)\overset{\Phi}{\to} X_g(Y_1,\dots,Y_n)$, we have
\begin{equation}\label{eq:complete}
X_{1_1\vtl f\vtl p_1}(X_0,\dots,X_{n+1})\overset{id\otimes \Phi\otimes id}{\longrightarrow}X_{1_1\vtl
g\vtl p_1}(X_0,Y_1\dots,Y_k,X_{n+1})
\end{equation}
Indeed, note that by Lemmas \ref{lemma:causal_product} (i) and
\ref{lemma:onechain_causal}, we have $(1_1\vtl f\vtl p_1)^*=p_1\otimes f^*\otimes 1_1$. 
Taking adjoints and applying Proposition \ref{prop:Xf_const},  \eqref{eq:complete} is equivalent to 
\[
\tilde X_0^*\otimes X_{g^*}(\tilde Y_1,\dots, \tilde Y_k)\otimes \tilde
X_{n+1}\overset{id\otimes \Phi^*\otimes id}{\longrightarrow} \tilde X_0^*\otimes
X_{f^*}(\tilde X_1,\dots, \tilde X_k)\otimes \tilde
X_{n+1},
\]
 which is quite easy to see.

\end{remark}

We will next show that the causal product is related to the ordinal sum $\star$ of the corresponding
posets.

\begin{prop}\label{prop:vtl_ordinal} Let $f\in \Te_{n}$, $g\in \Te_{m}$
and
consider the decomposition $[n+m]=[n]\boxplus[m]$. Replace the labels of $\Pe_g$ by their
translations $L_S\mapsto n+L_S=\{n+i,\ i\in L_S\}$.
Then  
\begin{enumerate}
\item[(a)] If $[n]\in \Pe_f$ and $\emptyset \in \Pe_g$, then $\Pe_{f\vtl g}=\Pe_f\star
(\Pe_g\setminus\{\emptyset\})$, with all labels remaining the same.
\item[(b)] If $[n]\in \Pe_f$ and $\emptyset \notin \Pe_g$, then $\Pe_{f\vtl
g}=(\Pe\setminus \{[n]\})\star \Pe_g$, where the labels of $[n]$ are added to the
labels of elements in $\mathrm{Min}(\Pe_g)$.
\item[(c)] If $[n]\notin \Pe_f$ and $\emptyset \in \Pe_g$, then $\Pe_{f\vtl g}=\Pe_f\star
(\Pe_g\setminus\{\emptyset\})$, where the free outputs  of $f$ are added to the label sets
of elements in $\mathrm{Min}(\Pe_g\setminus \emptyset)$.
\item[(d)] If $[n]\notin \Pe_f$ and $\emptyset \notin \Pe_g$, then $\Pe_{f\vtl
g}=\Pe_f\star \{\bullet\}\star \Pe_g$, where $\{\bullet\}$ is a one-element poset with
label $L_\bullet=O_f^F$.

\end{enumerate}

\end{prop}

\begin{proof} By definition of the causal product, we have  
\[
f\vtl g=\sum_{S\in \Pe_f\setminus\{[n]\}} \hat f_Sp_S+ (\hat f_{[n]}-1 + \hat
g_\emptyset)p_{[n]}+\sum_{T\in \Pe_g\setminus \{\emptyset\}} \hat g_T p_{[n]\boxplus T}.
\]
The term in brackets can be equal to 1, -1, or 0, depending on whether $[n]\in \Pe_f$ and
$\emptyset\in \Pe_g$. The statement is now immediate.

\end{proof}

\begin{exm}\label{exm:causal}
The following Hasse diagrams show some examples of the causal products in cases (a)-(d) of
Proposition \ref{prop:vtl_ordinal}. All type functions are of the form $\beta\vtl f$ for a
chain type $\beta$.
By  Proposition
\ref{prop:append_chain_f} below, all the results are type functions. In the diagrams, the
indices belonging to $\beta$ are in green and yellow color. The color of the other indices
indicate the signalling relations.

\begin{center}
\begin{minipage}[c]{0.25\textwidth}
\centering
{\footnotesize (a)}
\vskip 10pt

\includegraphics[scale=0.7]{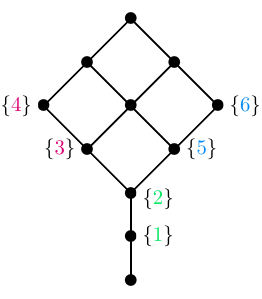}

{\footnotesize $\gamma_2\vtl (\gamma_2\otimes \gamma_2)$}
\end{minipage}%
\begin{minipage}[c]{0.25\textwidth}
\centering
{\footnotesize (b)}
\vskip 15pt

\includegraphics[scale=0.7]{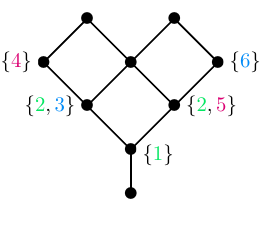}
\vskip 20pt

{\footnotesize $\gamma_2\vtl (\gamma_2\otimes \gamma_2)^*$}
\end{minipage}%
\begin{minipage}[c]{0.25\textwidth}
\centering
{\footnotesize (c)}
\vskip 10pt

\includegraphics[scale=0.7]{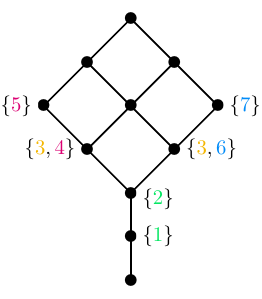}

{\footnotesize $\gamma_3^*\vtl (\gamma_2\otimes \gamma_2)$}
\end{minipage}%
\begin{minipage}[c]{0.25\textwidth}
\centering
{\footnotesize (d)}
\vskip 8pt

\includegraphics[scale=0.7]{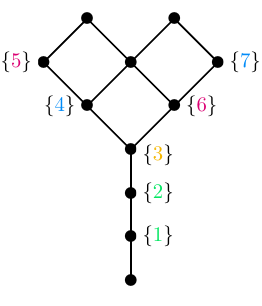}

{\footnotesize $\gamma^*_3\vtl(\gamma_2\otimes\gamma_2)^*$}

\end{minipage}
\end{center}

To identify maps of these types, note that from $f\le p_{O}^*$ with $O=O_f$, we get $\beta\vtl f\le \beta\vtl
p_{O}^*$ for any chain $\beta$. Since $\beta\vtl p_O^*$ is again a chain type, it follows
that any map of type $\beta\vtl f$ is essentially a comb, where the first tooth satisfies
additional constraints given by the type $f$ and  $\beta$. In example (a), $\beta=\gamma_2$ and $f$ is the type
of bipartite no-signalling channels, hence the corresponding map is a 2-comb as depicted
below, where the inputs can only signal to outputs of the same color or green. 
We can similarly find the maps in the case (b), here the first tooth is
a process matrix, with colors indicating the signalling restrictions:

\begin{center}
\begin{minipage}[c]{0.46\textwidth}
\centering
\includegraphics[scale=0.8]{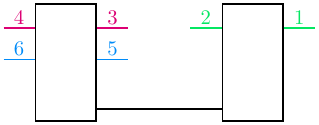}

{\footnotesize (a)}
\end{minipage}%
\begin{minipage}[c]{0.46\textwidth}
\centering
\includegraphics[scale=0.8]{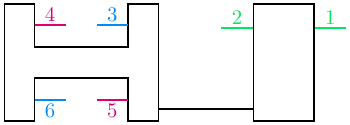}

{\footnotesize (b)}
\end{minipage}%

\end{center}

In the cases (c) and (d), $\beta=\gamma_3^*=\gamma_2\otimes 1_1=\gamma_2\vtl 1_1$ (Example
\ref{exm:T3} (d) and Lemma \ref{lemma:onechain_causal})
and $\beta\vtl f=(\gamma_2\vtl 1_1)\vtl f=\gamma_2\vtl(1_1\vtl f)$. It follows that the
corresponding maps have a similar form as in (a) resp. (b), but the first tooth has an
additional global output (indexed by 3). Here the pink and blue inputs can also signal to
the  yellow  output:

\begin{center}
\begin{minipage}[b]{0.46\textwidth}
\centering
\includegraphics[scale=0.8]{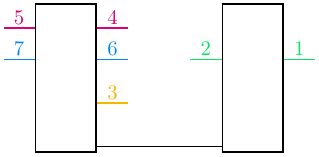}

{\footnotesize (c)}
\end{minipage}%
\begin{minipage}[b]{0.46\textwidth}
\centering
\includegraphics[scale=0.8]{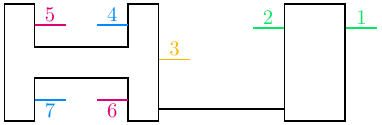}
\vskip 5pt
{\footnotesize (d)}
\end{minipage}%

\end{center}

\end{exm}

It can be seen from the proof of Proposition \ref{prop:vtl_ordinal} 
that if both $f$ and $g$ are chain types with chains of $N$ and $M$ elements,
respectively, then $f\vtl g$ is a chain type for a chain with $M+N\pm 1$
elements. Note also that that this construction can be
interpreted as appending the two chains in the respective order. 
 It is not clear that if $f$ and $g$ are arbitrary type functions, then $f\vtl g$ is a type function
as well, but using the inequalities in Lemma \ref{lemma:Pf} and
\ref{lemma:causal_tensor}, we have
\[
p_{I_f}\otimes p_{I_g}\le f\otimes g\le f\vtl g\le (f^*\otimes g^*)^*\le (p_{O_f}\otimes
p_{O_g})^*,
\]
so that $f\vtl g$ is a subtype, with inputs in $I_f\boxplus I_g$ and outputs in $O_f\boxplus
O_g$.

Our next result shows that if $f$ or $g$ is a chain type, we always obtain a type
function.

\begin{prop}\label{prop:append_chain_f}  Let $f\in \Te_{n_1}$ and let $\beta\in \Te_{n_2}$
be a chain type. Then both $f\vtl \beta$ and $\beta\vtl f$ are types, with outputs
$O=O_f\boxplus O_\beta$ and inputs $I=I_f\boxplus I_\beta$. 

\end{prop}

\begin{proof} Let  $\beta=\sum_{k=1}^{N}(-1)^{k-1}p_{S_k}$
for some odd $N$ and $S_1\subsetneq \dots \subsetneq S_N\subseteq [n_2]$. We will proceed
by induction on $N$. Suppose $N=1$. If $S_1=\emptyset$,
then $\beta=1_{n_2}$ and we have by Lemma \ref{lemma:onechain_causal}
\[
f\vtl 1_{n_2}=f\otimes 1_{n_2}\in \Te_{n_1+n_2}
\]
and
\[
1_{n_2}\vtl f=(p_{n_2}\vtl f^*)^*=(f\otimes p_{n_2})^*\in \Te_{n_1+n_2}.
\]
Assume that  $S_1=[n_2]$, then $\beta=p_{n_2}$ and the assertion follows by duality. 
If $\emptyset\ne S_1\subsetneq [n_2]$, then we have $\beta\approx p_{m_1}\otimes
1_{m_2}=p_{m_1}\vtl 1_{m_2}$ for $m_1=|S_1|$, $m_1+m_2=n_2$. Then 
\[
\beta\vtl f\approx p_{m_1}\vtl (1_{m_2}\vtl f)\in \Te_{n_1+n_2},\quad f\vtl \beta\approx (f\vtl p_{m_1})\vtl 1_{m_2} \in
\Te_{n_1+n_2},
\]
by the first part of the proof and Lemma \ref{lemma:causal_product}.

Assume next that the assertion holds for all odd numbers $M<N$. Using Proposition
\ref{prop:vtl_ordinal} (d), we see that $\beta\approx \beta_1\vtl\beta_2$, where $\beta_1$ is
an $N-2$-element chain type and $\beta_2$ is a one-element chain  type. Then
\[
\beta\vtl f\approx \beta_1\vtl(\beta_2\vtl f),\qquad f\vtl\beta\approx (f\vtl\beta_1)\vtl\beta_2
\]
are type functions, by the induction assumption.

To prove the statement on the output and input indices, note that for any $i\in [n_1]\boxplus [n_2]$, we have  $e^i_{n_1+n_2}=e^j_{n_1}\theta_{n_2}$ or
$e^i_{n_1+n_2}=\theta_{n_1}e^k_{n_2}$
for some $j\in [n_1]$, $k\in [n_2]$. Then   
\[
f\vtl\beta(e^i)=f(e^j_{n_1})\ \text{ or } f\vtl \beta(e^i)=\beta(e^k_{n_2}).
\]
The statement on input/output indices  follow from Lemma \ref{lemma:fh_setting}. The proof
for $\beta\vtl f$ is similar. 

\end{proof}

\subsection{The structure of type functions}

 The  following structure theorem for type functions is one of the main
results of the paper. It shows that any type function $f$ is contained in the lattice
generated by chains with the same input/output indices as $f$. This is quite similar to the 'normal form' of
\cite{hoffreumon2026projective}, which is defined in the quantum case and uses 
the formalism of projections and their algebraic properties. Moreover, we show  that
these chains are concatenations of some basic chains in different orders.

\begin{theorem}\label{thm:structure}
Let  $f\in \Te_n$. Then there is  a decomposition
$[n]=\boxplus_{i=1}^k[n_i]$, chain types 
$\beta_1\in \Te_{n_1}$,\dots, $\beta_k\in
\Te_{n_k}$ such that $O_f=\boxplus_j O_{\beta_j}$, $I_f=\boxplus_j I_{\beta_j}$, finite index sets $A$, $B$ and permutations $\pi_{a,b}\in
\permut_k$, $a\in A$, $b\in B$ such that 
\[
f\approx \bigvee_{a\in A}\bigwedge_{b\in B} (\beta_{\pi^{-1}_{a,b}(1)}\vtl \dots \vtl
\beta_{\pi^{-1}_{a,b}(k)})=\bigwedge_{b\in B}\bigvee_{a\in A}(\beta_{\pi^{-1}_{a,b}(1)}\vtl \dots \vtl
\beta_{\pi^{-1}_{a,b}(k)}).
\]
\end{theorem}

\begin{proof} We will once again proceed by induction on $n$. Since any element in $\Te_n$
for $n\le 3$ is a chain type, the statement clearly holds in this case. Assume the
condition holds for all $m<n$. The condition is obviously invariant under permutations. 
Assume $f$ can be written in the given form, then
\[
f^*\approx \bigwedge _{a\in A}\bigvee_{b\in B} (\beta^*_{\pi^{-1}_{a,b}(1)}\vtl \dots \vtl
\beta^*_{\pi^{-1}_{a,b}(k)})=\bigvee_{b\in B}\bigwedge_{a\in A} (\beta^*_{\pi^{-1}_{a,b}(1)}\vtl \dots \vtl
\beta^*_{\pi^{-1}_{a,b}(k)}).
\]
Since $\beta_j^*$ is a chain type for each $j$, this proves the statement for $f^*$. 

It is now enough to show this form for  $f=f_1\otimes f_2$, where 
$f_1\in \Te_m$, $f_2\in \Te_{n-m}$ with $[n]=[m]\boxplus[n-m]$. By the induction assumption,
$f_1$ and $f_2$ satisfy the conditions, so that
\begin{align*}
f_1&\approx\bigvee_{a\in A}\bigwedge_{b\in B} (\beta^1_{\pi^{-1}_{a,b}(1)}\vtl \dots \vtl
\beta^1_{\pi^{-1}_{a,b}(k_1)})=\bigwedge_{b\in B}\bigvee_{a\in A} (\beta^1_{\pi^{-1}_{a,b}(1)}\vtl \dots \vtl
\beta^1_{\pi^{-1}_{a,b}(k_1)}),\\
f_2&\approx\bigvee_{c\in C}\bigwedge_{d\in D} (\beta^2_{\tau^{-1}_{c,d}(1)}\vtl \dots \vtl
\beta^2_{\tau^{-1}_{c,d}(k_2)})=\bigwedge_{d\in D} \bigvee_{c\in C}(\beta^2_{\tau^{-1}_{c,d}(1)}\vtl \dots \vtl
\beta^2_{\tau^{-1}_{c,d}(k_2)})
\end{align*}
for some chain types  $\beta^1_j\in \Te_{m_j}$, $[m]=\boxplus^{k_1}_{j=1}[m_j]$, and $\beta^2_j\in \Te_{l_j}$,
$[n-m]=\boxplus_{j=1}^{k_2}[l_j]$ and permutations $\pi_{a,b}\in \permut_{k_1}$, $\tau_{c,d}\in
\permut_{k_2}$.
Let 
\[
\beta^{a,b}_1:=\beta^1_{\pi^{-1}_{a,b}(1)}\vtl \dots \vtl
\beta^1_{\pi^{-1}_{a,b}(k_1)},\qquad \beta^{c,d}_2:= \beta^2_{\tau^{-1}_{c,d}(1)}\vtl \dots \vtl
\beta^2_{\tau^{-1}_{c,d}(k_2)}.
\]
Using the   properties of the tensor product (Lemma \ref{lemma:fproduct})ii), we get from
Lemma \ref{lemma:causal_tensor}
\[
f\approx \bigl(\bigvee_{a\in A}\bigwedge_{b\in B}\beta_1^{a,b}\bigr)\otimes \bigl(\bigvee_{c\in
C}\bigwedge_{d\in D}\beta_2^{c,d}\bigr)=\bigvee_{a,c}\bigwedge_{b,d}
(\beta_1^{a,b}\otimes \beta_2^{c,d})=\bigvee_{a,c}\bigwedge_{b,d}(\beta_1^{a,b}\vtl
\beta_2^{c,d})\wedge (\beta_2^{c,d}\vtl \beta_1^{a,b})
\]
On the other hand, using Lemma \ref{lemma:causal_tensor} and Lemma
\ref{lemma:causal_product}, we
get
\begin{align*}
f&\approx \bigl(\bigwedge_{b\in B}\bigvee_{a\in A}\beta_1^{a,b}\bigr)\otimes \bigl(\bigwedge_{d\in D}\bigvee_{c\in
C}\beta_2^{c,d}\bigr)\\
&=\left[\bigl(\bigwedge_{b\in B}\bigvee_{a\in A}\beta_1^{a,b}\bigr)\vtl \bigl(\bigwedge_{d\in D}\bigvee_{c\in
C}\beta_2^{c,d}\bigr)\right]\wedge \left[\bigl(\bigwedge_{d\in D}\bigvee_{c\in
C}\beta_2^{c,d}\bigr)\vtl\bigl(\bigwedge_{b\in B}\bigvee_{a\in A}\beta_1^{a,b}\bigr)
\right]\\
&= \bigl(\bigwedge_{b,d}\bigvee_{a,c} \beta_1^{a,b}\vtl \beta_2^{c,d}\bigr) \wedge
\bigl(\bigwedge_{b,d}\bigvee_{a,c} \beta_2^{c,d}\vtl \beta_1^{a,b}\bigr).
\end{align*}
We have the decomposition $[n]=\boxplus_{j=1}^k[n_j]$, with $k=k_1+k_2$ and $n_j=m_j$,
$j=1,\dots, k_1$,  $n_j=l_{j-k_1}$, $j=k_1+1,\dots,k$, and chain types $\beta_j\in
\Te_{n_j}$, $\beta_j=\beta_j^1$ for $j=1,\dots,k_1$ and $\beta_j=\beta^2_{j-k_1}$ for
$j=k_1+1,\dots,k$. To get the permutation sets, let $A'=A\times C$, $B'=B\times D\times
\permut_2$ and define $\pi_{a',b'}$ in $\permut_k$ as the block permutation with respect to the
decomposition $[k]=[k_1]\boxplus[k_2]$ (see Appendix \ref{sec:permut})
\[
\pi_{(a,c),(b,d,\lambda)}=\rho_\lambda\circ(\pi_{a,b}\boxplus \tau_{c,d}).
\]
This  finishes the proof.

\end{proof}

\begin{remark} Note that the fact that the minima and maxima in the above theorem can be
exchanged is by no means automatic and follows from the properties of the causal
product $\vtl$ and the structure of type functions.

\end{remark}

\begin{remark} In general, it is not clear for which sets of permutations such a
combination of chain types will be a chain type. Nevertheless, since all the connected
chains have the same input and output indices, a function of the form as in Theorem
\ref{thm:structure} will always be a subtype. That is, the objects corresponding to such a
function will describe a set of channels, obtained by taking pullbacks and pushouts of
objects describing combs.  On the other hand, Example \ref{exm:T3sub} shows that not all subtypes
are of this form.

\end{remark}

\begin{remark} Let us note that the set of basic chains $\beta_1,\dots,\beta_k$ in Theorem
\ref{thm:structure} is not given uniquely. We will consider a way to obtain some choice of this
set in the next section.

\end{remark}

\subsection{The reduced poset $\Pe_f^0$}

\label{sec:pf0}
 In this paragraph we will be concerned with the problem of obtaining the
decomposition in the structure theorem \ref{thm:structure} for a type function $f$. To this end, we introduce
the reduced poset of $f$,  denoted by $\Pe_f^0$, defined as the subposet in $\Pe_f$, consisting of the elements with 
nonempty labels and $\emptyset$ if it is contained in $\Pe_f$.
 We will show that any $f\in \Te_n$ is fully determined by $\Pe_f^0$ (and
$n$). This is convenient, because $\Pe_f^0$ is much smaller and easier to visualise that $\Pe_f$. More
importantly, from $\Pe_f^0$, one can find a  decomposition of $f$ as causal products  tensor products and complements
of other functions. In particular, it is possible to obtain from $\Pe_f^0$ some choice of
the chain types  $\beta_1,\dots,\beta_k$ in the structure decomposition of $f$.

We start by some basic properties  and examples of $\Pe_f^0$. Some further properties, and more
technical parts of the proofs, can be found in Appendix \ref{app:pf0}. 

Recall  from Section \ref{sec:labelled}  that for all $T\in \Pe_f$, we have $T=\cup\{L_{T'},\ T'\in \Pe_f^0, T'\subseteq
T\}$.

\begin{lemma}\label{lemma:p0_basic} Let $f\in \Te_n$.
\begin{enumerate}
\item[(i)] $\mathrm{Min}(\Pe_f^0)=\mathrm{Min}(\Pe_f)$.
\item[(ii)] $\Pe_f^0$ is a chain $\iff$  $\Pe_f$ is a chain $\iff$  $\Pe_f^0= \Pe_f$.

\item[(iii)] If $\Pe^0_f$ has a largest element,  then it is the largest element in
$\Pe_f$. In this case  $f$ or $f^*$ has a free output. 

\end{enumerate}
\end{lemma}

\begin{proof} 
(i) Obvious. For (ii), assume that $\Pe^0_f$ is a chain and let $S,T\in \Pe_f$. Let $i\in
S\setminus T$ and $j\in T$, and let $S'\subseteq  S$ and $T'\subseteq  T$ be such that $S',T'\in
\Pe_f^0$ and $i\in L_{S'}$, 
$j\in L_{T'}$. Since $\Pe_f^0$ is a chain, $S'$ and $T'$ are comparable. If $S'\subseteq T'$,
then $S'\subseteq  T'\subseteq  T$, so that $i\in T$, which is not possible. Hence
$T'\subseteq S' \subseteq S$, for all $T'\in \Pe_f^0$, $T'\subseteq T$. Hence $T\subseteq
S$, and $\Pe_f$ is a chain. It is clear that then $\Pe_f^0=\Pe_f$.  

If $\Pe_f$ is not a chain, then there are some type
functions
$f_1$, $f_2$ such that $f=f_1\otimes f_2$ or $f=(f_1\otimes f_2)^*$. Moreover, the ranks
of $f_1$ and $f_2$ are at least 2. It follows that both $\Pe_{f_1\otimes f_2}$ and
$\Pe_{(f_1\otimes f_2)^*}$ contain an element $S\boxplus T$, where $S\in \Pe_{f_1}$, $T\in
\Pe_{f_2}$ but none of the two elements is minimal. Then there is some $S'\in \Pe_{f_1}$
and $T'\in \Pe_{f_2}$ such that $S'\boxplus T$, $S\boxplus T'\subseteq S\boxplus T$, so that no
element of $S\boxplus T$ is a label. Hence
$S\boxplus T\notin
\Pe_f^0$, so that $\Pe_f\ne \Pe_f^0$. 

To prove (iii) let $T$ be the largest element in $P^0_f$. Then 
\[
\cup\Pe_f=\cup L_S \subseteq T\subseteq \cup \Pe_f.
\]
 If follows that $T=\cup\Pe_f $ is the largest element in $\Pe_f$. If $T\ne [n]$, then
 clearly, $f$ has some free outputs. If $T=[n]$, then since $[n]\in \Pe_f^0$, we have
 $\emptyset\ne L_{[n],f}=O_{f^*}^F$, so that $f^*$ has free outputs.

\end{proof}

The basic constructions on $\Pe_f^0$ can be easily obtained from those on $\Pe_f$.

\begin{lemma}\label{lemma:p0_constr} Let $f\in \Te_{n}$, $g\in \Te_{m}$. 
\begin{enumerate}
\item[(i)] $\Pe_f^0$ and $\Pe_{f^*}^0$ are related as follows:
\begin{itemize}
\item If $S\notin \{\emptyset,[n]\}$, then $S\in \Pe^0_{f} \iff S\in \Pe_{f^*}^0$.
\item  $\emptyset \in \Pe^0_f$ $\iff$ $\emptyset \notin \Pe^0_{f^*}$
\item $[n]\in \Pe^0_f$ $\iff$ $f^*$ has  free outputs $\implies$ $[n]\notin
\Pe_{f^*}^0$.
\end{itemize}
In particular, if both $f$ and $f^*$ have no free outputs, then
$\Pe_{f^*}^0=\Pe_f^0\triangle \{\emptyset\}$.
\item[(ii)] $\Pe^0_{f\otimes g}\simeq \{(S,T):\ S\in \Pe_f^0,\ T\in \Pe_g^0, \ S\in
\mathrm{Min}(\Pe_f) \text{ or } T\in \mathrm{Min}(\Pe_g)\}$, with ordering given as $(S,T)\le
(S',T')$ if and only if $S\le S'$, $T\le T'$, and with labels as in Corollary \ref{coro:Pf}(iii).

\item[(iii)] The causal product $f\vtl g$ is connected with the ordinal sum $\star$ of
posets as follows. Replace the labels of $\Pe^0_g$ by their
translations $L_S\mapsto n+L_S=\{n+i,\ i\in L_S\}$. 
\begin{itemize}
\item If $\emptyset \in \Pe_g^0$, then
\[
\Pe^0_{f\vtl g}=\Pe_f^0\star (\Pe_{g}^0\setminus \{\emptyset\})
\]
with free outputs of $f$ added to all labels of elements covering $\emptyset$ in
$\Pe_g^0$. 
\item If $\emptyset\notin \Pe_g^0$ and $f$ has no free outputs, then 
\[
\Pe^0_{f\vtl g}=(\Pe_{f}^0\setminus \{[n]\})\star \Pe_g^0
\]
with labels of $[n]$ added to labels of all elements in $\mathrm{Min}(\Pe_g^0)$.
\item If $\emptyset\notin \Pe_g^0$ and $f$ has free outputs, then 
\[
\Pe^0_{f\vtl g} =\Pe_f^0\star \{\bullet\}\star\Pe_g^0
\]
where the label set $L_{\bullet}$ consists of free outputs of $f$, all other labels remain the
same.

\end{itemize}

\end{enumerate}

\end{lemma}

\begin{proof} (i) The first two assertions are clear, the rest follows from  the
proof of Lemma \ref{lemma:p0_basic} (iii).
The statement (ii) follows from Corollary \ref{coro:Pf} (iii). 
For (iii) we use Proposition \ref{prop:vtl_ordinal}.
If $\emptyset \in \Pe_g^0$, then parts (a) and (c) of Proposition \ref{prop:vtl_ordinal} imply the result. 
If $\emptyset \notin \Pe_g^0$, we apply parts (b) and (d). If $f$ has no free
outputs, either of the situations may occur, but note that in the situation (d) the new
element has an empty label, so it does not belong to $\Pe_f^0$. If $f$ has free outputs,
then only the situation in (d) is possible, with the new element now
having a nonempty label. This  implies the statement.

\end{proof}

\begin{exm}\label{exm:pf0_basic} The diagrams below show the diagrams of the restricted posets $\Pe_f^0$ for
no-signalling channels $\gamma_2\otimes\gamma_2$, $n$-slot process matrices
$(\gamma_2\otimes \dots \otimes \gamma_2)^*$ and multiround process matrices
$(\gamma_4\otimes \gamma_6)^*$ (compare with the $\Pe_f$ diagrams in Example
\ref{exm:T4}). Notice that the diagrams for process matrices are disconnected and also the sets
of labels are separated,  compare this with Proposition \ref{prop:nofree_components}
below.  
\begin{center}
\begin{minipage}[b]{0.3\textwidth}
\centering
\includegraphics[scale=0.8]{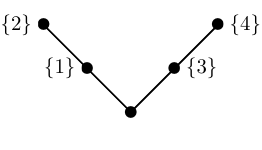}
\vskip -2mm
{\footnotesize no-signalling channel}
\end{minipage}%
\begin{minipage}[b]{0.3\textwidth}

\centering
\includegraphics[scale=0.8]{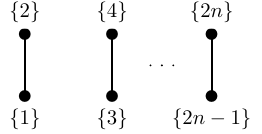}

{\footnotesize $n$-slot process matrix}
\end{minipage}%
\begin{minipage}[b]{0.3\textwidth}

\centering
\includegraphics[scale=0.8]{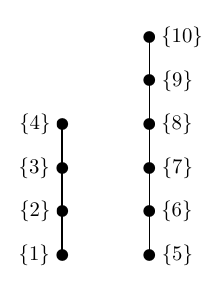}

{\footnotesize multiround process matrix}
\end{minipage}%

\end{center}

\end{exm}

\begin{exm}\label{exm:adapter} The next diagram shows the restricted poset $\Pe_f^0$ for
the type $a$ of adapters, see Example \ref{exm:subtype_afp}.  Since $a=g\to g$ for the
process matrix type $g$,
the diagram can be obtained from the diagram of $g$ in Example \ref{exm:pf0_basic} and Lemma
\ref{lemma:p0_constr}. The function of the input process matrix (pink indices in the diagram) acts on $s_1...s_4$,
while the target process matrix (blue indices) acts on $s_5\dots s_8$. This poset has a more
complicated structure as the previous ones, in particular, notice that the labels are
repeated.
\medskip
\begin{center}
\includegraphics[scale=0.8]{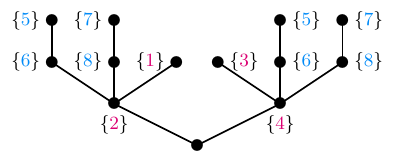}
\medskip

\footnotesize{adapter}
\end{center}

\end{exm}

\subsection{Decomposition of $\Pe_f^0$}

We now proceed by showing that $\Pe_f^0$ can be decomposed into a set of chains from which
$f$ is constructed by applying $\vtl$, $\otimes$ and $ ^*$.

It the statement below, notice that we have $f\vtl 1_\varepsilon=1_\varepsilon\vtl f=f$,
for the trivial type function $1_\varepsilon$.

\begin{prop}\label{prop:pf0_nofree_decomp} Let $f\in \Te_n$. There is a decomposition
$[n]=[k][k+1,m][m+1,n]$, where any of the subintervals might be empty, type functions
$\beta_1\in \Te_k$, $\beta'\in \Te_{n-m-1}$ and $h\in \Te_{m-k-1}$ such that $\beta$,
$\beta'$ are chain types, $h$ nor $h^*$ have any free indices and 
\[
f\approx \beta\vtl h\vtl \beta'.
\]

\end{prop}

\begin{proof}
We have by \eqref{eq:nofree}  that $f\approx p_{k_1}\otimes h_1 \otimes
1_{l_1}$, where $k_1=|I^F_f|$, $l_1=|O^F_f|$ and $h_1\in \Te_{n_1}$, $n_1:=n-k_1-l_1$ has no free indices. By Lemma
\ref{lemma:onechain_causal} we get $f\approx p_{k_1}\vtl h_1\vtl 1_{l_1}$.  Assume $h_1^*$ has free
indices, then $h_1^*\approx p_{k_2}\vtl h_2\vtl 1_{l_2}$, so that 
\[
f\approx p_{k_1}\vtl 1_{k_2}\vtl h_2^* \vtl p_{l_2}\vtl 1_{l_1}=\beta_2\vtl h_2^*\vtl
\beta_2'
\]
where $\beta_2:=p_{k_1}\vtl 1_{k_2}$ and $\beta_2':=p_{l_2}\vtl 1_{l_1}$ are chain types, and $h_2\in \Te_{n_2}$, $n_2:=n_1-k_2-l_2$ has no free indices. If $h_2^*$ has again free
indices, we may proceed in this way, obtaining 
\[
f\approx \beta_i\vtl h_i\vtl \beta_i'
\]
with chain types $\beta_i$, $\beta_i'$ and $h_i\in \Te_{n_i}$ at each step. Since $n_i$ decreases, we either get to
$n_i\le 3$, in which case $f$ must be a chain type, and we may put $k=m=n$ or $k=m=0$  in the above
decomposition, or we get to the situation when $h_i$ and $h_i^*$ have no free
indices. 

\end{proof}

In the situation of the above Proposition, if $f$ is not a chain, $\Pe_h^0$ and the chain
types $\beta_1$ and $\beta_2$ can be
seen from $\Pe_f^0$ as follows. Put $f_1=\beta_1\vtl h$, so that $f=f_1\vtl \beta_2$. 
Since $f_1$ has no largest element and no free outputs, we obtain from Lemma \ref{lemma:p0_constr} that
$\Pe_f^0= \Pe_{f_1}^0\star
(\Pe_\beta\setminus\{\emptyset\})$, with the same sets of labels.
It follows that there exists a largest element in $\Pe_f^0$ with the property that
it covers more than one element. Let this element be $S$ and let $T_1,\dots, T_k$ be the
elements covered by $S$. There is a chain $S=S_1\lneq \dots \lneq S_K$, where $S_K$ is
the largest element in $\Pe_f^0$. We then have $\Pe_{f_1}^0\simeq \cup_j T_j^{\downarrow}$. 
If $K$ is even, add an element $S_0$ with empty
label at the bottom of the chain to obtain a chain of even length. This then corresponds
to the chain type $\beta_2$.  Some examples are given below.

\begin{exm}\label{exm:pf0_causal} The following labelled Hasse diagrams correspond to
causal products of the form $f\vtl \beta$ for a chain $\beta$:

\begin{center}
\begin{minipage}[c]{0.3\textwidth}
\centering
{\footnotesize (a)}
\vskip 30pt

\includegraphics[scale=0.9]{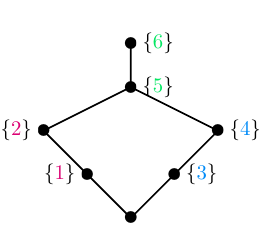}

{\footnotesize $(\gamma_2\otimes \gamma_2)\vtl \gamma_2$}
\end{minipage}%
\begin{minipage}[c]{0.3\textwidth}
\centering
{\footnotesize (b)}
\vskip  13pt

\includegraphics[scale=0.9]{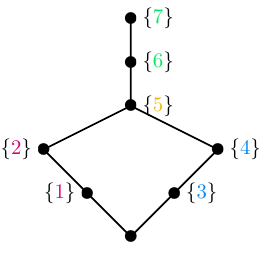}

{\footnotesize $(\gamma_2\otimes \gamma_2)\vtl \gamma_3$}
\end{minipage}%
\begin{minipage}[c]{0.3\textwidth}
\centering
{\footnotesize (c)}
\vskip 15pt

\includegraphics[scale=0.9]{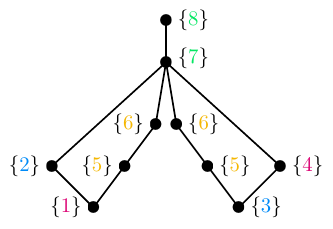}
\vskip 5pt

{\footnotesize $((\gamma_2\otimes \gamma_2)^*\otimes \gamma_2)\vtl \gamma_2$}
\end{minipage}
\end{center}

Similarly as in Example \ref{exm:causal}, we have $f\vtl\beta\le p_O^*\vtl \beta$, so that 
the maps of the corresponding types have the form of a comb, where this time the last
tooth has a special structure given by $f$ and $\beta$:

\begin{center}
\begin{minipage}[c]{0.46\textwidth}
\centering
\includegraphics[scale=0.8]{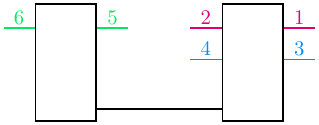}

{\footnotesize (a)}
\end{minipage}%
\begin{minipage}[c]{0.46\textwidth}
\centering
\includegraphics[scale=0.8]{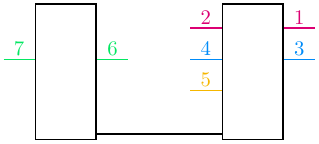}

{\footnotesize (b)}
\end{minipage}%

\end{center}
In (a),  green inputs can signal to all outputs, otherwise an input can only signal to
outputs of the same color. In the case (b), the channel in the last tooth has an
additional input (with index 5), which can signal to both pink  and blue outputs. This
is due to the fact that $\gamma_3=p_1\otimes \gamma_2=p_1\vtl
\gamma_2$, so that $f\vtl \gamma_3=(f\vtl p_1)\vtl \gamma_2$, which extends the type $f$
by a global past. The maps of the type in (c) have a 2-comb structure as well, but the
constraints are more complicated:
\begin{center}
\includegraphics[scale=0.8]{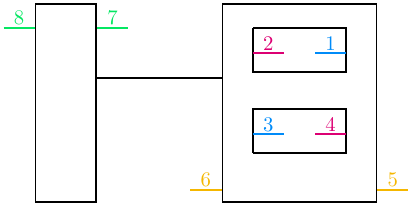}

{\footnotesize (c)}
\end{center}
The last tooth is a no-signalling composition of a process matrix with a channel. The colors
are related to signalling conditions as in case (a). Let us remark that the signalling
relations can be seen directly from the order structure of the reduced poset, see 
\cite{jencova2026order}.

\end{exm}

We now decompose $f_1$ as $f_1\approx \beta_1\vtl h$.
Assume first that $\Pe_{f_1}^0$ has no least element. In this case
$f_1^*$ has no free inputs, since $\mathrm{Min}(\Pe_{f^*_1}^0)=\{\emptyset\}$. We also have
$I^F_f=I^F_{f_1}$, and  $f_1\approx p_k\otimes h=p_k\vtl h$, with
$k=|I^F_{f_1}|$, which gives us $\beta_1=p_k$.

If  $\Pe_{f_1}^0$ has a least element $S_1$,  let $S$ be the smallest element in
$\Pe_f^0$ with the property that it is covered by more than one element. Let  $S_1\lneq
\dots \lneq S_K$ be a chain such that $S_K=S$ and let  $T_1,\dots, T_l$ be the elements
covering $S$. Put $L:=\cap_j L_{T_j}$. 
Using Lemma  \ref{lemma:p0_constr} as before, we have the following situations. 

{\it Case $L=\emptyset$.}\enspace   If $K$ is odd, then the chain corresponds to a chain type
$\beta_1\in \Te_{k}$, $k=|S_K|=\sum_j |L_{S_j}|$ and $\Pe_h^0=\cup_j
T_j^{\uparrow,f_1}\cup\{\emptyset\}$. If $K$ is even, then $\Pe_h^0=\cup_j
T_j^{\uparrow,f_1}$ and $\beta_1$
is the chain type for the chain $S_1\subsetneq \dots\subsetneq S_{K-1}$ (with free outputs
in $L_{S_K}$). 

{\it Case $L\ne \emptyset$.} \enspace Add an element $S_{K+1}$ at the end of the chain, with label
$L_{S_{K+1}}=L$ and put $k=|S_{K+1}|$. If $K$ is odd, then $\beta_1\in \Te_k$ is the chain type for the chain
$S_1\lneq \dots \lneq S_K$ (with free outputs in $L$) and $\Pe_H^0=\cup_j
T_j^{\uparrow, f_1}\cup \{\emptyset\}$, with labels of $T_j$ replaced by $L_{T_j}\setminus L$.  If
$K$ is even, then $\Pe_h^0=\cup_j T_j^{\uparrow,f_1}$ and $\beta_1$ is the chain type  for
$S_1\subsetneq\dots\subsetneq S_{K+1}$.

See also the diagrams in Example \ref{exm:causal} (with the corresponding reduced poset $\Pe_f^0$) for some examples.

\begin{exm}[Comb to comb]\label{exm:comb_to_comb} The diagram on the left below depicts the function
$g=(\gamma_2\to \gamma_4)=(\gamma_2\otimes \gamma_4^*)^*\in \Te_6$. By Proposition \ref{prop:chains_combs} and
Corollary  \ref{coro:maps}, the corresponding higher order objects represent maps from channels
to 2-combs. 
\begin{center}
\begin{minipage}[c]{0.45\textwidth}
\centering
\includegraphics[scale=0.8]{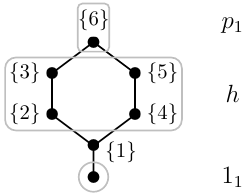}
\end{minipage}
\begin{minipage}[c]{0.45\textwidth}
\centering
\includegraphics[scale=0.8]{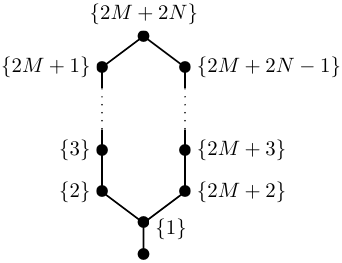}
\end{minipage}
\end{center}
Applying Proposition \ref{prop:pf0_nofree_decomp}  together with
the remarks below it, we obtain that $g\approx 1_1\vtl h\vtl p_1$, as indicated on the
diagram.  It is easily checked that
$h$ is the process matrix type $h\simeq(\gamma_2\otimes\gamma_2)^*$ (see Example \ref{exm:T4}), and
$g$ is the type of process matrices with global past and future (Example
\ref{exm:pm_past_future}). In this way, we decomposed $g$
into chain types, with respect to the decomposition of $[6]$ into four intervals
$[6]=[1][2,3][4,5][6,6]$. Let
$s=s^1\dots s^4$ be the corresponding concatenation. We indicate the component on which
any of the chain types is acting by an upper index, so that 
 $g\approx 1_1^1\vtl (\gamma_2^2\otimes \gamma_2^3)^*\vtl p_1^4$. By Lemmas \ref{lemma:causal_product} and \ref{lemma:causal_tensor}, we have
\begin{align*}
g&\approx 1^1_1\vtl((\gamma_2^2\vtl \gamma_2^3)\wedge (\gamma_2^3\vtl \gamma_2^2))^*\vtl
p^4_1=
1^1_1\vtl((\gamma_2^2\vtl \gamma_2^3)^*\vee(\gamma_2^3\vtl\gamma_2^2)^*)\vtl p^4_1\\
&=(1_1\vtl (\gamma^2_2)^*\vtl (\gamma_2^3)^*\vtl p_1^4)\vee (1_1\vtl (\gamma^3_2)^*\vtl
(\gamma_2^2)^*\vtl p_1^4).
\end{align*}
It can be checked that the last expression is equal to $\gamma_6\vee (\gamma_6\circ
\rho_\lambda)$, where $\lambda\in \permut_4$  acts as $(1234)\mapsto (1324)$ and
$\rho_\lambda$ is the block permutation corresponding to $\lambda$ and the above
decomposition of $[6]$, see Appendix \ref{sec:permut}. Similarly, for any $M,N$, the type
describing higher order objects corresponding to mappings of $M$-combs into $N$-combs (see
the diagram on the right) is
an element of $\Te{2(M+N)}$, and  with respect to the decomposition
$[2(M+N)]=[1][M][N-2][1]$,  it can be decomposed as
\begin{align*}
1_1\vtl(\gamma_{2M}\otimes \gamma_{2N-2})^*\vtl p_1&=(1_1\vtl \gamma^*_{2M}\vtl
\gamma_{2N-2}^*\vtl p_1)\vee (1_1\vtl \gamma^*_{2N-2}\vtl
\gamma_{2M}^*\vtl p_1)\\
&=\gamma_{2(M+N)}\vee (\gamma_{2(M+N)}
\circ\rho_\lambda).
\end{align*}
In the case of quantum objects, this corresponds to results of \cite{perinotti2017causal}.

\end{exm}

We will now deal with the case when $f$ and $f^*$ have no free indices.
Let $\Pe$ be a poset with labels in $[n]$ and let $\Pe_1,\Pe_2\subseteq \Pe$ be nonempty. We will say that $\Pe_1$ and $\Pe_2$ are {independent
components} of $\Pe$ if $\Pe=\Pe_1+\Pe_2$ (direct sum of posets) and $L_S\cap L_T=\emptyset$ for 
any $S\in \Pe_1$ and $T\in \Pe_2$. In this case, we will write
\[
\Pe=\Pe_1\indep \Pe_2.
\]

\begin{prop}\label{prop:nofree_components} Let $f\in \Te_n$ be such that $f$ and  $f^*$ have no
free indices. Assume that $\emptyset \notin \Pe_{f}$. 
\begin{enumerate}
\item[(i)] If $f^*\approx f_1\otimes f_2$ for some type functions $f_1$ and
$f_2$, then 
\[
\Pe_f^0\simeq (\Pe_{f_1}^0\setminus \emptyset) \indep  (\Pe_{f_2}^0\setminus
\emptyset).
\]
\item[(ii)] If $\Pe_f^0=\Pe_1\indep \Pe_2$ for some labelled subposets $\Pe_1$ and $\Pe_2$, 
then there are some type functions $f_1$ and $f_2$ such that $\Pe_1=(\Pe_{f_1}^0\setminus
\emptyset)$, $\Pe_2=(\Pe_{f_2}^0\setminus \emptyset)$ and
$f^*\approx f_1\otimes f_2$.
\item[(iii)] If $f\approx f_1\otimes f_2$ for type functions $f_1$ and
$f_2$, then no 
decomposition of $\Pe_f^0$ into independent components exists.

\end{enumerate}

\end{prop}

\begin{proof} Under the assumptions, $\emptyset \in \Pe_{f^*}$ and 
$\Pe_f^0=\Pe_{f^*}^0\setminus \{\emptyset\}$. Assume that $f_1\in \Te_{n_1}$ and $f_2\in \Te_{n_2}$ are such that $f^*=f_1\otimes f_2$
for the decomposition $[n]=[n_1]\boxplus [n_2]$. Since $\emptyset \in
\Pe^0_{f^*}=\Pe^0_{f_1\otimes f_2}$, we see by Lemma \ref{lemma:p0_constr} (ii) that both
$\Pe^0_{f_1}$ and $\Pe^0_{f_2}$ must contain $\emptyset$, and  $\Pe_{f^*}^0$ consists
of $\Pe_{f_1}^0$ and $\Pe_{f_2}^0$ glued at $\emptyset$, with labels of $\Pe_{f_2}^0$
translated by $n_1$. Since $\Pe_f^0=\Pe^0_{f^*}\setminus \{\emptyset\}$, the assertion (i) 
follows.  

For (iii), assume that $f=f_1\otimes f_2$ and $\Pe_f^0=\Pe_1\indep \Pe_2$. Note that none of
the functions can be a 1-element chain, since then $f$ would have free inputs or outputs. Let
$\mathrm{Min}(\Pe_{f_1})=\{U_1,\dots,U_k\}$,  $\mathrm{Min}(\Pe_{f_2})=\{V_1,\dots,V_l\}$.
By Lemma \ref{lemma:p0_constr}, we have  $\mathrm{Min}(\Pe_{f})=\mathrm{Min}(\Pe_{f_1})\times \mathrm{Min}(\Pe_{f_2})$. 
For some $i$ and $j$, let $(U_i,V_j)\in \Pe_1$. Since $(U_i,V_j)$ and $(T,V_j)$ are
comparable for any $T\in \Pe^0_{f_1}$, $U_i\le T$, we must have $(T,V_j)\in \Pe_1$ for all
such $T$. By Lemma \ref{lemma:pf0_mincover}, there is some $T$ that covers $U_i$. 
But then $L_{(T,V_j)}=L_{T,V_{j'}}=L_T$ for all $j'$, so that $(T,V_{j'})\in \Pe_1$
for all $j'$. Since $(U_i,V_{j'})\le (T,V_{j'})$ for all $j'$, this implies that $(U_i,
V_{j'})\in \Pe_1$ for all $j'$. By the same reasoning with $V_j$, we get that all $(U_i,
V_j)\in \Pe_1$, which is not possible.

For (ii), assume that $\Pe_f^0=\Pe_1\indep \Pe_2$. Since either $f$ or $f^*$ is, up to a
permutation, a tensor
product of type functions and we cannot have $f\approx f_1\otimes f_2$ by (iii), 
 it must hold that $f^*\approx f_1\otimes f_2$. But then by (i) $\Pe_f^0=\Pe'_1\indep \Pe'_2$, with 
$\Pe_i'=\Pe^0_{f_i}\setminus \{\emptyset\}$. Let 
\[
\Pe_f^0=\mathcal Q_1\indep \dots \indep  \mathcal Q_M
\]
be the finest decomposition into independent components. Then there are some $C,D\subset [M]$
such that 
\[
\Pe_1=\Indep_{i\in C} \mathcal Q_i,\quad  \Pe_2=\Indep_{i\in [M]\setminus C} \mathcal
Q_i, \quad \Pe'_1=\Indep_{i\in D} \mathcal Q_i,\quad  \Pe'_2=\Indep_{i\in [M]\setminus D}
\mathcal Q_i.
\]
Assume that $M\ge 3$, otherwise each $\Pe_i$ is one of $\Pe'_j$ and we are done. 
Then $D$ or $[M]\setminus D$ has at least two elements. So assume $|D|\ge 2$. 
This implies that $\Pe_1'$ has no largest element and since we have $\Pe_1'=\Pe_{f_1}^0\setminus\{\emptyset\}$,  
$\Pe_{f_1}^0$ has no largest element either. By Lemma \ref{lemma:p0_constr} (i), this
implies that $f_1^*$ has no free outputs, and since 
 $f_1$ has no free outputs by the assumptions,
the same Lemma implies that  $\Pe_1'=\Pe_{f_1^*}^0$. It follows that  $f_1^*$ satisfies the assumptions
 in (iii). Hence $f_1\approx g_1\otimes g_2$ for some type functions $g_1$ and $g_2$ and we
 obtain that 
 \[
\Pe_{g_1}^0\setminus \{\emptyset\} = \Indep_{i\in D'} \mathcal Q_i,\quad \Pe_{g_2}^0\setminus
\{\emptyset\} = \Indep_{i\in D\setminus D'}\mathcal Q_i
 \]
for some $D'\subset D$. Continuing in this way, we obtain that for any $i\in [M]$ there is
some type function $g_i$ such that $\emptyset \in \Pe_{g_i}^0$ and $\mathcal
Q_i=\Pe_{g_i}^0\setminus \{\emptyset\}$. It follows that 
\[
\Pe_1=\Indep_{i\in C}(\Pe_{g_i}^0\setminus \{\emptyset\})=\Pe^0_{\otimes_{i\in
D}g_i}\setminus\{\emptyset\}
\]
and $f_1\approx \otimes_{i\in D}g_i$, similarly, $\Pe_2=\Pe^0_{\otimes_{i\in
[M]\setminus D}g_i}\setminus\{\emptyset\}$ and $f_2\approx\otimes_{i\in [M]\setminus
D}g_i$.

\end{proof}

It follows from the above proof that in the situation of the above Proposition, if $f^*\approx g_1\otimes
\dots\otimes g_k$, we can identify $\Pe_{g_l}^0$  by looking at the independent
components of $\Pe_f^0$.
\begin{exm}\label{exm:components} The next diagram shows $\Pe_f^0$ for a type function
$f\in \Te_{16}$:

\begin{center}
\includegraphics[scale=0.8]{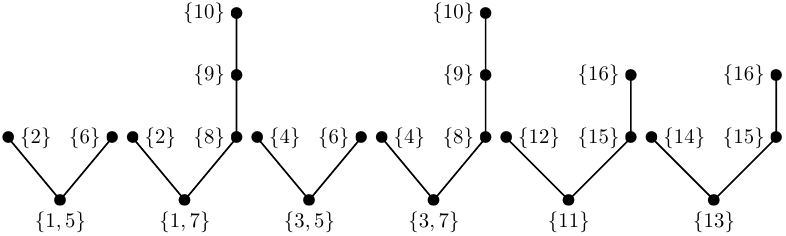}
\end{center}

Note that $\Pe_f^0$ has no largest or smallest element, and also no free indices, so the
assumptions of Proposition \ref{prop:nofree_components} are satisfied. As a poset, $\Pe_f^0$  is a direct sum of 6 posets, but there are only two independent
components: one with labels $\le 10$, and one with labels  $>10$. It follows that
$f^*\approx
g_1\otimes g_2$. We can obtain
$\Pe_{g_1}^0$ and $\Pe_{g_2}^0$ by adding $\emptyset$, that is, an unlabelled minimal
element to each component, and relabelling if necessary:
\begin{center}
\begin{minipage}[c]{0.6\textwidth}
\centering
\includegraphics[scale=0.8]{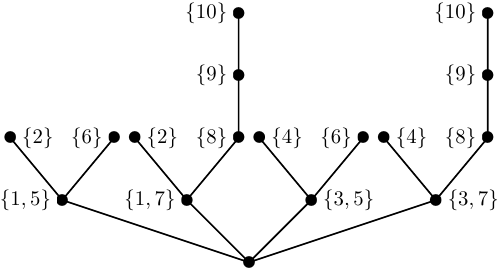}
\end{minipage}
\begin{minipage}[c]{0.35\textwidth}
\centering
\includegraphics[scale=0.8]{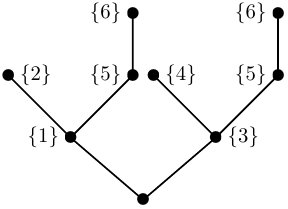}
\end{minipage}
\end{center}
We can also see that both $g_1^*$ and $g_2^*$ are products. Using the procedure described
in Appendix \ref{app:pf0}, we obtain that $g_1\approx ((\gamma_2\otimes \gamma_2)^*\otimes (\gamma_2\otimes
\gamma_4)^*)^*$ and $g_2\approx ((\gamma_2\otimes\gamma_2)^*\otimes \gamma_2)^*$.

\end{exm}

 The case when $\Pe_f^0$ has no independent components is somewhat more complicated.
 By the proposition  above, this
is the case when $f\approx f_1\otimes\dots\otimes f_k$. The proof of the following result is in Appendix \ref{app:pf0}.

\begin{prop}\label{prop:nofree_product}  Let $f\in \Te_n$ be such that $f$ and $f^*$ have
no free indices and $\emptyset\notin \Pe_f^0$. Assume $f\approx f_1\otimes
\dots\otimes f_k$ is a finest decomposition of $f$ as a product.  Then the labelled posets
$\Pe_{f_l}^0$, $l=1,\dots,k$ can be obtained from the structure of $\Pe_f^0$, up to a
permutation on the labels.

\end{prop}

\begin{theorem}\label{thm:pf0} Every type function $f\in \Te_n$ is fully determined by the
labelled poset $\Pe_f^0$.  

\end{theorem}

\begin{proof}  We will proceed by induction on
$n$. If $f$ is a chain type, the assertion follows from Lemma
\ref{lemma:p0_basic} and Proposition \ref{prop:chains}, so that the assertion holds for
$n\le 3$. Assume it is true for all $m<n$ and let $f\in \Te_n$. By Proposition
\ref{prop:pf0_nofree_decomp}  and remarks below it, if $f$ or
$f^*$ has
some free indices, then we have $f\approx \beta_1\vtl h \vtl \beta_2$, where $\beta_1,\beta_2$
are chain types and $h\in \Te_m$ is such that $h$ and $h^*$ do not have any free indices.
Moreover,  the chain types and $\Pe_h^0$ 
 can be obtained from $\Pe_f^0$. Since $m<n$,  $h$ is
 determined by $\Pe_h^0$ by the induction assumptions, so we are done.

If both $f$ and $f^*$ have no free indices, we may assume that $\emptyset \notin \Pe_f$, otherwise we 
replace $f$ by $f^*$. Then if $\Pe_f^0$ has independent components, we have
$f^*\approx f_1\otimes f_2$ for some type functions $f_i\in \Te_{n_i}$, $i=1,2$ and $n=n_1+n_2$,
such that the components have the form $\Pe_{f_i}^0\setminus \{\emptyset\}$. Since
$n_1,n_2<n$, we are done. If $\Pe_f^0$ has no independent components, then
$f\approx f_1\otimes\dots \otimes f_k$ for some $f_i\in \Te_{m_i}$, $i=1,\dots,k$, $m_1+\dots
+m_k=n$, and we obtain   $\Pe_{f_l}^0$ for all $l=1,\dots,k$ from $\Pe_f^0$ by Proposition
\ref{prop:nofree_product}. Again, the assertion follows by applying induction assumption
on each $f_l$.

\end{proof}

Using repeatedly the procedure described in the above proof, we get to the situation when
all the obtained components are necessarily chains. In this way, we get a decomposition of $f$ into
chain types, together with a recipe for a construction of  $f$ from these chain types by using
tensor products, complements  and causal products. It is then clear, using also the proof
of Theorem \ref{thm:structure}, that these chain types give us a choice of
$\beta_1,\dots,\beta_k$ in the structure theorem \ref{thm:structure}.

\section{Conclusions}

We have studied the category $\Af$ of affine subspaces in finite dimensional vector
spaces, endowed with a *-autonomous structure inherited from the compact closed category
$\FV$. We used this structure to define a hierarchy of higher order objects, which,
restricted to certain objects we call quantum, contains the higher order quantum maps. 
The quantum objects have the space $M_n^h$ as the underlying vector space and satisfy the
condition that both the affine subspace and its affine dual contains a multiple of the
identity. This restriction satisfies the assumptions of the setting in \cite{kissinger2019acategorical}.

We used the combinatorial characterization of types as in \cite{perinotti2017causal, bisio2019theoretical,
apadula2024nosignalling} to relate functions in $\Fe_n$ to certain objects in $\Af$,
constructed over a fixed sequence of first order objects $X_1,\dots,X_n$.  We proved that
this relation defines a representation of the boolean algebra $\Fe_n$ in the lattice of
subspaces of the tensor product of the underlying vector spaces $V_{X_1}\otimes
\dots\otimes V_{X_n}$. We also proved that all higher
order objects can be obtained in this way, the corresponding functions
are called the type functions.

We then studied the type functions and their properties. We defined a  labelled poset
$\Pe_f^0$ for each type function $f$ and proved that $f$ is related to combs if and only if
$\Pe^0_f$ is a chain. We showed that $f$ is fully determined by $\Pe_f^0$ and devised a
procedure to decompose $\Pe_f^0$ to a set of basic chains, from which the type function
$f$ is constructed. We proved a structure theorem, saying that $f$ can be obtained by
taking maxima and minima over concatenations of these chains in different orders.

There is a number of questions left for future research. The first set of open  questions
pertains the type functions themselves. It is still not clear how the type functions, or,
equivalently, the corresponding posets $\Pe_f$ or $\Pe_f^0$, can be characterized among
elements in $\Fe_n$. In this respect, it might be interesting to study some general notion
of labelled posets, their structure and categorical constructions over them.

In connection to higher order objects, it will be shown in a subsequent work
\cite{jencova2026order} that the
signalling relations of a type can be obtained from the order structure in  $\Pe_f^0$ (or
$\Pe_f$). As shown in
\cite{apadula2024nosignalling}, the signalling relations restrict the way how the higher
order maps can be connected so that the resulting map is given by a subtype. 
This should correspond to some construction over the posets of the involved type functions. 

In the context of the structure theorem, one may ask which sets of permutations over
connections of chains would lead to a type function after taking maxima and minima. By the proof of Theorem
\ref{thm:structure}, it seems that these should be of a block form, coming from the
operadic structure on permutations, see \ref{sec:permut}.

Another important problem is a deeper study of causal order of higher order objects,
related to causal separability, \cite{oreshkov2016causal}. 
Causal  nonseparability cannot be the property of a type function itself, 
but involves also the positive cone. Nevertheless, it is an interesting question how the problem of
witnessing causal nonseparability \cite{araujo2015witnessing} fits into the framework of $\Af$. 
 There are other important properties that are not expressed as affine
constraints, such as LOSR or LOSE channels, or constructions such as quantum circuits with
classical or quantum constraints \cite{wechs2021quantum}. Finding necessary additional
structures that would enable us to study such questions in $\Af$ is another direction of
future investigations.

Finally, as already mentioned in the introduction, one may think of building a similar
theory of higher order maps in a general probabilistic theory, see e.g.
\cite{plavala2023general}. In this framework, state spaces are represented by compact
convex sets that can be obtained by an intersection of a positive cone in a vector space by
a hyperplane given by an interior element of the dual positive cone.  In this way, it is
related to a first order object in $\Af$ and we may think of channels as linear maps
transforming between hyperplanes, satisfying certain positivity properties. Here the
choice of positive cones and their tensor products plays an important role.

Besides the quantum theory, one of the basic  examples is the  classical theory, where the state spaces are probability
simplexes. The underlying positive cones form a compact closed category
$\mathbf{Mat}(\mathbb R_+)$, which enables us to identify the higher order classical
theory with the classical objects in our framework. This situation appears also in the
quantum case, where the underlying compact closed category is $\bf{CPM}$, the category of completely positive maps.

The higher order theory over classical state spaces is far from trivial, since
it is known that there are purely classical higher order processes 
(e.g. the Lugano process \cite{baumeler2016thespace}) that exhibit indefinite causality. 
As mentioned above, such properties cannot be explained by type functions themselves, but
our approach enables us to study both quantum and classical higher order theory on an equal
footing.

For general GPTs, we cannot expect to have a compact closed category at our disposal. 
But if a theory can be constructed based of ordered vector spaces, with tensor product
giving it a structure of a closed monoidal or even *-autonomous category, then one can
expect the higher order maps to be described by higher order objects in $\Af$, intersected
by a positive cone.

\section*{Acknowledgement}  
I am indebted to Bert Lindenhovius and Gejza Jen\v{c}a for fruitful discussions, and  to  Adam Jen\v{c}a for technical support. 
I am also grateful to the anonymous referees whose comments helped me to considerably improve the
paper.  The research was supported by the grants VEGA 2/0128/24 and by the Slovak Research
and Development Agency  under the contract No. APVV-22-0570. Funded by the EU NextGenerationEU through the Recovery and Resilience
Plan for Slovakia under the project 07I04-04-V05-00114 (AI-ONKO).

\bibliographystyle{quantum}
\bibliography{all_hom}

\appendix

\numberwithin{equation}{section}

\section{Some basic definitions}
\label{app:basic}

For $m\le n\in \mathbb N$, we will denote the corresponding interval $\{m,m+1,\dots,n\}$ by
$[m,n]$. For $m=1$, we will simplify to  $[n]:=[1,n]$. Let $\permut_n$ denote the set of all permutations of $[n]$.

\subsection{Block permutations}
\label{sec:permut}

 For $n_1,n_2\in \mathbb N$, $n_1+n_2=n$, 
we will denote by $[n]=[n_1]\boxplus [n_2]$ the decomposition of $[n]$ as a concatenation of two 
intervals
\[
[n]=[n_1][n_1+1,n_1+n_2].
\]
Similarly, for $n=\sum_{j=1}^kn_j$, we have the decomposition
\[
[n]=\boxplus_{j=1}^k[n_j]=[m_1,m_1+n_1][m_2,m_2+n_2]\dots[m_k,m_k+n_k],
\]
where $m_j:=\sum_{l=1}^{j-1} n_j$ (so $m_1=0$). Note that the order of $n_1,\dots, n_k$ in
this decomposition is
fixed. 

We have two kinds of special permutations related to the above decomposition. For
$\sigma_j\in \permut_{n_j}$, we denote by $\boxplus_j \sigma_j\in \permut_n$ the permutation that acts as
\[
m_j+l\mapsto m_j+\sigma_j(l),\qquad l=1,\dots,n_j,\ j=1,\dots, k. 
\]
On the other hand, we have for any $\lambda\in \permut_k$ a unique permutation
$\rho_\lambda\in\permut_n$  such that $\rho_\lambda^{-1}$ acts as
\[
[m_1,m_1+n_1][m_2,m_2+n_2]...[m_k,m_k+n_k]\mapsto
[m_{\lambda(1)}+n_{\lambda(1)}][m_{\lambda(2)}+n_{\lambda(2)}]\dots[m_{\lambda(k)}+n_{\lambda(k)}]
\]
Note that we have
\[
\rho_\lambda\circ(\boxplus_j\sigma_j)=(\boxplus_j \sigma_{\lambda(j)})\circ\rho_\lambda.
\]
(These permutations  come from the operadic structure on the set of
all permutations $\permut_*$. See \cite{leinster2004higher} for the definition of and operad.)

\subsection{Partially ordered sets}\label{app:poset}

An overall reference for this section is \cite{stanley2011enumerative}.

A partially ordered set, or a {poset}, is a set $\Pe$ endowed with a reflexive, antisymmetric and
transitive relation $\le$, called the partial order. We will only encounter  the situation
when $\Pe$ is finite. A basic example of a poset is the set $\Pe(X)$ of all subsets of a
finite set $X$, ordered by inclusion. If $X=[n]$, we will denote $\Pe(X)$ by $2^n$. 

A subposet in a poset $\Pe$ is a $\mathcal Q\subseteq \Pe$ endowed with the partial order
relation inherited from $\Pe$. For any subset $\mathcal R\subseteq \Pe$, we define two
special 
subposets in $\Pe$ as
\[
\mathcal R^\downarrow=\{p\in \Pe,\ p\le r \text{ for some } r\in \mathcal R\}, \ \mathcal
R^\uparrow=\{p\in \Pe,\ p\ge  r \text{ for some } r\in \mathcal R\}.
\]

The set of minimal elements in $\Pe$ will be denoted by $\mathrm{Min}(\Pe)$. 
For elements $p,q\in \Pe$, we say that $q$ covers $p$, in notation $p\cover q$,  if $p\le q$
and for any $r$ such that $p\le r\le q$ we have $r=p$ or $r=q$. If $p$ covers a minimal
element, we will say that $p$ is a minimal covering element.

A totally ordered subposet $\Ce\subseteq \Pe$ is called a chain in $\Pe$. Such a chain  is
maximal if it is not contained in any other chain in $\Pe$.
The length of a chain $\Ce$ is defined as $|\Ce|-1$. 

We say that a poset $\Pe$  is graded of  rank $k$ if every maximal chain of $P$ has the
same length $k$. In this case, there is a (unique) rank function 
$\rho: \Pe\to \{0,1,\dots,k\}$ such that $\rho(p)=0$ if $p$ is a
minimal element of $\Pe$ and $\rho(q)=\rho(p)+1$ if $p\cover q$. Basic examples of graded
posets are chains, antichains and $2^n$.

If $\Pe$ and $\mathcal Q$ are posets with disjoint sets, their direct sum $\Pe+\mathcal Q$ is a poset defined as 
the disjoint union $\Pe\cup \mathcal Q$, such that the order is preserved in each
component and elements in different components are incomparable. 
 Another way to compose
$\Pe$ and $\mathcal Q$ is the ordinal sum $\Pe\star \mathcal Q$, where the underlying set
is again the union $\Pe\cup \mathcal Q$ and the order in each component is preserved, but
for $p\in \Pe$ and $q\in \mathcal Q$ we have $p\le q$. A third way to compose posets that
we will use is the direct product $\Pe\times \mathcal Q$, where the underlying set is the
cartesian product $\Pe\times \mathcal Q$, with $(p_1,q_1)\le (p_2,q_2)$ if $p_1\le p_2$ in
$\Pe$  and $q_1\le q_2$ in $\mathcal Q$.
If $\Pe_1$ and $\Pe_2$ are graded posets with rank functions $\rho_1$ and $\rho_2$, then 
$\Pe_1\times \Pe_2$ is graded as well, with rank function $\rho$ given as
\[
\rho(p_1,p_2)=\rho_1(p_1)+\rho_2(p_2).
\]

The Hasse diagram of a finite poset $\Pe$ is a graph whose vertices are elements of $\Pe$
and there is an edge between $p$ and $q$ if $p\cover q$, and if $p\lneq r$, then $r$ is drawn above
$p$. Two posets are isomorphic if and only if they have the same Hasse diagrams.

\subsection{Binary strings}
\label{app:binary}
A binary string of length $n$ is a sequence  $s=s_1\dots s_n$, where $s_i\in
\{0,1\}$. Such a string can be interpreted as an element $\{0,1\}^n$, but also as a 
map $[n]\to \{0,1\}$, or a subset in  $[n]:=\{1,\dots,n\}$. It will be convenient to use
all these interpretations, but we will distinguish between them. The strings in
$\{0,1\}^n$ will be denoted by small letters, whereas the corresponding subsets of $[n]$
will be denoted by the corresponding capital letters. More specifically, for $s\in \{0,1\}^n$ and 
$T\subseteq [n]$, we denote
\begin{equation}\label{eq:string_subset}
S:=\{i\in [n],\ s_i=0\},\qquad t:=t_1\dots t_n,\ t_j=0 \iff j\in T.
\end{equation}
As usual, the set of all subsets of $[n]$ will be denoted by $2^n$. 
With the inclusion ordering and complementation $S^c:=[n]\setminus S$,
$2^n$ is a boolean algebra, with the smallest element $\emptyset$ and largest element
$[n]$.  

The group $\permut_n$ has an obvious action on $\{0,1\}^n$. Indeed,
 for a string $s$  interpreted as a map $[n]\to 2$, we may define the action of
$\sigma\in \permut_n$ by precomposition as
\[
\sigma(s):=s\circ\sigma^{-1}=s_{\sigma^{-1}(1)}\dots s_{\sigma^{-1}(n)}.
\]
Note that in this way we have $\rho(\sigma(s))=(\rho\circ \sigma)(s)$. For a decomposition
$[n]=\boxplus_{j=1}^k[n_j]$, we have a corresponding decomposition of
any string $s\in \{0,1\}^n$ as a concatenation of strings
\[
s=s^1\dots s^k,\qquad s^j\in \{0,1\}^{n_j}.
\]
For permutations $\sigma_j\in \permut_{n_j}$ and  $\lambda\in
\permut_k$, we have
\[
\rho_\lambda\circ(\boxplus_j\sigma_j)(s^1\dots s^k)=\rho_\lambda(\sigma_1(s^1)\dots
\sigma_k(s^k))=\sigma_{\lambda(1)}(s^{\lambda(1)})\sigma_{\lambda(2)}(s^{\lambda(2)})\dots
\sigma_{\lambda(k)}(s^{\lambda(k)}).
\]

\subsection{Boolean functions and the  M\"obius transform}
\label{sec:boolean}

A function $f:\{0,1\}^n\to \{0,1\}$ is called a boolean function. 
The set of boolean functions, with pointwise ordering and complementation given by the
negation $\bar f=1-f$,  is a boolean algebra that can be identified with $2^{2^n}$.
We will denote the maximal element (the constant 1 function) by $1_n$. Similarly,
we denote the constant zero function by $0_n$.  For boolean
functions $f,g$, the pointwise minima and maxima will be denoted by $f\wedge g$ and $f\vee
g$. It is easily seen that we have
\begin{equation}\label{eq:wedgevee_fun}
f\vee g= f+g-fg,\qquad f\wedge g=fg,
\end{equation}
all the operations are pointwise. We now introduce and important example.

\begin{exm}\label{ex:pS}
For $S\subseteq [n]$, we define
\[
p_S(t)=\Pi_{j\in S}(1-t_j),\qquad t\in \{0,1\}^n.
\]
That is, $p_S(t)=1$ if and only if $S\subseteq T$. In particular,
$p_\emptyset=1_n$ and $p_{[n]}$ is the characteristic function of the zero string.
Clearly, for $S,T\subseteq [n]$ we have $p_{S\cup T}=p_Sp_T=p_S\wedge p_T$, in particular,
$p_S=\Pi_jp_{\{j\}}$. 
\end{exm}

By the M\"obius transform, all boolean functions can be expressed as combinations of the functions $p_S$, $S\subseteq
[n]$ from the previous example.

\begin{theorem}\label{thm:mobius} Any $f:\{0,1\}^n\to 2$ can be expressed  in the form 
\[
f=\sum_{S\subseteq [n]} \hat f_Sp_S
\]
in a unique way. The coefficients  $\hat f_S\in \mathbb R$ obtained as
\[
\hat f_S=\sum_{\substack{t\in \{0,1\}^n\\ t_j=1, \forall  j\in S^c}} (-1)^{\sum_{j\in
S}t_j}f(t).
\]

\end{theorem}

\begin{proof}  By the M\"obius inversion formula (see \cite[Sec.
3.7]{stanley2011enumerative} for details),
functions $f, g: 2^n\to \mathbb R$ satisfy
\[
f(S)=\sum_{T\subseteq S} g(T),\qquad S\in 2^n
\]
if and only if 
\[
g(S)=\sum_{T\subseteq S}(-1)^{|S\setminus T|} f(T).
\]
We now express this in terms of the corresponding strings $s$ and $t$.
It is easily seen that $T\subseteq S$ if and only if
$s_j=0$ for all $j\in T$, equivalently, $t_j=1$ for all $j\in S^c$. Moreover,
in this case we have  $|S\setminus T|=\sum_{j\in S} t_j$. This shows that $g(S)=\hat f_S$,
as defined in the statement. The first equality now gives
\[
f(s)=f(S)=\sum_{T\subseteq S} g(T)=\sum_{T:s_i=0,\forall i\in T}\hat f_T=\sum_{T:
p_T(s)=1}\hat f_T=\sum_{T\subseteq [n]} \hat f_Tp_T.
\]
For uniqueness, assume that $f=\sum_{T\subseteq [n]} c_Tp_T$ for some coefficients $c_T\in
\mathbb R$. Then 
\[
f(s)=\sum_{T: p_T(s)=1}c_T=\sum_{T\subseteq S}c_T.
\]
Uniqueness now follows by  uniqueness in the M\"obius inversion formula.

\end{proof}


\subsection{The boolean algebra $\Fe_n$}

\label{app:functions}
Let us introduce the subset of boolean functions 
\[
\Fe_n:=\{f:\{0,1\}^n\to 2,\ f(\theta_n)=1\},
\]
where we use $\theta_n$ to denote the zero string $00\dots 0$. 
In other words, $\Fe_n$ is the interval of all elements greater than $p_{[n]}$ in the boolean algebra 
$2^{2^n}$ of all boolean functions. With the pointwise ordering, $\Fe_n$ is a distributive lattice, with top element $1_n$ and 
 bottom element $p_n:=p_{[n]}$. We also define complementation  in $\Fe_n$ as
\[
f^*:=1_n-f+p_n.
\]
It can be easily checked that with these structures $\Fe_n$ is a boolean algebra, though
it is not a subalgebra of $2^{2^n}$.

Note that $p_S\in \Fe_n$ for any $S\subseteq [n]$, so in particular for $[k]$, with $k\le
n$. If $k<n$, we denote these functions as before by $p_{[k]}$, using the notation $p_n$
only for the distinguished bottom  element.

We now introduce some more operations in $\Fe_n$. For $f\in \Fe_n$ and any permutation
$\sigma\in \permut_n$, we clearly have $f\circ \sigma\in \Fe_n$.
Further, let $f\in \Fe_{n_1}$ and $g\in \Fe_{n_2}$. With the decomposition
$[n_1+n_2]=[n_1]\boxplus [n_2]$
and the corresponding concatenation of strings $s=s^1s^2$,  we define
the function $f\otimes g\in \Fe_{n_1+n_2}$ as
\[
(f\otimes g)(s^1s^2)=f(s^1)g(s^2),\qquad s^1\in \{0,1\}^{n_1},\ s^2\in \{0,1\}^{n_2}.
\]
Let $\lambda\in \permut_2$ be the transposition, then we have for any $f\in \Fe_{n_1}$ and
$g\in \Fe_{n_2}$
\[
(g\otimes f)=(f\otimes g)\circ \rho_\lambda,
\]
where $\rho_\lambda$ is the block permutation defined in  Section \ref{sec:permut}.

If $f,g\in \Fe_n$ are such that $g=f\circ\sigma$ for some $\sigma\in \permut_n$, we write
$f\approx g$. It is easily observed that if $f_1\approx g_1$ and $f_2\approx g_2$, then
$f_1\otimes f_2\approx g_1\otimes g_2$ and if $f\approx g$ then also $f^*\approx g^*$.

We now show some further important properties of these operations.

\begin{lemma}\label{lemma:fproduct} For $f\in \Fe_{n_1}$ and  $g,h\in \Fe_{n_2}$, we have
\begin{enumerate}
\item[(i)] $f\otimes g\le (f^*\otimes g^*)^*$, with equality if and only if either
$f=1_{n_1}$ and $g=1_{n_2}$, or $f=p_{n_1}$ and $g=p_{n_2}$.
\item[(ii)] $f\otimes (g\vee h)= (f\otimes g)\vee (f\otimes h)$, $f\otimes (g\wedge h)=
(f\otimes g)\wedge (f\otimes h)$.
\end{enumerate}

\end{lemma}

\begin{proof} The inequality in (i) is easily  checked, since $(f\otimes g)(s^1s^2)$ can be 1 only if
$f(s^1)=g(s^2)=1$. If both $s^1$ and $s^2$ are the zero strings, then $s^1s^2=\theta_{n_1+n_2}$ and both sides
are equal to 1. Otherwise, the condition $f(s^1)=g(s^2)=1$ implies that $(f^*\otimes
g^*)(s^1s^2)=0$, so that the right hand side must be 1. If $f$ and $g$ are both
constant 1, then 
\[
(1_{n_1}\otimes 1_{n_2})^*=1_{n_1+n_2}^*=p_{n_1+n_2}=p_{n_1}\otimes
p_{n_2}=1_{n_1}^*\otimes 1_{n_2}^*,
\]
in the case when both $f$
and $g$ are the minimal elements equality  follows by
duality. Finally, assume the equality holds and that $f\ne 1_{n_1}$, so that there is some $s^1$ such that 
$f(s^1)=0$. But then $s^1\ne \theta_{n_1}$, so that $f^*(s_1)=1$  and for any $s^2$,
\[
0=(f\otimes g)(s^1s^2)=(f^*\otimes
g^*)^*(s^1s^2)=1-f^*(s^1)g^*(s^2)+p_{n_1+n_2}(s^1s^2)=1-g^*(s^2),
\]
which implies that $g(s^2)=0$ for all $s^2\ne\theta_{n_2}$, that is, $g=p_{n_2}$. By the same argument,
$f=p_{n_1}$ if $g\ne 1_{n_2}$, which implies that either $f=1_{n_1}$ and $g=1_{n_2}$, or
$f=p_{n_1}$ and $g=p_{n_2}$.

The statement (ii) is easily proved from \eqref{eq:wedgevee_fun}.

\end{proof}

Consider the decomposition $[n]=[n_1]\boxplus [n_2]$ and let $S\subseteq [n_1]$,
$T\subseteq [n_2]$. We then denote by $S\boxplus T$ the disjoint union 
\begin{equation}\label{eq:disu}
S\boxplus T:=S\cup (n_1+T)=S\cup\{n_1+j,\ j\in T\}.
\end{equation}
We summarize some easy properties of the basic functions $p_S$, $S\subseteq [n]$.

\begin{lemma}\label{lemma:PSPT}
\begin{enumerate}
\item[(i)] For $S,T\subseteq [n]$, we have $S\subseteq T$ $\iff$ $p_T\le p_S$ $\iff$
$p_Sp_T=p_S$.
\item[(ii)] For $S\subseteq [n]$, $\sigma\in \permut_n$,
$p_S\circ\sigma=p_{\sigma^{-1}(S)}$.
\item[(iii)] For $S\subseteq [n_1]$ and $T\subseteq [n_2]$, $p_S\otimes p_T=p_{S\boxplus T}$.

\end{enumerate}
\end{lemma}

Let $f\in \Fe_n$ and let $\hat f$ be the M\"obius transform. Note that since $f$ has
values in $\{0,1\}$, we have by the proof of Theorem \ref{thm:mobius}
\[
\forall S\in 2^n, \quad \sum_{T\subseteq S} \hat f_T=f(s)\in \{0,1\}; \qquad \sum_{T\in 2^n} \hat
f_T=f(\theta_n)=1.
\]

\begin{prop}\label{prop:mobius} 

\begin{enumerate}
\item[(i)] For $f\in \Fe_n$ and  $\sigma\in \permut_n$, 
$\widehat{(f\circ \sigma)}_S=\hat f_{\sigma(S)}, \qquad S\subseteq [n]$.
\item[(ii)] For $f\in \Fe_n$, $\widehat{f^*}_S=\begin{dcases} 1-\hat f_S & S=\emptyset\text{ or } S=[n],\\
-\hat f_S & \text{otherwise}.
\end{dcases}$
\item[(iii)] For $f\in \Fe_{n_1}$, $g\in \Fe_{n_2}$, we have 
$\widehat{(f\otimes g)}_{S\boxplus T}=\hat f_S\hat g_T$, $S\subseteq [n_1]$, $T\subseteq
[n_2]$.
%
\end{enumerate}

\end{prop}

\begin{proof} All statements follow easily from Lemma \ref{lemma:PSPT} and the uniqueness part in Theorem
\ref{thm:mobius}.

\end{proof}

\section{Affine subspaces}
\label{sec:app_affine}
Let $V$ be a finite dimensional real vector space. A subset $A\subseteq V$ is an {affine
subspace} in $V$ if for any choice of  $a_1,\dots, a_k\in A$ and  $\alpha_1,\dots,\alpha_k\in \mathbb R$
such that $\sum_i\alpha_i=1$, we have $\sum_i\alpha_i a_i\in A$. It is clear that
$A=\emptyset$ is trivially an affine subspace.  Moreover, any linear subspace in $V$ is an affine subspace,
and an
affine subspace $A$ is linear if and only if $0\in A$. If $A \neq\emptyset$ and also
$0\notin A$, we say that $A$ is {proper}. 

A proper  affine subspace $A\subseteq V$ can be determined in two ways. Let 
\[
\lin(A):=\{a_1-a_2,\ a_1,a_2\in A\}.
\]
It is easily verified that $\lin(A)$ is a linear subspace, moreover, for any $a\in A$, we
have
\begin{equation}\label{eq:affine_l}
\lin(A)=\{a_1-a,\ a_1\in A\},\qquad A=a+\lin(A).
\end{equation}
We put $\dim(A):=\dim(\lin(A))$, the dimension of $A$. 

Let $V^*$ be the vector space dual of $V$ and let $\<\cdot,\cdot\>$ be the
duality. For a subset $C\subseteq V$, put
\[
C^*:=\{v^*\in V^*,\ \<v^*,a\>=1,\ \forall a\in C\}.
\]
Let $\tilde a\in A^*$ be any element and let $\Span(A)$ be the linear span of $A$ in
$V$. We then have
\begin{equation}\label{eq:affine_s}
A=\Span(A)\cap \{\tilde a\}^*,
\end{equation}
independently of $\tilde a$. The relation between the two expressions for $A$, given by
\eqref{eq:affine_l} and \eqref{eq:affine_s} is obtained as
\begin{equation}\label{eq:LandS}
\Span(A)=\lin(A)+\mathbb R\{a\},\qquad \lin(A)=\Span(A)\cap \{\tilde a\}^\perp,
\end{equation}
independently of $a\in A$ or $\tilde a\in A^*$. Here $+$ denotes the direct sum of
the vector spaces and $C^\perp$ denotes the annihilator of a set $C$.

\begin{exm}\label{exm:app_affine}
A basic example of an affine subspace is a hyperplane in $\mathbb R^n$, given by a vector
$x\in \mathbb R^n$ and a constant $c\in \mathbb R$ as 
\[
H=\{y\in \mathbb R^n,\ \<x,y\>=c\}.
\]
$H$ is a proper affine subspace iff $x\ne 0$ and $c\ne 0$, we then  may assume
that $c=1$ by properly normalizing $x$. In this case, $L_H=\{x\}^\perp$,  $S_H=\mathbb
R^n$ and $H^*=\{x\}$. 

In general,  an affine subspace $A\subseteq \mathbb R^n$ is the intersection of (finitely many)
hyperplanes, which shows that $A$ can be represented as the set of all solutions of a
system of linear equations $My=b$, for some $m\times n$ matrix $M$ and a vector $b\in
\mathbb R^m$. Clearly, $A$ is proper iff $b\ne 0$ and a solution of the system exists. In
this case $L_A$ is the space of solutions of the homogeneous system. Using a change of
basis if necessary, we may assume that $b=(1,0,\dots,0)$. Then $S_A$ is the set of
solutions of the homogeneous system $\tilde My=0$, where $\tilde M$ is the $(m-1)\times
n$ matrix obtained from $M$ by removing the first row, and the space $A$ is the
intersection of  $S_A$ by the hyperplane determined by the first row of $M$.

A simple geometric example is a straight line in $\mathbb R^3$, given by a point $x$ and a
vector $u$ as
\[
A=\{x+tu,\ t\in \mathbb R\}.
\]
This affine subspace is proper iff $x\ne 0$. Here $L_A=\mathbb Ru$ and $S_A$ is the plane containing the origin and $A$.
Any other plane containing $A$ is determined as $\{z\}^*$ for  some vector $z\perp u$ such that
$\<z,x\>=1$, this gives the representation $A=S_A\cap \{z\}^*$. 

\end{exm}

The following lemma is  easily proven.

\begin{lemma}\label{lemma:dual} Let $C\subseteq V$ be any subset. Then $C^*$ is an affine subspace in 
$V^*$ and we have
\[
0\in C^* \iff C= \emptyset, \qquad C^*=\emptyset\iff 0\in \aff(C).
\]
Assume $C\ne \emptyset$ and $0\notin \aff(C)$. Then
\begin{enumerate}
\item[(i)] $C^*$ is proper and we have $\lin(C^*)=C^\perp=\Span(C)^\perp$,
\item[(ii)] $\aff(C)=C^{**}$ and for any $c_0\in
C$, we have
\begin{align*}
\lin(C):= \Span\{c_1-c_2,\ c_1,c_2\in C\}=\Span\{c-c_0,\ c\in C\}=\lin(C^{**}).
\end{align*}

\end{enumerate}

\end{lemma}

\begin{coro}\label{coro:dual} Let $A\subseteq V$ be a proper affine subspace. Then 
\begin{enumerate}
\item[(i)] $A^*$ is a proper affine subspace in $V^*$ and $A^{**}=A$.
\item[(ii)] $\lin(A^*)=\Span(A)^\perp$, $\Span(A^*)=\lin(A)^\perp$.
\item[(iii)] $\dim(A^*)=\dim(V)-\dim(A)-1$.
\end{enumerate}

\end{coro}

The proper affine subspace $A^*$ in the above Corollary will be called the {affine
dual}
of $A$. Note that the dual depends on the choice of the ambient vector space $V$.

\subsection{A proof of Lemma \ref{lemma:combs} (ii)}
\label{app:lemma}

Let us restate the statement (ii) of the lemma.

\begin{lemma} Let $Z$ be any object in $\Af$ and let
$F_0,F_1,F_2$ be first order objects. Then we have
\[
(Z\multimap F_1)\multimap (F_0\multimap F_2)=((F_0\multimap Z)\multimap
F_1)\multimap F_2.
\]
\end{lemma}

\begin{proof} It is clear by the definition of $\multimap$ that we need to prove that
\[
(Z\multimap F_1)\otimes F_0\simeq (F_0\multimap Z)\multimap F_1.
\]
Put  $X:=F_0\otimes (Z\multimap F_1)$ and  $Y:=(F_0\multimap Z)\otimes F_1^*$, then it is enough to show that $X=Y^*$.
We first check that  $V_X=V_{F_0}\otimes V_Z^*\otimes V_{F_1}=V_{Y^*}$. 
Let now  $w\in A_{X}$ and $v\in A_{Y}$. Since $A_{F_1^*}=\{\tilde a_{F_1}\}$,  we have  $v=u\otimes
\tilde a_{F_1}$ for some $u\in A_{F_0\multimap Z}$ and 
\[
\<w,v\>=\<w,u\otimes \tilde a_{F_1}\>=\<(id\otimes \tilde a_{F_1})(w),u\>.
\]
Using Lemma \ref{lemma:combs} (iii) we get  $(id\otimes \tilde a_{F_1})(w)\in A_{F_0\otimes Z^*}$,
which implies that $\<w,v\>=1$. Hence  $A_{X}\subseteq A_{Y}^*$. Note that, in fact,  in any
*-autonomous category, we have a canonical embedding $(A\multimap B)\otimes
C\hookrightarrow
(A^*\otimes C)\parr B\simeq (C\multimap A)\multimap B$.

It is now enough to prove that $d_X=d_{Y^*}$, or equivalently, $\dim(S_X)=\dim(S_{Y^*})$. We have
\[
S_X=S_{F_0}\otimes S_{(Z\otimes F_1^*)^*} = V_{F_0}\otimes L_{Z\otimes F_1^*}^\perp,
\]
so that 
\[
\dim(S_X)=D_{F_0}(D_ZD_{F_1}-d_{Z\otimes F_1^*})= D_{F_0}D_ZD_{F_1}-D_{F_0}d_Z,
\]
where we used Lemma \ref{lemma:tensor_spaces} and the fact that $F_1$ is first order, so
that $d_{F_1^*}=0$. Further, since $S_{Y^*}=L_Y^\perp$, we have
\[
\dim(S_{Y^*})=D_{F_0}D_ZD_{F_1}-d_Y=D_{F_0}D_ZD_{F_1}-d_{(F_0\otimes
Z^*)^*},
\]
where we again have used the fact that $F_1$ is first order. We have
\[
d_{(F_0\otimes Z^*)^*}= \dim(S_{F_0\otimes
Z^*}^\perp)=D_{F_0}D_Z-D_{F_0}\dim(L_{Z}^\perp)=D_{F_0}D_Z-D_{F_0}(D_Z-d_Z)=D_{F_0}d_Z,
\]
so that indeed $d_X=d_{Y^*}$.

\end{proof}

\section{Labelled posets and type functions} \label{app:pf0}

We start by showing some basic properties of the reduced poset $\Pe_f^0$, $f\in \Te_n$.
 By definition, $\Pe_f^0$ is  obtained from $\Pe_f$ by removing any element
that is not minimal and has an empty label, hence it is a
poset whose elements are labelled by subsets in $[n]$. The elements in $\Pe_f^0$ will be
denoted by capital letters $S,T,R,\dots$, but they will not be viewed as subsets of $[n]$.
The label set of $T$ will be denoted by $L_T$. The
order  relation in $\Pe_f^0$ will be denoted as $\le$.

The labelling in $\Pe_f^0$ has some immediate properties: if $S,T\in \Pe_f^0$, then $S\le
T$ implies that $L_S\cap L_T=\emptyset$ and if $S,T\in \mathrm{Min}(\Pe_f^0)$, then
$L_S\subseteq L_T$ implies $S=T$. Also, the smallest element, if present, is the only element in
$\Pe_f^0$ that may have an empty label set.

\begin{lemma}\label{lemma:pf0_mincover} If $\Pe_f^0$ has more than one element,
then any $S\in \mathrm{Min}(\Pe_f^0)$ is covered by at least one element.

\end{lemma}

\begin{proof} We will proceed by induction on $n$. The assertion is clearly true for
chains, so for $n\le 3$. Assume it holds for all $m<n$ and let $f\in \Te_n$. 
 Assume that  $S\in
\mathrm{Min}(\Pe_f)$ is not covered by any element. If $S$ is the least element, then
$\Pe_f^0=\{S\}$. 
Otherwise, 
$\Pe_f^0$ does not contain $\emptyset$ and has no largest element, so that $\emptyset \in \Pe_{f^*}^0$ and
 $f^*$ has no free outputs.

If $f\approx
f_1\otimes f_2$, then $S=(S_1,S_2)$ for $S_i\in \mathrm{Min}(\Pe_{f_i}^0)$ and it is clear
that both $S_1$ and $S_2$ must be coverless. By the induction assumption,
$\Pe_{f_1}^0=\{S_1\}$ and $\Pe_{f_2}^0=\{S_2\}$, so that $\Pe_f^0=\{(S_1,S_2)\}$.

 Assume next that $f^*\approx f_1\otimes f_2$. In this case, both
$\Pe_{f_i}^0$ contain $\emptyset$ and $f_i$ have no free outputs. 
Since $S$ covers $\emptyset$ in $\Pe_{f^*}^0$, we must have the same situation in one of
$\Pe^0_{f_i}$, that is, say in $\Pe_{f_1}^0$, there is some $S_1$ such that $\emptyset\cover
S_1$ and $S=(S_1,\emptyset)$.

Note that $S$ cannot be covered in $\Pe_{f^*}^0$ unless $[n]\in \Pe_{f^*}^0$, since any
other cover in $\Pe_{f^*}^0$ would be also present in $\Pe_f^0$. But it is easy to see
from Lemma \ref{lemma:p0_constr} (ii) that this is impossible. Indeed, assume that
$\Pe_{f^*}^0$ has a largest element $T$, then it must have one of the forms
$T=(T_1,\emptyset)$ or $(\emptyset, T_2)$. But then, say,  $(T_1,\emptyset)=T\ge (\emptyset,
S_2)$ for any $S_2\in \Pe_{f_2}^0$, we conclude that $\Pe_{f_2}=\Pe_{f_2}^0=\{\emptyset\}$, in which
case $f^*$ would have free outputs. 

It follows that $S$ has no cover in $\Pe^0_{f^*}$ and we must have the same situation for
$S_1$ in $\Pe^0_{f_1}$. Again, $\Pe_{f_1}^0$ has no largest element, so that $f_1^*$ has
no free outputs. Since also $f_1$ has no free outputs,
$\Pe^0_{f_1^*}=\Pe^0_{f_1}\setminus\{\emptyset\}$, so that $S_1$ is a minimal element in
$\Pe_{f_1^*}^0$ that has no cover. By the induction assumption, $\Pe^0_{f_1^*}=\{S_1\}$,
 which is not possible with the assumed structure of $\Pe^0_{f_1}$.

\end{proof}

\begin{lemma}\label{lemma:pf0_covermin} Any element $T\in \Pe_f^0$ can cover at most one
minimal element. 

\end{lemma}

\begin{proof} We will proceed by induction on $n$. Since the assertion is trivial for
chains, it holds for $n\le 3$. Assume it is true for $m<n$ and let $f\in \Te_n$. Let
$T$ be an element that covers $T_1,\dots, T_k\in \mathrm{Min}(\Pe^0_f)$, $k>1$, so that we
must have  $\emptyset \notin \Pe^0_f$. 
Then  $T$ is not the largest element, otherwise by Lemma \ref{lemma:p0_basic}, $T$ would
be the largest element in $\Pe_f$ in which case the rank of $f$ would be 1, which is not
possible. It follows that  $T$  and $T_1,\dots,
T_k$ are all contained in $\Pe_{f^*}^0$. 

Assume that $f\approx f_1\otimes f_2$, then any $T_i$ is of the form
$T_i=(T^i_1,T^i_2)$, with $T^i_j\in \mathrm{Min}(\Pe_{f_j}^0)$.  If $i\ne i'$, we may
assume that, say, $T^i_1\ne T^{i'}_1$. Since $T$ covers both $T_i$ and $T_{i'}$, we must
have $T^i_2=T^{i'}_2$ and $T=(S,T^i_2)$ for some $S\in \Pe_{f_1}^0$ such that
$T^i_1,T^{i'}_1\cover S$. By the induction assumption, this is not possible. 

If $f^*\approx f_1\otimes f_2$, then we may assume that there are
$S, S_1,\dots,S_k\in \Pe_{f_1}^0$ such that $\emptyset \cover S_i \cover S$. 
Then $S_1,\dots,S_k$ are minimal elements in  $\Pe^0_{f_1^*}$ which are either uncovered
or all covered by $S$. The former  case is excluded by Lemma \ref{lemma:pf0_mincover} and
the latter by the induction assumption.

\end{proof}

\begin{lemma}\label{lemma:mc_elements}  Any minimal covering element in $\Pe_f^{0}$ is
minimal covering in $\Pe_f$.
\end{lemma}

\begin{proof} Let $S\in \Pe_f^0$ be a minimal covering element. By Lemma
\ref{lemma:pf0_mincover}, $S$ covers exactly one minimal element $U\in
\mathrm{Min}(\Pe_f^0)$. Assume that $S$ is not minimal covering in $\Pe_f$. Since $U\in
\mathrm{Min}(\Pe_f)$ by Lemma \ref{lemma:p0_basic} and $U\le S$ in $\Pe_f$, this can only happen if there is some
$T\in \Pe_f$ such that $U\cover T\le S$. In this case, $T\notin \Pe_f^0$, so that
$L_T=\emptyset$. Since $T=\cup_{T'\le T}L_{T'}$, this implies that there must be some
$U'\ne U$ strictly below $T$ and since  $\Pe_f$ is graded, $U'$ must be a minimal element and
$U'\cover T$. This implies that $U'\le S$, but since $S$ cannot cover $U'$ in $\Pe_f^0$,
there must be, by Lemma \ref{lemma:pf0_covermin}, some minimal covering $S'\in \Pe_f^0$ such that $U'\cover
S'\le S$. In $\Pe_f$, this implies that  $U\le T\le S'\le S$, which implies that $S'=S$ and
hence $S$ covers both $U$ and $U'$, a contradiction.

\end{proof}

We now introduce the following relations. For $f\in \Te_n$ and indices $i,j\in [n]$, we say that $j$ covers
$i$ in $f$, in notation $i<_f j$, if there are some $S,T\in \Pe_f^0$, $S\cover T$, such
that $i\in L_S$ and $j\in L_T$. We say that $j$ strongly covers $i$ in $f$, in notation
$i\cover_f j$,  if for all
$T\in \Pe_f^0$ such that $j\in L_T$, there is some $S\in \Pe_f^0$ with $i\in S\cover T$.

To clarify these notions, consider the following diagram of the reduced poset $\Pe_f^0$
for the product of process matrices:
\begin{center}
\includegraphics[scale=0.9]{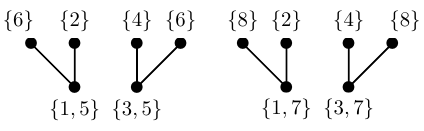}
\end{center}
Here we have $1\cover_f 2$, $5<_f 2$ but $5\not\cover_f 2$,  and $2\not<_f 3$.

We will say that an index $i\in [n]$ is minimal in $f$ if $i\in U$ for some $U\in
\mathrm{Min}(\Pe_f^0)$ and that $j\in [n]$ is minimal covering if $j\in L_T$ for some minimal
covering element $T\in \Pe_f^0$. It is clear from Proposition \ref{prop:pfinput} and
Lemmas \ref{lemma:p0_basic} and \ref{lemma:mc_elements} that  any
element $S\in\Pe_f^0$ is minimal (minimal covering) if its label set contains a minimal
(minimal covering) index.

\begin{lemma}\label{lemma:strong_cover} If $\Pe_f^0$ is not a singleton,  
any minimal index $i$ of $f$ is strongly covered by some minimal covering index $j$.

\end{lemma}

\begin{proof} We will once more proceed by induction. The statement is quite obvious for
chains, so for $n\le 3$. Assume it holds for all $m<n$ and let $f\in \Te_n$. 
Let $i\in U$ for some $U\in \mathrm{Min}(\Pe_f^0)$.

Let $f=f_1\otimes f_2$. Then  $U=(U_1,U_2)$ for $U_1\in \mathrm{Min}(\Pe_{f_1}^0)$,
$U_2\in \mathrm{Min}(\Pe_{f_2}^0)$ and 
we may assume that $i\in U_1$. If $\Pe_{f_1}^0$ is a singleton, then
$\Pe_{f_1}^0=\Pe_{f_1}=\{U_1\}$. It follows that $i$ is a free input of $f$, contained in any
minimal element of $\Pe_f^0$ and  
it is easily seen that 
$i$ is strongly covered in $f$ by any minimal covering index of $f$. 
If $\Pe_{f_1}^0$ is not a singleton, then $i\cover_{f_1} j$ for some minimal covering
index $j$ of $f_1$, by the induction assumption.
Let $T\in \Pe_{f}^0$, then $j\in L_{T}$ if and only if $T=(T_1,U_2)$ for $U_2\in
\mathrm{Min}(\Pe_{f_2}^0)$ and $j\in L_{T_1}$. Since $i\cover_{f_1} j$, there is some
$V_1$ such that $i\in V_1$ and $V_1\cover T_1$, so that $i\in (V_1,U_2)\cover T$.

Let now $f=(f_1\otimes f_2)^*$. Since the statement is trivial if $\emptyset \in \Pe_f^0$
and easily seen to be true for free inputs, we may assume that $f$ has no free inputs
and that $\emptyset \notin \Pe_f^0$, which implies that $f^*$ has no free inputs either.
Then $f\simeq h\vtl \beta$ for a chain $\beta$ and $h$ and $h^*$ have no free indices. 
Since by our assumptions $\Pe_f$ has no least element, we see from Lemmas
\ref{lemma:pf0_mincover} and \ref{lemma:pf0_covermin} that there must be at least two
minimal covering elements, which are mutually incomparable and hence they must all belong
to $\Pe_h^0$, as well as all the minimal elements. If $\beta$ is nontrivial, the statement
follows from the induction assumption. If $\beta$ is trivial, then $f=h$ and by
Proposition \ref{prop:nofree_components}, $\Pe_f^0=\Pe_1\indep \Pe_2$, with
$\Pe_k=\Pe_{f_k}\setminus\{\emptyset\}$, $k=1,2$. A minimal index $i$ of $f$ is contained in the
labels of exactly one of the components, say $\Pe_1$, and  is a label of a minimal element
in $\Pe_1$, hence $i$ is a minimal index of 
$f_1^*$. If $\Pe_{f_1^*}^0$ is a singleton, then $\Pe_1$
is a two-element chain, so that $i$ is covered by all the labels of the top element $T$ in
$\Pe_1$, which are easily seen to be  minimal covering indices  of $f$. Moreover, any
label set $L_S$ of $f$
containing some of these indices must be related to $\Pe_1$ and hence $S=T$. 

If $f_1^*$ is not a singleton, then by the induction assumption, $i\cover_{f_1^*} j$ for some minimal covering
index $j$ of $f_1^*$. By similar arguments as above, it can be seen that $j$ is a minimal
covering index of $f$ and  that $i\cover_f j$. 

\end{proof}

We now proceed towards the proof of Proposition \ref{prop:nofree_product}. 
Let $f\in \Te_n$ be such that $f$ and $f^*$ have no free indices and  $\emptyset\notin \Pe_f$. Assume that
$\Pe_f^0$ has no independent component, so that $f\approx f_1\otimes \dots \otimes f_k$ for some
type functions $f_1,\dots,f_k$ and $[n]=[n_1]\boxplus\dots\boxplus [n_k]$. Assume that this is a finest  decomposition
of this form, so that no $f_l$ is a product.  These assumptions will be kept throughout
this section. Note that for any $l$, $f_l$ cannot have free indices, since
these would be also free indices of $f$.

By Lemma \ref{lemma:p0_constr}, any element in $\Pe_f^0$ can be written as  $T=(T_1,\dots,T_k)$, with all $T_l\in
\mathrm{Min}(\Pe_{f_l}^0)$ except possibly one index $l_0$. Application of a 
permutation  is manifested only on the label sets.

\begin{lemma}\label{lemma:colors} There
is a decomposition $C_1,\dots, C_k$ of $[n]$ and a
bijection $\varphi_l:[n_l]\to C_l$ such that for any $T=(T_1,\dots,T_k)\in
\Pe_f^0$
\[
L_{(T_1,\dots,T_k)}=\begin{dcases} \varphi_l(L_{T_l}) & 
\text{if } \exists \ l, T_l\notin \mathrm{Min}(\Pe_{f_l})\\
\cup_l \varphi_l(L_{T_l}), & \text{otherwise}.
\end{dcases}
\]

\end{lemma}

\begin{proof} 
Let $\sigma\in \permut_n$ be such that $f\circ\sigma=f_1\otimes \dots\otimes f_k$. 
By Corollary \ref{coro:Pf}, 
\[
L_{(T_1,\dots,T_k),f}=\sigma^{-1}(L_{(T_1,\dots,T_k),\otimes_l
f_l})=\begin{dcases}\sigma^{-1}(m_l+L_{T_l}) & 
\text{if } \exists \ l, T_l\notin \mathrm{Min}(\Pe_{f_l})\\
\cup_l \sigma^{-1}(m_l+L_{T_l}), & \text{otherwise},
\end{dcases}
\]
where $m_l=\sum_{i=1}^{l-1}n_i$. Put $C_l=\sigma^{-1}(m_l+[n_l])$ and
$\varphi_l(i)=\sigma^{-1}(m_l+i)$.

\end{proof}

To ease the subsequent notations, we will replace
the labels of $T_{l}\in \Pe_{f_{l}}^0$ with $\varphi_l(L_{T_l})\subseteq C_l$.
Any $i\in C_l$ is thus connected to $f_l$. We will refer to the inclusion of an index $i$ in $C_l$ as
coloring $i$ by a color $l\in\{1,\dots,k\}$.

\begin{lemma}\label{lemma:color_min} Assume that the coloring is known for all minimal and
minimal covering indices of $f$. Then we can reconstruct all
$\Pe_{f_l}^0$ from $\Pe_f^0$.

\end{lemma}

\begin{proof}
Let  $U\in \mathrm{Min}(\Pe_f^0)$,  $U=(Z_1,\dots, Z_k)$ with
  $Z_l\in \mathrm{Min}(\Pe_{f_l})$. By the assumption, we know the coloring of any $i\in
  L_U=\cup_l L_{Z_l}$. Since $L_{Z_l}$ is a label set of $f_l$ if and only if
  $L_{Z_l}\subseteq C_l$, we obtain all label sets of minimal elements in $\Pe_{f_l}^0$.
  
For any   $U=(Z_1,\dots,Z_k)\in \mathrm{Min}(\Pe_f)$, let 
\[
\tilde L_U^l:=L_U\cap
(\cup_{l'\ne l}C_{l'})=\{i\in L_U, i\notin C_l\}=\cup_{l'\ne l}
L_{Z_{l'}},
\]
i.e. $\tilde U^l$ is  the set of all indices in $L_U$ of color
different from $l$.  It follows by the properties of the minimal label sets that 
$\tilde L^l_U\subseteq L_{U'}$ if and only if $\tilde L^l_U=\tilde L_{U'}^l$, so that 
$U'=(Z_1,\dots,Z_l',\dots,Z_k)$ for some $Z'_l\in \mathrm{Min}(\Pe_{f_l}^0)$. 

Fix a minimal element $U$ and consider the subposet in $\Pe_f^0$, given as
\[
\Pe_l:=\{U'\in \mathrm{Min}(\Pe_f^0),\ \tilde L^l_U\subseteq L_{U'}\}^\uparrow.
\]
From Lemma \ref{lemma:p0_constr} (ii) we see that after removing the minimal elements of $\Pe_l$, the poset
decomposes into independent components, one of which corresponds to $\Pe_{f_l}^0$ with
removed minimal elements. This component can be recognized by the labels of minimal
covering elements (which are now minimal elements in the component), colored by $l$. 
To this component, we add back the minimal elements in $\Pe_l$, with the order relations
as in $\Pe_f^0$, but removing the labels that are in $\tilde L_U^l$.

\end{proof}

We now show how to obtain the coloring of labels of minimal and minimal covering elements.

Let $L_1,\dots, L_M$ be all the different label sets for minimal covering elements in
$\Pe_f^0$ and let us denote $\mathbb T_a:=\{T\in \Pe_f^0,\ L_T=L_a\}$, $a=1,\dots,M$. Note that 
all $T\in \mathbb T_a$ are labeled by minimal covering indices and hence  must be minimal covering elements. 
Also, by Lemma \ref{lemma:color_min}, each
such $T$ covers exactly one minimal element $U$. By Lemma \ref{lemma:p0_constr}, we have
$T=(T_1,\dots,T_k)$, where all $T_{l'}$ are minimal in $\Pe_{f_{l'}}$ except a
single index $l$, for which $T_{l}$ is minimal covering in $\Pe_{f_{l}}^0$. Moreover,
$L_a=L_T=L_{T_l}$, in particular, all indices in $L_a$ have the same color $l$.

We introduce the following sets. Let $\mathcal U$ be the set of minimal indices of $f$ and
\begin{align*}
\mathcal V_a&:=\{i\in \mathcal U,\ i\cover_f j \text{ for some } j\in L_a\}\\
\mathcal W_a&:=\{i\in \mathcal U,\ i\not <_f j,\  \forall j\in L_a\}.
\end{align*}
Then all indices in $\mathcal V_a\cup\mathcal W_a\cup L_a$ must have the same color.
Indeed, we have already seen that $L_a\subseteq C_l$ for some $l$. Assume that a minimal
index $i\in C_{l'}$ for $l'\ne l$.  
Since $i$ cannot be a free input of $f_{l'}$, there must be some minimal elements $U_{l'}$
and $U'_{l'}$ in $\Pe_{f_{l'}}^0$ such that $i\in U_{l'}$ and $i\notin U'_{l'}$. Choose any
$j\in L_a$ and let $T_l$ be a minimal covering element in $\Pe_{f_l}^0$ such that $j\in L_{T_l}$.
Let also $U_j\in \mathrm{Min}(\Pe_{f_j}^0)$ for $j\ne l'$. 
 
Put $T=(U_1,\dots, U_{l'},\dots, T_l, \dots U_k)$ and
$T'=(U_1,\dots, U'_{l'},\dots, T_l, \dots U_k)$. Then $j\in L_T=L_{T'}$ and $i\in
(U_1,\dots,U_k)\cover T$, so that $i<_f j$, while $T'$ cannot cover any element containing
$i$, so that $i\not\cover_f j$. Hence
$p\notin \mathcal V_i\cup\mathcal W_i$.

For all $a$, let us denote
$C'_a:=\mathcal V_a\cup \mathcal W_a\cup L_a$ and define $a\sim b$ if 
$C_a'\cap C_b'\ne \emptyset$ or  if
there are some
$T\in \mathbb T_a$ and $S\in \mathbb T_b$ that have a common upper bound in $\Pe_f^0$. Take the transitive closure of this relation (also denoted by $\sim$).
We next prove several claims:

\medskip
\noindent
\textbf{Claim 1.} For  any $i\in \mathcal U$, there is some $a\in [M]$ such that $i\in \mathcal
V_a$, so that all minimal indices are in some set $C'_a$. Indeed, by Lemma \ref{lemma:strong_cover},
any $i\in \mathcal U$ is strongly covered by some minimal covering element $j$, which must
belong to some $L_a$.

\medskip

\noindent
\textbf{Claim 2.} If $a\sim b$, then $L_a,L_b\in C_l$ for some $l$.  To see this, assume
that $a$ and $b$ are such that $C'_{a}\cap C'_{b}\ne \emptyset$. Since $C'_a\subseteq C_l$
if $L_a\subseteq C_l$ and $C_l$ and $C_{l'}$ are disjoint
for $l\ne l'$, this implies that $L_a,L_b$ belong to the same component. Similarly, if
there are some $S\in \mathbb T_{a}$, $T\in \mathbb T_{b}$ with a common upper
bound, then $L_a$ and $L_b$ must be connected to the same function $f_l$. Hence both
relations imply that $L_a,L_b\subseteq C_l$ for some $l$, and this is clearly preserved by
the transitive closure.

\medskip
\noindent
\textbf{Claim 3.} If $\emptyset \notin \Pe_{f_l}$ and  $L_a,L_b\subseteq C_l$,  then
$a\sim b$. 
Assume first that  $\Pe_{f_l}^0$ has a largest element, then we can easily construct
elements $S\in \mathbb T_a$ and $T\in \mathbb T_b$ that have a common upper bound. Hence
$a\sim b$. 
Assume that $\Pe_{f_l}$ has no largest element, then one can check that $f_l$ and $f_l^*$ have no free indices. 
Since the decomposition
$f=f_1\otimes \dots\otimes f_k$ is the  finest decomposition of $f$ as a product, $f_l$
cannot be a product. By Proposition \ref{prop:nofree_components}, we obtain that $\Pe_{f_l}^0$ must have
independent components,  $\Pe_{f_l}^0=\Pe_1\indep \Pe_2$. Note that none of the $\Pe_i$ can be
a singleton by Lemma \ref{lemma:pf0_covermin}, so that either of the components has at
least one  of the label sets $L_{a'}$ associated with it.
 Assume that  $L_a$ belongs to $\Pe_1$, then any
minimal index $i$ of $f_l$  such that $i<_f j$ for $j\in L_a$ must belong to $\Pe_1$. 
 Hence $\mathcal W_a$ contains all minimal indices of $\Pe_2$. In particular, $\mathcal
 V_{a'}\subseteq \mathcal W_a$ for
all $L_{a'}$ in $\Pe_2$, so that $a\sim a'$ for all such $a'$.  If
$L_b$ is connected with $\Pe_2$, we are done, otherwise we have $a\sim a'\sim b$ by
repeating the argument with $a$ replaced by $a'$.

\medskip 
\noindent
\textbf{Claim 4.} If $\emptyset \in \Pe_{f_l}$, and $L_a\subseteq C_l$, then $L_a=C'_a$ (so that $\mathcal
V_a=\mathcal W_a=\emptyset$) and $a\not\sim b$ for any $b$ with $\mathcal V_b$ or
$\mathcal W_b$ nonempty.   
If $L_a$ and $L_b$ are two such label sets,  then they are connected with
the same $f_l$  if and only if they are in the same  independent component of the labelled poset $U^\uparrow\setminus \{U\}$,
here  $U$ is an arbitrary minimal element in $\mathrm{Min}(\Pe_f^0)$.
Indeed, the first statement is clear from the definition of $\mathcal V_a$ and $\mathcal
W_a$, and Claim 2. For the second statement, note that by Lemma \ref{lemma:p0_constr}, $U^\uparrow\setminus \{U\}$
contains $\Pe_{f_l}^0\setminus\{\emptyset\}$ for any $\Pe_{f_l}$ containing $\emptyset$  as one of its independent
components. If $\Pe_{f_l}^0$ has a largest element, then $\Pe_{f_l}^0\setminus\{\emptyset\}$ cannot
have independent components. Otherwise we have
$\Pe_{f_l}^0\setminus\{\emptyset\}=\Pe_{f_l^*}^0$ (since $f_l$ has no free outputs) and since $f_l$ is not a product,
$\Pe_{f_l^*}^0$ cannot have any independent components  by Proposition
\ref{prop:nofree_components}. 

\begin{proof}[Proof of Proposition \ref{prop:nofree_product}] Let $C_a'$, $a=1,\dots,M$ be
as described above. Assume that the equivalence relation $\sim$ has $k'$ equivalence
classes, then pick some of the colors $1,\dots, k'$ for each equivalence class $[a]_\sim$ and use it
for all indices  in  $\cup_{b\in [a]_\sim} C'_b$. Take all $L_a$ such that  $\mathcal
V_a=\mathcal W_a=\emptyset$ that are not related to any $b$ with nonempty $\mathcal V_b$ or
$\mathcal W_b$,  and use the procedure described in Claim 4 to merge some of the equivalence
classes if necessary. Claims 1-4 show that for all label indices for minimal and minimal
covering elements, we obtain a coloring that matches the decomposition of $f$ as a tensor
product $f\approx f_1\otimes\dots\otimes f_k$. Using Lemma \ref{lemma:color_min}, we get
all the labelled posets $\Pe_{f_l}^0$, with labels transformed by the bijections
$\varphi_l$. Applying any bijection $C_l\to [n_j]$ on the label sets, we obtain
$\Pe_{f_l}^0$ up to a permutation of the labels.

\end{proof}

\begin{exm} The following example shows the diagram of $\Pe_f^0$ for a function $f\in
\Te_{18}$:
\begin{center}
\includegraphics[scale=0.9]{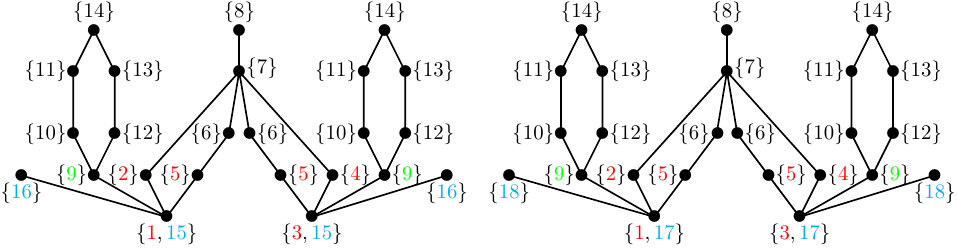}
\end{center}
It is clear that $f$ has no free outputs, since all indices in $[18]$ are contained in the
label sets. Also $f^*$ has no free output, since $\Pe_f^0$ has no largest element. It is
also clear that $f$ and $f^*$ have no free inputs, and $\emptyset\notin \Pe_f^0$.
The poset consists of two direct summands, but the labels are repeating, so there are no
independent components. Hence $f\approx f_1\otimes\dots\otimes f_k$ and the
assumptions of Proposition \ref{prop:nofree_product} are satisfied. We will now show how
to find the functions using the procedure described in the proof of Proposition
\ref{prop:nofree_product}. There are 6 different  labels of the minimal covering elements.
The label sets $L_a$ and the sets
$\mathcal V_a$ and $\mathcal W_a$ are shown in the following table:
\begin{center}
\begin{tabular}{|c|c|c|c|c|c|c|}
\hline
$a$& 1 & 2& 3&4 &5 &6\\
\hline
$L_a$ & $\{2\}$ & $\{4\} $& $\{5\}$ & $\{9\}$ & $\{16\}$ & $\{18\}$\\
\hline
$\mathcal V_a$ & $\{1\}$& $\{3\} $& $\emptyset$ & $ \emptyset$ & $\{15\}$ & $\{17\}$\\
\hline
$\mathcal W_a$& $\{ 3\}$& $\{ 1\} $& $\emptyset$ & $ \emptyset$ & $\{ 17\}$ & $\{15\}$\\
\hline
\end{tabular}

\end{center}
It is immediately clear that $1\sim 2$ and $5\sim 6$, since $C'_1\cap C'_2\ne \emptyset$
and $C'_5\cap C'_6\ne \emptyset$. We also have $1\sim 2\sim 3$, since vertices labeled by $L_1=\{2\}$,
$L_2=\{4\}$ and $L_3=\{5\}$ have a common upper bound (labelled by $\{7\}$). It follows
that we have three color classes: $C_1=\{1,2,3,4,5\}$ (colored red in the diagram),
$C_2=\{15,16,17,18\}$ (colored cyan) and $C_3=\{9\}$ (colored green). Since $C_3$ does not
contain any labels of minimal elements, the corresponding function $f_3$ satisfies
$\emptyset \in \Pe_{f_3}^0$. We now apply Lemma \ref{lemma:color_min} to obtain
$\Pe_{f_1}^0$. Let $U=\{1,15\}$, then $\tilde L^1_U=\{15\}$ and $\Pe_1=\{U'\in
\mathrm{Min}(\Pe_f^0), 15\in L_{U'}\}^\uparrow$ is then the left half of the diagram. 
Continuing similarly for $C_2$ and $C_3$, we obtain the following three factors of
$\Pe_f^0$:
\begin{center}
\begin{minipage}[b]{0.3\textwidth}
\centering
\includegraphics[scale=0.9]{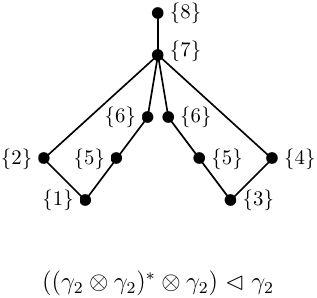}
\end{minipage}
\begin{minipage}[b]{0.3\textwidth}
\centering
\includegraphics[scale=0.9]{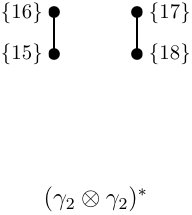}
\end{minipage}
\begin{minipage}[b]{0.3\textwidth}
\centering
\includegraphics[scale=0.9]{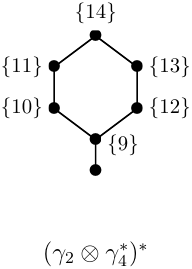}
\end{minipage}
\end{center}

The first factor here corresponds to a more complicated type (see Example
\ref{exm:pf0_causal} (c)), while the second and third
factors are the types of process matrices  (Examples \ref{exm:T4}, \ref{exm:pf0_basic} and  \ref{exm:ns_pm})
and process matrices with global past and future (Examples \ref{exm:pm_past_future} and
\ref{exm:comb_to_comb}).

\end{exm}

\end{document}